\documentclass[letter, 10pt]{article}

\usepackage[margin=1.5in]{geometry}

\usepackage[utf8]{inputenc}
\usepackage{hyperref}
\usepackage[nottoc,numbib]{tocbibind}
\usepackage{amsmath}
\usepackage{bm}
\usepackage{amsfonts}
\usepackage{amssymb}
\usepackage{amsthm}
\usepackage{makecmds}
\usepackage{todonotes}
\usepackage{graphicx}
\usepackage{subcaption}
\usepackage{caption}
\usepackage{mathtools}
\usepackage{centernot}
\usepackage{diagbox}
\usepackage{setspace}
\usepackage{multirow}

\usepackage{pgf}
\usepackage{tikz}
\usetikzlibrary{automata,positioning,arrows}

\numberwithin{equation}{section}
\numberwithin{figure}{section}

\newtheorem{theorem}{Theorem}[section]
\newtheorem{lemma}[theorem]{Lemma}
\newtheorem{prop}[theorem]{Proposition}
\newtheorem{conj}[theorem]{Conjecture}

\newtheorem{corollary}{Corollary}[theorem]

\AfterEndEnvironment{lemma}{\noindent\ignorespaces}
\AfterEndEnvironment{corollary}{\noindent\ignorespaces}

\newtheorem{defn}{Definition}[section]
\AfterEndEnvironment{defn}{\noindent\ignorespaces}
\newtheorem{probenv}{Problem}[section]

\newenvironment{varscope}{}{}

\mathchardef\mhyphen="2D % Define a "math hyphen"

\usepackage{xspace}	% provides some spacing logic
\newcommand*{\eg}{e.g.\@\xspace}
\newcommand*{\ie}{i.e.\@\xspace}

\makeatletter
\newcommand*{\etc}{%
    \@ifnextchar{.}%
        {etc}%
        {etc.\@\xspace}%
}
\newcommand*{\etal}{%
    \@ifnextchar{.}%
        {et al}%
        {et al.\@\xspace}%
}
\newcommand*{\iid}{%
    \@ifnextchar{.}%
        {i.i.d}%
        {i.i.d.\@\xspace}%
}
\makeatother

\makecommand{\nats}{{\mathbb N}}
\makecommand{\ints}{{\mathbb Z}}

\makecommand{\perm}{\sigma}
\makecommand{\varperm}{\varsigma}
\makecommand{\pile}{y}

\makecommand{\phase}{t}
\makecommand{\numphase}{T}

\makecommand{\numpile}{m}
\makecommand{\ordervar}{\psi}

\makecommand{\embed}{\phi}

\makecommand{\qshuffle}{{\rm \qsym {\text -} shuffle}}
\makecommand{\sshuffle}{{\rm \ssym {\text -} shuffle}}

\makecommand{\pilesfn}{{\rm piles}}

\makecommand{\qphase}{{\rm multi {\text -} phase}}
\makecommand{\sphase}{{\rm s {\text -} phases}}
\makecommand{\digit}{{\rm d}}
\makecommand{\Numpiles}{M}
\makecommand{\twodots}{\mathinner {\ldotp \ldotp}}
\makecommand{\schedules}{{\mathbf y}}
\makecommand{\cumread}{{\rm cumread}}

\makecommand{\shuffleOp}{{\rm shuffle}}
\makecommand{\pileType}{\chi}
\makecommand{\pileTypeInd}{\kappa}
\makecommand{\assignfn}{\pile}
\makecommand{\assignFn}{\assignfn}
\makecommand{\pileIndex}{p}
\makecommand{\condReverseFn}{r}
\makecommand{\assignTypeFn}{z}
\makecommand{\assignTypeSum}{Z}

\makecommand{\numPileSeq}{M}
\makecommand{\numPile}{\numpile}
\makecommand{\numPhase}{\numphase}
\makecommand{\numAccum}{\nats}
\makecommand{\reals}{{\mathbb R}}

\makecommand{\reversePerm}{\perm_{\rm J}}

\makecommand{\pileTypeVec}{{\boldsymbol\chi}}
\makecommand{\pileAssignments}{{\mathcal X}}
\makecommand{\indord}{{\rm ord}}

\makecommand{\arityFn}{{\rm arity}}
\makecommand{\xor}{\oplus}

\makecommand{\permsize}{n}
\makecommand{\nclause}{{n_{\rm c}}}
\makecommand{\nvar}{{n_{\rm v}}}
\makecommand{\radix}{r}
\makecommand{\logict}{{\rm T}}
\makecommand{\logicf}{{\rm F}}

\makecommand{\qsym}{{\mathbf q}}
\makecommand{\ssym}{{\mathbf s}}
\makecommand{\arbpile}{{\mathcal P}}

\def\NP/{NP}
\def\NPHard/{NP-Hard}
\def\NPComp/{NP-Complete}
\def\POLY/{P}

\makecommand{\SAT}{{\rm SAT}}

\makecommand{\notimplies}{\centernot\implies}
\makecommand{\prob}{{\mathbb P}}
\makecommand{\expect}{{\mathbb E}}
\makecommand{\normdist}{{\mathcal N}}

\makecommand{\fnDomain}{{\rm Dom}}
\makecommand{\fnRange}{{\rm Rng}}

\makecommand{\dfneq}{\triangleq}

\title{
Sorting by pile shuffles on queue-like and stack-like piles
can be 
hard
}
\author{Kyle B. Treleaven \\
\small \texttt{ktreleav@alum.mit.edu}
}
\date{June 5, 2025}

\begin{document}

\maketitle

\begin{minipage}{0.9\textwidth}
\begin{abstract}

Inspired by a common technique for shuffling a deck of cards on a table without riffling,
we continue the study of a prequel paper
on the pile shuffle and its capabilities as a sorting device.
We study two sort feasibility problems
of general interest
concerning pile shuffle,
first introduced in the prequel.
These problems are characterized by:
(1) bounds on the number of sequential rounds of shuffle, and piles created in each round;
(2) the use of a heterogeneous mixture of queue-like and stack-like piles, 
as when
each round of shuffle may have a combination of face-up and face-down piles;
and
(3) the ability of the dealer to choose the types of piles used during each round of shuffle.
We prove
by a sequence of reductions from the Boolean satisfiability problem (SAT) that
the more general problem is \NPHard/.
We leave as an open question the complexity of its arguably more natural companion,
but discuss avenues for further investigation.
Our analysis leverages
a novel framework,
introduced herein,
which equates instances of shuffle
to members of a particular class of deterministic finite automata.
\end{abstract}

\end{minipage}

\vspace{2em}

\section{Introduction}
\label{sec:intro}

Pile shuffle is a common method for shuffling a deck of cards on a table without riffling,
and is more or less a non-mechanical version
of the \emph{shelf shuffle}~\cite{10.1214/12-AAP884, FULMAN2021112448}
used in shuffling machines in casinos:
A deck of cards is dealt into piles on an empty card table, and
the resulting piles are stacked up in some order to form the new deck.
We typically think of shuffle as a randomly driven process,
which may be repeated until the deck is ``shuffled enough''.

Considerable existing literature
explores the interesting combinatorics and probability theory questions
about the possible outcomes of randomized shuffles, including shelf shuffle~\cite{10.1214/12-AAP884, FULMAN2021112448}.
However,
other questions arise if
we may choose the actions within a shuffle deliberately to achieve a desired outcome.
In this paper
we continue a study
from a companion prequel paper~\cite{treleaven2025sortingpermutationspileshuffle},
of the capabilities
of pile shuffle
as a device
for sorting decks of cards.
We note however that sort pertains to randomization
in that
%by recognizing that
good mixing might be achieved most directly
by combining a sorting device with a randomized labeler.

%Pile shuffle is limited as a sorting device in that
%not every deck of cards can be sorted in a bounded number of ``rounds'' of shuffle on a bounded number of piles.
%
We study two sort feasibility problems
of general interest in pile shuffle,
first introduced in~\cite{treleaven2025sortingpermutationspileshuffle}.
Our study
%is novel in that we consider
considers a relatively general version of
pile shuffle,
allowing a heterogeneous mixture
of queue-like and stack-like piles, 
as when
each round of shuffle may use a combination of face-up and face-down piles:
\begin{probenv}[Repeated-round Dealer's choice pile shuffle sort]
\label{prob:repeated-round}
%The \emph{repeated-round heterogeneous pile shuffle sort problem} is
Can a given deck of cards be sorted into a predefined order
using $\numphase \geq 0$ rounds of pile shuffle,
with no more than $\numpile \geq 1$ piles created in each round?
The dealer may choose the \emph{type} of each pile (queue or stack) arbitrarily in each round of shuffle.
\end{probenv}
\begin{probenv}[Variable-round Dealer's choice pile shuffle sort]
\label{prob:variable-round}
%The \emph{variable-round heterogeneous pile shuffle sort problem} is,
Given per-round pile capacities $\left( \numpile_1, \ldots, \numpile_\numphase \right)$,
can a given deck of cards be sorted into a predefined order
using $\numphase$ rounds of pile shuffle,
where
$\numpile_\phase$ is the maximum number of piles allowed in round $\phase$,
$1 \leq \phase \leq \numphase$?
Again, the dealer may choose the type of each pile arbitrarily in each round of shuffle.
\end{probenv}

Problem~\ref{prob:repeated-round} represents a more or less organic scenario,
of shuffle in bounded time in physical piles, on a surface
whose shape does not change from round to round.
Problem~\ref{prob:variable-round} generalizes Problem~\ref{prob:repeated-round}
to a degree,
%---the former contains all the instances of the latter---%
in that the pile capacity can change from round to round.

In this paper,
we demonstrate
by a sequence of novel reductions from the Boolean satisfiability problem (SAT) that
at least one of those variants is \NPHard/.
In particular,
we prove the main result of the paper:
\begin{theorem}
\label{thm:variable-round-hard}
The variable-round Dealer's choice pile shuffle sort problem (Prob.~\ref{prob:variable-round}) is \NPHard/.
\end{theorem}
Unfortunately our study leaves as an open question whether repeated-round sort feasibility is \NPHard/.
However, the author's conjecture, which we discuss
at the end of the paper,
is that it is.
\begin{conj}
\label{conj:repeated-hard}
The repeated-round Dealer's choice pile shuffle sort problem (Prob.~\ref{prob:repeated-round}) is \NPHard/.
\end{conj}

The paper leverages
a novel framework,
introduced herein,
which equates instances of shuffle
to members of a particular class of deterministic automata.

\subsection{Literature Review}

Sorting is one of the fundamental problems in computer science~\cite{Knuth:art3},
with deep connections to combinatorics and computational complexity theory.
%%
%It has attracted significant research since the beginning of computing,
%no doubt because the theoretical properties of sorting algorithms have direct and significant impact
%on the performance of large-scale computer systems.
%
The most popular branches of the literature
tend to focus on the typical random-access model of computer memory,
where data at a particular memory address may be mutated at unit cost.
In such setting the most common objectives are to obtain space- and time- efficient sort algorithms,
optimizing resource consumption.

Sorting problems based on the physical constraints of storage present other unique challenges, and
are also of both practical and theoretical significance.
In physically-motivated settings like pile shuffle,
even asking whether a collection of items can be sorted at all (feasibility)
can become an interesting question.

\emph{Patience sort.}
%While pile shuffle and shelf shuffle~\cite{FULMAN2021112448} use extremely similar mechanisms,
Our study of pile shuffle sort
is arguably closest to that of Patience sort~\cite{chandramouli2014patience, burstein2006combinatorics},
in that it is similarly motivated by sorting in piles, rather than randomizing.
%them.
%
Patience sort
represents a greedy strategy for the so-called Floyd's Game, and
happens to be an efficient means to compute longest increasing subsequences.
Both Patience sort and sort by pile shuffle are so-called \emph{distribution sorts}~\cite[p.~168]{Knuth:art3}, in that
the sort occurs in a distribution phase, followed by a collection phase.
However, Patience sort is a kind of \emph{merge sort}---%
the piles are merged together into the output one card at a time---%
whereas pile shuffle sort is a so-called \emph{bucket sort}---%
the piles are concatenated, each as a whole, in a chosen order.
%
%\todo{kill?}{Indeed, one might call ours ``Impatience'' sort.}
%

While bucket sorts classically employ some kind of sub-sort within each bucket (sometimes recursively)
that is not always physically practical.
Instead, pile shuffle sort is a non- sub-sorted, a.k.a. \emph{pure}, bucket sort.
In lieu of sub-sorting,
we rely on
multiple rounds of shuffle
to augment the power of a bounded number of piles,
at the expense of more time to execute the shuffle.
In the motivating case of sort with a physical facility (\eg, a table of a specific size)
that is often a necessary trade-off.

\emph{Sorting networks.}
%\subsubsection{Sorting on networks of stacks and queues}
There is a significant body of literature about the sortable permutations of so-called sorting \emph{networks}~%
\cite{bona2002survey, adda7f6a9afa4f559318c219c37a8dfe, halperin_complexity_2008}.
%
%Indeed Patience sort can be cast within this framework~\cite{tarjan}.
While the study was founded originally on networks of queues and stacks, 
a wide variety of more complex shuffle devices---%
including ones inspired by card shuffling---%
have been studied~\cite{pudwell_sorting_2023,dimitrov_sorting_2022}.
Although bucket sorts like pile shuffle might also be cast in the network framework,
the fit is not entirely natural, and
the author is not aware of previous such treatments.

\emph{Pancake sort.}
The so-called ``Pancake'' sort~\cite{GATES197947} is another physically inspired problem,
of sorting a disordered stack of pancakes by repeatedly inserting a spatula at some point in the stack and flipping all pancakes above it.
Pancake sort appears in applications in parallel processor networks, and can provide an effective routing algorithm between processors~%
\cite{GARGANO1993315,4032188}.
Pancake sort has also been called an ``educational device'', and 
it was shown to be \NPHard/ in~\cite{BULTEAU20151556} by a reduction from $3$SAT.
In this paper
we prove
by reduction from SAT
that
at least some variants of Dealer's choice pile shuffle sort are \NPHard/.

\subsection{Previous Work}

In the prequel paper~\cite{treleaven2025sortingpermutationspileshuffle}
we formulated a mathematical model of pile shuffle, and
presented necessary and sufficient conditions for a single shuffle to sort an input permutation.
From those conditions we derived efficiently computable, linear-time formulas for:
(1) the minimum number of piles required for sortation
given specified pile types,
and
(2) transcription of a sorting shuffle on the minimum number of piles (a \emph{minimal} sort) in the same scenario.
%
%We showed that,
We confirmed that, in the homogeneous all-queues or all-stacks cases,
the number of piles required is in terms
of well-studied ascent and descent permutation statistics~\cite{bona_combinatorics_2004, butler_stirling-euler-mahonian_2023}
with deep connections to combinatorics.

The analysis of this paper is based
on a mathematical framework,
also introduced in~\cite{treleaven2025sortingpermutationspileshuffle},
by which
we may interpret the result
of a sequence of shuffles on given pile types
as an equivalent single-round shuffle on fixed-type ``virtual piles''.
The framework generalized many of our findings in~\cite{treleaven2025sortingpermutationspileshuffle} to the multi-round setting, and
it confirmed that
%repetition augments the power of pile shuffle, confirming that
$m$ piles over $T$ rounds has the capacity of $m^T$ piles in a single round.
The present paper is motivated by our discovery, also
%We also discovered
in~\cite{treleaven2025sortingpermutationspileshuffle},
that
if the dealer is allowed to choose the types of the piles arbitrarily
during some of multiple rounds of shuffle,
then deciding sort feasibility can be non-trivial.
In contrast,
we found that
sort feasibility remains tractable despite dealer choice during a single round of shuffle.

\subsection{Organization}

The rest of the paper is organized as follows:
In Section~\ref{sec:background} we define notation and summarize the necessary background for the paper.
In Section~\ref{sec:problem}
we reintroduce the mathematical model of pile shuffle from~\cite{treleaven2025sortingpermutationspileshuffle},
and state the objectives of the present paper.
In Section~\ref{sec:prev-results} we summarize
prerequisite results about sort with pile shuffle, again from~\cite{treleaven2025sortingpermutationspileshuffle}.
In Sections~\ref{sec:algebra} and~\ref{sec:chains} we build upon those results
to develop a mathematical framework central to the proof strategy of our main result:
We demonstrate a correspondence between instances of shuffle and members of a particular class of finite automata.
(Section~\ref{sec:algebra} is preliminary to the framework introduced in Section~\ref{sec:chains}.)
In Section~\ref{sec:proof-strategy} we outline the proof strategy for the remainder of the paper,
presenting a sequence of three (3) sort feasibility problems to which we reduce $\SAT$ subsequently.
We present the gadgets common to all proofs in Section~\ref{sec:gadgets},
and complete our three reductions in Sections~\ref{sec:Q1}, \ref{sec:Q2},
and~\ref{sec:variable-round-feasibility-complexity}.
We discuss potential strategies for strengthening our results to prove Conj.~\ref{conj:repeated-hard} in Section~\ref{sec:on-conj}.
Finally, we summarize our results and offer conclusions in Section~\ref{sec:conclusion}.

\section{Background}
\label{sec:background}

We summarize most of the necessary mathematical background in this section, and
introduce the notation used throughout the paper.
However, our discussion will assume familiarity with foundational computational complexity theory---%
including strings and languages, and deterministic finite automata (DFA)---%
which we do not cover here;
also, ideally, with decision problems and the \POLY/ and \NP/ problem classes,
although these are briefly reviewed in Appendix~\ref{sec:complexity}.
These topics are well covered in~\cite{sipserIntroductionTheoryComputation2021},
with strings, languages, and DFA covered as mathematical preliminaries.

\subsection{Notation}

We use the following notation throughout the paper.

\emph{Ranges.}
We will denote by $\nats$ the natural numbers starting from unity ($1$), and by $[n]$ the natural range $\{ 1, 2, \ldots, n \}$.
We will denote
by $[n] \pm k$ the shifted range $\{ i \pm k : i \in [n] \}$;
in particular, $[n] - 1 = \{ 0, 1, \ldots, n-1 \}$ is often useful.
We will denote by $\nats_0$ the naturals starting from zero instead of unity, \ie, the whole numbers.

\emph{Functions.}
We denote the domain of a function $f$ by $\fnDomain(f)$.
Given a set $A \subseteq \fnDomain(f)$, we may abuse notation and 
let $f(A)$ denote the image of the set $\{ f(x) : \ x \in A \}$.
Additionally, we use the notation $f(g) = h$
for the composition of functions 
$f: Y \to Z$ and $g: X \to Y$,
\ie, $h(x) = f(g(x))$ for all $x \in X$.

\emph{Indicator expressions.}
In some equations we use a so-called indicator notation, where
$[P]$ denotes the indicator function associated with a proposition $P$.
That is, $[P] = 1$ if the proposition is true, otherwise $[P] = 0$.
For example $f(x) = [ x \geq 3 ]$ means that $f(x) = 1$ over $x \geq 3$ and zero elsewhere.

%\emph{Modular arithmetic.}
%We use 2-modulo arithmetic sparingly, with notation $a \xor b \dfneq (a + b) \mod 2$.

\emph{Sequences and strings.}
We will denote the set of all sequences of length $n$ on a set $S$ by $S^{[n]}$, or more lazily by $S^n$.
In this paper we frequently denote sequences in string notation.
For example,
we may write $W = w_1 w_2 w_3$ to specify a sequence of three elements,
with $W(1) = w_1$, $W(2) = w_1$, and $W(3) = w_3$.
As much as possible
we will use upper case letters to denote strings and lower case letters to denote individual symbols.
We will denote the length of a string $W$ by $|W|$; in our example $|W| = 3$.

We will write $AB$ for the concatenation of $A$ and $B$, where
$A$ and $B$ may be any combination of symbols, strings, or languages.
As is fairly common,
we use superscripting in string form to denote repetition, and parentheses for grouping.
For example, we would write out $A^2(AB)^2b^3$ as $AAABABbbb$.
%if $ABC$ is the partition of a string into three segments, then $AB^3C = ABBBC$.
%
As usual we write $A^*$ for the \emph{Kleene closure}, or arbitrary repetition of $A$, which is always a language.
We denote the set of all strings on $S$ by $S^*$.

\emph{Permutations.}
\makecommand{\Perm}[1]{{[#1]!}}
A permutation of a finite set $X$ is a bijective mapping $\perm: X \to X$.
It is well known that there are $n!$ (factorial) permutations of a set of size $n$.
We will denote by $X!$ the set of all such permutations.
In particular,
we will denote by $[n]!$ the set of all permutations of the standard range $[n]$.
Abusing notation slightly,
we will denote the set of permutations of a standard range of \emph{any} length
(including zero)
by $[\nats_0]! \dfneq \cup_{n \geq 0} \, [n]!$.

The identity permutation
%on a domain $X$
is the special permutation $\perm_I(x) = x$ for all $x \in X$.
A composition of two permutations is also a permutation.
The inverse of a permutation $\perm$, denoted $\perm^{-1}$, is the permutation with the property
$\perm^{-1}(\perm) = \perm(\perm^{-1}) = \perm_I$.

\makecommand{\questionSet}{{\mathbb Q}}
\makecommand{\probVar}{\questionSet}
\makecommand{\questionVar}{q}
%\makecommand{\answerFn}{a}
\makecommand{\tsym}{\top}
\makecommand{\fsym}{\bot}

\makecommand{\problemVar}{\probVar}

\emph{Decision problems.}
We will denote a decision problem by 
the set $\questionSet$ of all its questions, with
an implicit subset $\questionSet^+ \subseteq \questionSet$
of those which are in the affirmative.

\makecommand{\assignvec}{{\mathbf x}}

\subsection{Boolean satisfiability ($\SAT$)}
The Boolean satisfiability problem ($\SAT$) is:
\begin{defn}[Boolean satisfiability]
Given a propositional logic formula in conjunctive normal form (CNF),
%---aka, AND-of-ORs---%
does a satisfying assignment of its variables exist?
\end{defn}
A CNF formula $\phi$ is
a set of clauses $\{ \phi_j \}_{j=1}^m$
on variables $\{ x_i \}_{i=1}^n$,
where each clause is defined by a subset of the literals:
The literals include
the positive literals $\{ x_i \}_{i=1}^n$ and the negative literals $\{ \neg x_i \}_{i=1}^n$.
A satisfying assignment (or \emph{solution}) to $\phi$ is a truth assignment---%
of either $\top$ (true) or $\bot$ (false) to each variable of the formula---%
that makes every clause in the formula true.
A clause $\phi_j$ is true
if $x_i = \top$ for some $x_i \in \phi_j$, or $x_i = \bot$ for some $\neg x_i \in \phi_j$.
%(Note the empty clause, containing no literals, cannot be true.)

%
In the paper we denote the set of all solutions of a formula $\phi$ by $\SAT(\phi)$.
Similarly, we denote by $\SAT(\phi_j)$ the set of solutions of a single clause $\phi_j$.
Therefore, $\SAT(\phi) = \cap_j \SAT(\phi_j)$.

$\SAT$ is famously \NPComp/.

%\clearpage
\section{Problem Statement}
\label{sec:problem}

In this section
we reintroduce the mathematical model of pile shuffle
presented in~\cite{treleaven2025sortingpermutationspileshuffle} and used throughout our study.
Then,
we state the objective of the paper.

Pile shuffle begins
with a deck of cards and an empty card table, and 
has two phases---distribution (or ``the deal''), and collection:
During the deal, until the deck is empty,
we place the next card from the top of the deck onto the table,
either directly on the table, creating a new \emph{pile}, or
on top of another pile previously created.
During the collection phase, 
after the deck has been fully dealt out,
we pick up each pile as a whole,
one at a time in some order,
and add it to the bottom of the new deck.

\emph{Deck representation:}
We assume a fixed assignment of labels from $[n]$ to $n$ distinct elements of a deck, so that
any deck ordering can be represented by a permutation $\perm \in [n]!$.
In particular, we represent the order of a given deck
by the permutation with the property that 
label $s$ is in the $\perm(s)$-th position
for every $s \in [n]$.
We call this the \emph{embedding} convention,
which is a departure from a perhaps more typical---and inverse---\emph{sequence} convention, where
$\perm(k)$ would be the label in the $k$-th position.
We use embedding notation throughout the study because of its property that
$s$ precedes $t$ in the deck if and only if $\perm(s) < \perm(t)$.
We say a deck is \emph{sorted} if it is represented by the identity permutation $\perm_I$.

\emph{Pile types:}
We consider piles of two types in this paper: queues and stacks.
A queue maintains the order of cards placed on it.
For example,
if we deal the sequence 1234 into a queue, and pick it back up, we obtain 1234 again.
On a card table, flipping the cards over during the deal creates this kind of behavior.
If label $s$ precedes $t$ in the input deck, and they are placed into the same queue together, then
$s$ precedes $t$ in the new deck also.
In contrast, a stack reverses the order of placement:
If $s$ precedes $t$ in the deck and they are placed into the same stack, then
$t$ precedes $s$ in the new deck;
if we deal 1234 into a stack, we obtain 4321 back.
Dealing cards into a pile \emph{without} flipping them over creates stack-like behavior.
(The cards in a single pile, or deck, must all face the same way.)

\emph{Pile assignment:}
The outcome of a pile shuffle depends on
(1) the input permutation, given,
(2) the pile types used,
(3) the assignment of cards onto piles during the deal, and
(4) the order 
in which the piles are retrieved
during collection.
If we number the piles in the order that they are picked up,
then the assignment of labels to piles can be described
by a function $\pile: [n] \to \nats$ of \emph{pile assignments},
assigning each label $s \in [n]$
to the $\pile(s)$-th pile collected.
Pile shuffle ensures that
the new deck $\varperm$ obeys
$\pile(s) < \pile(t) \implies \varperm(s) < \varperm(t)$,
for every pair of labels $s$ and $t$;
this is true regardless of the pile type(s) that are used.

\emph{Example:}
Suppose we shuffle a deck with label sequence $\perm^{-1} = 456123$ using pile assignments $\pile = 421242$ on stacks.
We can imagine dealing from a face-up deck of cards,
and using $\pile$ to determine
which pile to place
each card
%the card with the showing face 
into.
Note that $\pile$
is a function of
a card's label only, and ignores its position in the deck.
First we place
item $\perm^{-1}(1) = 4$ into pile $\pile(4) = 2$, then
item $\perm^{-1}(2) = 5$ into pile $\pile(5) = 4$, and so on.
After the deal we will see piles as below.
\[
\left.
\begin{array}{cccc}
  & 2 & & \\
  & 6 & & 1 \\
3 & 4 & & 5 \\
\hline
P1 & P2 & P3 & P4
\end{array}
\right.
\]
Note the number of non-empty piles used during shuffle is equal to the number of \emph{distinct} pile assignments,
in this case three.
Collecting the piles in increasing order (left-to-right) we obtain the new deck $\varperm$;
in this case $\varperm^{-1} = 326415$.
If the piles were queues rather than stacks, then we would obtain $\varperm^{-1} = 346251$.

\emph{Objective:}
The objective of this paper is
to
construct a proof of our main result,
Theorem~\ref{thm:variable-round-hard},
demonstrating that
the variable-round Dealer's choice pile shuffle sort problem (Prob.~\ref{prob:variable-round}) is \NPHard/.
We accomplish it
by a sequence of novel reductions from the Boolean satisfiability problem (SAT)
to related problems.
Our study leaves open whether repeated-round sort feasibility is \NPHard/.
However, we support the author's conjecture (Conj.~\ref{conj:repeated-hard})
by suggesting a proof strategy for future work.

\section{Previous Results}
\label{sec:prev-results}

In this section
we summarize the foundational results about the pile shuffle derived in~\cite{treleaven2025sortingpermutationspileshuffle}.
We focus on the general case of shuffling with a heterogeneous mixture of queue-like and stack-like piles,
with the homogeneous queue-only and stack-only shuffles as special cases.

\subsection{Sorting with a (single) pile shuffle}

\makecommand{\pileTypeSingle}{x}
\makecommand{\newPileFn}{{\rm sep}}

\makecommand{\numQueue}{q}
\makecommand{\numStack}{s}

\makecommand{\updateQuota}{f}
\makecommand{\updateQueues}{{\updateQuota_{\rm Q}}}
\makecommand{\updateStacks}{{\updateQuota_{\rm S}}}

\makecommand{\dpstate}{{\mathbf x}}

% :shrug:
\makecommand{\sortScenario}{\pileAssignments}

\makecommand{\descfn}{{\rm desc}}
\makecommand{\readings}{{\rm read}}
\makecommand{\ascruns}{{\rm ascrun}}
\makecommand{\ascsFn}{{\rm ascs}}
\makecommand{\descrunFn}{{\rm descrun}}

A single round of pile shuffle can be captured 
by a parametrized relation
between input and output permutations which
we denote by
\[
\shuffleOp \, \pileType \, \assignFn \, \perm = \varperm
.
\]
In this relation,
$\varperm \in \Perm{n}$ represents the result of
shuffling a deck $\perm \in \Perm{n}$ in a single round, using
pile assignments $\pile$
(previously described)
with
%an assignment of types to piles given by 
the pile \emph{types} given by $\pileType$.
A sequence $\pileType$ on the alphabet
%$$\arbpile \dfneq \{ \text{($\qsym$)ueue}, \, \text{($\ssym$)tack} \}$$
$$\arbpile \dfneq \{ \text{$\qsym$[ueue]}, \, \text{$\ssym$[tack]} \}$$
defines an assignment of types to piles
by the understanding that
for all $\pile \in \fnDomain(\pileType)$,
\[
\pileType(\pile) = \begin{cases}
\qsym, & \textrm{if pile $\pile$ is a queue,} \\
\ssym, & \textrm{if pile $\pile$ is a stack}
.
\end{cases}
\]
A shuffle on only queues (stacks) is the scenario $\pileType \in \qsym^*$ ($\pileType \in \ssym^*$).

\subsubsection{Results}

\makecommand{\itOne}{i}
\makecommand{\itTwo}{j}

A given execution of pile shuffle---%
defined by $(\pileType, \pile)$---%
is a sort if it produces as output the identity permutation,
$\varperm = \perm_I$.
We present the condition for a heterogeneous shuffle to sort its input permutation:
\begin{lemma}[Heterogeneous sort]
\label{lemma:hetero-sort}
%Let $\prec_\qsym \, = \, <$ and $\prec_\ssym \, = \, >$; \ie,
Let $\prec_\qsym$ be the binary relation defined by 
$\itOne\prec_\qsym\itTwo \iff \itOne<\itTwo$, and
let $\prec_\ssym$ be defined by 
$\itOne\prec_\ssym\itTwo \iff \itOne>\itTwo$.
A shuffle $(\pileType, \pile)$ sorts permutation $\perm\in\Perm{n}$
($\shuffleOp \, \pileType \, \assignFn \, \perm = \perm_I$)
if and only if
\begin{equation}
\label{eq:hetero-sort}
\pile(s+1)
\geq
\pile(s) + \left[
	\perm(s+1) \prec_{ \pileType(\pile(s)) } \perm(s)
\right]
\qquad \forall s \in [n-1]
.
\end{equation}
%for all $s \in [n-1]$.
\end{lemma}
The condition~\eqref{eq:hetero-sort} can be checked efficiently
in a single scan of
the permutation.
%the permutation representing the starting deck state.
%
An intuition behind it is as follows:
A sort may never assign element $(s+1)$ to a lower-valued pile than $s$, \ie,
a pile collected earlier in the collection phase.
However, as long as either
$\pile(s)$ is a queue and $s$ precedes $(s+1)$ during the deal,
or else
$\pile(s)$ is a stack and $(s+1)$ precedes $s$ during the deal,
then $(s+1)$ may be placed in any pile $\pile(s+1) \geq \pile(s)$.
Otherwise---%
mathematically, whenever
$\perm(s+1) \prec_{ \pileType(\pile(s)) } \perm(s)$---%
then $(s+1)$
must be placed into a strictly higher-valued pile,
%so that it is 
collected later; \ie,
$\pile(s+1) \geq \pile(s) + 1$.

As an example, 
when dealing the sequence $1423$,
$2$ could be dealt into the same queue as $1$, since it is dealt after $1$,
but
$3$ cannot be dealt into the same queue as $4$, since $4$ precedes it;
ultimately, $4$ must be found in a higher-valued pile than $3$, collected later.
Conversely,
$2$ cannot be dealt into the same stack as $1$, whereas
$3$ could be dealt into the same stack as $4$.

A natural question is: what is the minimum number of piles needed to sort a given permutation?
Pile shuffle offers two variables for optimization,
the pile assignments $\pile$ and the pile types $\pileType$.
While they might typically be chosen together,
if the sequence $\pileType$ were fixed a priori,
then
squeezing~\eqref{eq:hetero-sort} yields a recurrence expression
of the minimum number of piles needed to sort the input permutation.
\begin{lemma}
\label{lemma:hetero-func-bound}
Suppose
$(\pileType, \pile)$ is a sort
of a permutation $\perm \in [n]!$,
\ie, it satisfies~\eqref{eq:hetero-sort},
starting
%w.l.g.
without loss of generality
from $\pile(1) \geq 1$.
Let 
$\pile^*$ be defined
(in terms of $\perm$ and $\pileType$ only)
by
\begin{equation}
\label{eq:min-piles-mixed}
\begin{cases}
\pile^*(1) = 1
\\
\pile^*(s+1) = \pile^*(s) + \left[
	\perm(s+1) \prec_{ \pileType(\pile^*(s)) } \perm(s)
\right]	& s \in [n-1]
.
\end{cases}
\end{equation}
Then 
%$(\pileType, \pile^*)$
$\pile^*$
is a minimal sort of permutation $\perm$
on $\pileType$,
\ie,
%in that 
%the number of piles
$\pile^*(n) \leq \pile(n)$
for any original sort
$\pile$
on $\pileType$.
\end{lemma}

For the homogeneous all-queues or all-stacks cases,
\cite{treleaven2025sortingpermutationspileshuffle}
presented versions of Lemma~\ref{lemma:hetero-func-bound}
in terms of well-studied ascent and descent permutation statistics.

Lemma~\ref{lemma:hetero-func-bound}
is fundamental to the investigations of this paper.
It demonstrates that
the types assignment $\pileType$---%
of some sort $(\pileType, \pile)$---%
acts
as a certificate/proof
that sort 
of a permutation $\perm$
is feasible on a given set of piles:
If a sort with $\pileType$ exists,
then
\eqref{eq:min-piles-mixed} recovers one efficiently in a linear scan.
Conversely,
if~\eqref{eq:min-piles-mixed} has no solution with $\pileType$---\ie, if $\pileType$ is a finite sequence and too short---then no other \eqref{eq:hetero-sort} sort with $\pileType$ may exist either.
%This shifts much of our attention in the sequel from pile assignments to pile \emph{type} assignments.

%
\begin{defn}
\label{defn:sort-with-types}
We say a pile types assignment $\pileType$ sorts a permutation $\perm$ if
there exists $\pile$ such that $(\pileType, \pile)$ sorts $\perm$.
\end{defn}

%\subsubsection{Dealer's choice pile shuffle sort}
%\label{subsec:dealer-choice-single}

%
If the dealer is free
to choose the type of each of $\numpile$ piles arbitrarily during the deal,
then
a given deck may be sorted if and only if there exists a type assignment $\pileType \in \arbpile^\numpile$
which sorts it.
%sort in the sense of Definition~\ref{defn:sort-with-types}.
%%
While the search space is exponential in the number of piles ($2^\numpile$ possible assignments),
\cite{treleaven2025sortingpermutationspileshuffle} showed
that
a minimal sort can be obtained, if one exists,
in time that is linear in the permutation length,
by
%a two-equation recurrence
combining~\eqref{eq:min-piles-mixed} with
a greedy strategy for choosing a minimizing types assignment $\pileType^*$.

\subsubsection{Demonstration}

In this section
we offer a demonstration of sorting with pile shuffle as guided by Lemma~\ref{lemma:hetero-func-bound}.
We start by writing an example permutation in the so-called two-line notation,
a two-row matrix where the permutation pre-image is enumerated across the top row, and
the matching image is written underneath it:
\[
\left\{\begin{array}{cccccccc}
1 & 2 & 3 & 4 & 5 & 6 & 7 & 8
\\
4 & 8 & 7 & 5 & 3 & 1 & 2 & 6
\end{array}\right\}
=
\left\{\begin{array}{ccc}
\ldots & s & \ldots
\\
\ldots & \perm(s) & \ldots
\end{array}\right\}
.
\]

We use the set-notation brackets rather than the typical parentheses to emphasize that the column order does not matter.
However, with the first row in the normal ascending order,
the action of~\eqref{eq:min-piles-mixed} can be shown 
by adding $\pile^*$ and $\pileType(\pile^*)$ as additional rows to the matrix;
in this example,
we let $\pileType = \qsym\ssym\qsym$, \ie,
a queue, then a stack, then another queue:
\[
\left\{\begin{array}{cccccccc}
1 & 2 & 3 & 4 & 5 & 6 & 7 & 8 \\
4 & 8 & 7 & 5 & 3 & 1 & 2 & 6 \\
1 & 1 & 2 & 2 & 2 & 2 & 3 & 3
\\
\qsym & \qsym & \ssym & \ssym & \ssym & \ssym & \qsym & \qsym
%\\
%0 & 1 & 0 & 0 & 0 & 1 & 0 & 0
\end{array}\right\}
=
\left\{\begin{array}{ccc}
\ldots & s & \ldots
\\
\ldots & \perm(s) & \ldots
\\
\ldots & \pile^*(s) & \ldots
\\
\ldots & \pileType(\pile^*(s)) & \ldots
%\\
%\ldots & \left[ \perm(s+1) \prec_{\pileType(\pile^*(s))} \perm(s) \right] & \ldots
\end{array}\right\}
.
\]
We can see that $\pile^*$ increments specifically at permutation \emph{descents} from queue-like piles,
or at permutation \emph{ascents} from stack-like piles.

Now if we rearrange the columns so that the second row appears in ascending order, then
the sequence representation of the deck appears in the top row,
\[
\left\{\begin{array}{cccccccc}
6 & 7 & 5 & 1 & 4 & 8 & 3 & 2 \\
1 & 2 & 3 & 4 & 5 & 6 & 7 & 8 \\
2 & 3 & 2 & 1 & 2 & 3 & 2 & 1 
\end{array}\right\}
=
\left\{\begin{array}{ccc}
\ldots & s & \ldots
\\
\ldots & \perm(s) & \ldots
\\
\ldots & \pile^*(s) & \ldots
\end{array}\right\}
.
\]
This column order supports the creation of a \emph{shuffle tableau}:
In the next tableau,
each label $s$ in a row $\pile^*(s)$ and column $\perm(s)$ indicates that
element $s$ is placed into pile $\pile^*(s)$ as the $\perm(s)$-th placement of the deal.
\[
\left.
\begin{tabular}{r|cccccccc}
\diagbox{$\pile^*(s)$}{$\perm(s)$}
	& 1 	& 2 	& 3 	& 4 	& 5 	& 6 	& 7 	& 8	\\
\hline
1	&	&	& 	& 1	&	&	& 	& 2	\\
2	& 6	& 	& 5	& 	& 4	& 	& 3	& 	\\
3	& 	& 7	&	&	& 	& 8	&	& 
\end{tabular}
\right.
\]

Finally,
the piles are collected in row order, top-to-bottom.
Queue-like piles are collected left-to-right, whereas stack-like piles are collected right-to-left.
Collecting the contents of the tableau in this way demonstrates that
$\pile^*$ indeed sorts the input permutation on $\pileType = \qsym\ssym\qsym$.

\subsection{Sorting in multiple sequential pile shuffles}

\makecommand{\typeSchedules}{{X}}

\makecommand{\vrpile}{{\hat\pile}}
\makecommand{\vrPileType}{{\hat\pileType}}
\makecommand{\vrPileTypeFn}{{\hat\pileTypeFn}}
%\makecommand{\vrPileTypeFn}{{\bar\pileTypeFn}}

\makecommand{\designpile}{{\tilde\pile}}
%\makecommand{\reverseFn}{{\rm reverse}}
\makecommand{\reverseFn}{{\rm rev}}

\makecommand{\pileTypeFn}{\pileTypeInd}
\makecommand{\designPileTypeFn}{{\tilde\pileTypeFn}}

\makecommand{\vrNumPile}{{\hat\numpile}}

Next
we present the key results about sorting in multiple sequential rounds of shuffle, which
augments the power of a bounded number of piles
at the expense of more time to execute the sort.
%In the motivating case of sort with a physical facility (\eg, a table of a specific size) that is usually a necessary trade-off.

\subsubsection{Modeling}

A multi-round pile shuffle is simply a sequence of basic shuffles where
the output of one round becomes the input of the next one:
A shuffle in $\numphase \geq 0$ rounds can be modeled by a pair
$(\typeSchedules, \schedules)$
of
a sequence of pile type assignments $\typeSchedules = \left( \pileType_1, \ldots, \pileType_\numphase \right)$
and
a sequence of pile assignments
%$\schedules = \left( \pile_\phase \right)_{\phase=1}^\numphase$,
$\schedules = \left( \pile_1, \ldots, \pile_\numphase \right)$.
These induce a sequence of permutations
(deck states)
according to the recurrence
\begin{align}
\label{eq:multi-round-defn}
\begin{cases}
\perm_0 = \perm, \\
\perm_{\phase} = \shuffleOp \, \pileType_\phase \, \pile_{\phase} \, \perm_{\phase - 1}
%	& \textrm{for $\phase \geq 1$}; \\
	& \textrm{for $1 \leq \phase \leq \numphase$}, \\
\varperm = \perm_\numphase
;
%\end{itemize}
\end{cases}
\end{align}
%
%such that
again, $\varperm \in \Perm{n}$ is the final deck order
resulting by
shuffling a deck starting in order $\perm \in \Perm{n}$ in this way.
We denote the relation
\begin{align*}
\label{eq:define-shuffle-multi}
\shuffleOp \, \typeSchedules \, \schedules \, \perm = \varperm.
\end{align*}
In the sequel,
we refer to sequences like $\typeSchedules$ as \emph{type schedules},
since they provide the individual type assignments used in each round of shuffle.

\begin{defn}
\label{defn:hetero-multi-types-sort}
We say a sequence $\typeSchedules$
of pile type assignments
in multiple rounds of shuffle
(a type schedule)
sorts a permutation $\perm$
if
there exists $\schedules$ such that
$\left( \typeSchedules, \schedules \right)$ sorts $\perm$.
\end{defn}

\subsubsection{Results}
\label{sec:multi-round-results}

The second key contribution of~\cite{treleaven2025sortingpermutationspileshuffle} was
%a mathematical reduction
a mathematical interpretation
% any instance of 
of multi-round shuffle on fixed pile types---\ie, given a type schedule $\typeSchedules$---%
as
%an equivalent 
single-round shuffle on fixed-type ``virtual piles'':

\begin{lemma}
\label{prop:hetero-multi-sort}
For any heterogeneous shuffle $(\typeSchedules, \schedules)$,
the derivation from~\cite{treleaven2025sortingpermutationspileshuffle}
(reproduced in Appendix~\ref{sec:multi-round-derivation})
obtains a single-round heterogeneous ``virtual'' shuffle $(\vrPileType_1, \vrpile_1)$
such that
\[
\shuffleOp \, \typeSchedules \, \schedules \, \perm = \varperm
\iff
\shuffleOp \, \vrPileType_1 \, \vrpile_1 \, \perm = \varperm
.
\]
\end{lemma}

%Again, 
$\vrpile_1$ can be thought of as an assignment
of the cards of the deck among a set of
``virtual piles''
in a corresponding virtual single-round shuffle
with pile types given by $\vrPileType_1$.
The equation~\eqref{eq:hetero-multi-embedding} defines $\vrpile_1$
%for each $\phase\in [\numphase]$
as into a co-domain
$[\vrNumPile_1] - 1$,
where
$\vrNumPile_1 = \prod_{\phase=1}^\numphase \numpile_{\phase}$.
Therefore
we observe that
$m$ piles over $T$ rounds has the capacity of $m^T$ piles in a single round.
Note that
we enumerate virtual piles starting from $\vrpile(1) = 0$,
instead of $\pile(1) = 1$.
That is because
it greatly simplifies
%which greatly simplifies 
both~\eqref{eq:hetero-multi-embedding} and the equation~\eqref{eq:type-indicator-recurrence}
for the pile types $\vrPileType_1$.

Remarkably,
we can fully precompute
the virtual pile types $\vrPileType_1$
given only the type schedule $\typeSchedules$.
(That is,
we needn't know $\schedules$ a priori.)
Moreover,
because~\eqref{eq:hetero-multi-embedding} is a reversible recurrence,
then
any virtual shuffle $\vrpile_1$ on $\vrPileType_1$---%
\eg, a minimal sort obtained via~\eqref{eq:min-piles-mixed}---%
can be transformed
into a 
%corresponding
multi-round shuffle $\schedules$ on $\typeSchedules$ with the same input-output relation.
Therefore, Lemma~\ref{prop:hetero-multi-sort} brings to bear
the full power of Lemma~\ref{lemma:hetero-func-bound},
confirming the role of the type schedule ($\typeSchedules$)
as a generalized linear-time checkable certificate of sort feasibility
%---through our virtual piles interpretation---%
in the multi-round setting.

%
%We will discuss virtual piles again in greater detail in Section~\ref{sec:algebra}.

\subsubsection{Repeated-round sort with a \emph{fixed} number of piles}

Although
%we are conjecturing
Conj.~\ref{conj:repeated-hard} suggests
that 
the repeated-round Prob.~\ref{prob:variable-round}
is \NPHard/,
if we restrict it
to an a priori fixed number $\numpile$ of piles,
%so that only the number $\numphase$ of rounds of shuffle is instance data,
then a brute force search over all possible assignments decides feasibility in (technically) polynomial time:
If
there are greater than
$\numphase' = \left\lceil \log_\numpile(n) \right\rceil$ rounds, then
a permutation of length $n$ is trivially sortable with $\numpile^{\numphase'} \geq n$ virtual piles.
Otherwise,
we need check only as many as
$\left( 2^\numpile \right)^{\numphase'}
\in O\left( n^{
%\frac{ \numpile}{\log_2 \numpile
(\numpile / \log_2 \numpile)
} \right)$
certificates, each in $O(n)$ time.
While this approach is technically polynomial-time in the permutation length,
it suffers exponential growth in the number of piles $\numpile$.
This is
%similar to
reminiscent of
the way SAT is technically polynomial-time on any fixed set of variables, even though it is \NPHard/ in general.

\makecommand{\invert}{\bar}
%\makecommand{\invert}{\droang}

%\makecommand{\dual}{*}
\makecommand{\dual}{\ssym}

%\section{An algebra for pile type compositions}
\section{An algebra for virtual piles types}
\label{sec:algebra}

The ``virtual shuffle'' interpretation
just presented
proves fundamental to the analysis of multi-round pile shuffle.
However,
the characterizing equation~\eqref{eq:type-indicator-recurrence} is convoluted and generally unintuitive.
In this section we present a small symbolic algebra
which concisely and intuitively captures its effects.
In addition to providing useful insight
about the formation of virtual piles
corresponding to the concatenation of multiple rounds of shuffle,
the algebra and its notation will be relied upon heavily
in the sequel.

\begin{defn}
\label{defn:factored-schedules}
Given a type schedule
$\typeSchedules = \left( \pileType_1, \pileType_2, \ldots, \pileType_\numphase \right)$,
we will denote by
$\pileType_1 \backslash \pileType_2 \backslash \ldots \backslash \pileType_\numphase$
the virtual pile type assignments $\vrPileType_1$ 
corresponding to $\typeSchedules$ in the sense of Lemma~\ref{prop:hetero-multi-sort}.
\end{defn}

The definition
introduces a new backslash ($\backslash$) operator with several noteworthy properties
that can be verified with~\eqref{eq:type-indicator-recurrence}.
We omit the proofs, which are purely technical.
\begin{defn}[Inverse]
\makecommand{\word}{\pileType}
\makecommand{\pt}{x}

We will denote by $\invert \word$ the \emph{inverse} of a type sequence $\word = \pt_1 \ldots \pt_m$, 
defined by
$\invert \qsym = \ssym$, $\invert \ssym = \qsym$,
and
$\invert \word = \invert \pt_1 \ldots \invert \pt_m$;
that is,
each pile switches to the opposite type.
For example, if $\word = \qsym\qsym\ssym$, then $\invert \word = \ssym\ssym\qsym$.
\end{defn}
\begin{defn}[Dual]
\makecommand{\word}{\pileType}
\makecommand{\pt}{x}

We will denote by $\word^\dual$ the \emph{dual} of $\word$, which is the combination of inversion and sequence reversal,
$\word^\dual = \invert \pt_m \ldots \invert \pt_1$.
For example, if $\word = \qsym\qsym\ssym$, then $\word^\dual = \qsym\ssym\ssym$.
(Pure sequence reversal can be captured by $\invert \word^\dual = \ssym\qsym\qsym$.)
\end{defn}
The inverse and dual are both self-inverting, in the sense that
${\invert {\invert \pileType}} = \pileType$ and $(\pileType^\dual)^\dual = \pileType$.

\makecommand{\word}{\pileType}

\begin{lemma}[Singleton compositions]
\label{label:singleton-compose-lemma}

\vphantom{new-line}

\begin{enumerate}
\item $\word \backslash \qsym = \qsym \backslash \word = \word$,
\item $\word \backslash \ssym = \word^\dual$, and
\item $\ssym \backslash \word = \invert\word$.
\end{enumerate}

\end{lemma}
%
%For the compositions with a single queue specifically, \ie,
%$\word \backslash \qsym = \qsym \backslash \word = \word$,
%we appeal to the intuition that
%a round of shuffle with a single queue does not change its input.

%
\begin{lemma}%[Composition distributes concatenation]
\label{lemma:concat-lemma}

\makecommand{\pt}{x}

Composition distributes concatenation, \ie,
\[
\pileType \backslash ( \pt_1 \ldots \pt_m )
=
(\pileType \backslash \pt_1) \ldots ( \pileType \backslash \pt_m )
,
\]
for any $\pileType \in \arbpile^*$, $\pt_1 \ldots \pt_m \in \arbpile^*$.
\end{lemma}
For example,
$\qsym\qsym\ssym \backslash \qsym\ssym
= (\qsym\qsym\ssym \backslash \qsym)(\qsym\qsym\ssym \backslash \ssym)
= (\qsym\qsym\ssym)(\qsym\ssym\ssym)
= \qsym\qsym\ssym \qsym\ssym\ssym
$.

%
%We observe that
Given Lemmas~\ref{label:singleton-compose-lemma} and~\ref{lemma:concat-lemma},
any composition
$\pileType_1 \backslash \pileType_2 \backslash \ldots \backslash \pileType_\numphase$
can be systematically reduced by pairs to compute~\eqref{eq:type-indicator-recurrence},
by using the lemmas alternatingly starting from the right.
For example, 
\begin{align*}
\qsym \qsym \ssym \backslash \qsym\ssym\ssym \backslash \ssym
&= \qsym \qsym \ssym \backslash (\qsym\ssym\ssym \backslash \ssym)
\\
&= \qsym \qsym \ssym \backslash \qsym\qsym\ssym
\\
&= (\qsym \qsym \ssym \backslash \qsym) (\qsym \qsym \ssym \backslash \qsym) (\qsym \qsym \ssym \backslash \ssym) 
%.
= (\qsym \qsym \ssym) (\qsym \qsym \ssym) (\qsym \ssym \ssym)
\\
&= \qsym \qsym \ssym \qsym \qsym \ssym \qsym \ssym \ssym
.
\end{align*}
In fact,
the composition operator is associative.
Any contiguous sub-recurrence
of~\eqref{eq:multi-round-defn}
%, \eg,
%on
%%$(\perm_i, \ldots, \perm_j)$
%$\left(
%%(\pileType_\phase)_{\phase=i}^j,
%(\pileType_i, \ldots, \pileType_j), \,
%(\pile_i, \ldots, \pile_j)
%\right)$,
%%$\left(
%%(\pileType_\phase, \pile_\phase)_i^j
%%\right)$,
defines a multi-round sub-shuffle,
equally subject to Lemma~\ref{prop:hetero-multi-sort}.
Therefore we obtain the same result if we reduce from the left instead,
or indeed by any other sequence of adjacent pairs.
%in general.

%\clearpage

%\makecommand{\problemVar}{{\mathbf P}}
\makecommand{\scenarioParams}{\theta}

\section{Type schedules as finite automata}
\label{sec:chains}

\makecommand{\profile}{\delta}
\makecommand{\profileAlphabet}{\Delta}
\makecommand{\profileFn}{\profileAlphabet}
\makecommand{\step}{b} % "bima" (roughly) is Greek for "step"

\makecommand{\endState}{{\it end}}

\makecommand{\dfaVar}{\pi}
\makecommand{\dfaSet}{\Pi}
\makecommand{\dfaFn}{\dfaSet}

The main objective of the paper is
to prove Theorem~\ref{thm:variable-round-hard},
asserting that
variable-round Dealer's choice pile shuffle sort
(Prob.~\ref{prob:variable-round})
is \NPHard/.
We will accomplish that by reducing $\SAT$ to each of a sequence of problems,
ultimately leading to Prob.~\ref{prob:variable-round}.
In order to do so
we exploit a 
%fundamental
useful
correspondence between type schedules---%
the certificates of sort feasibility---%
and a class of finite automata
introduced in this section.

Let us examine again
the recurrence equation~\eqref{eq:min-piles-mixed}
obtaining pile assignments $\pile^*$
of a minimal sort
of a permutation $\perm \in [n]!$
on pile types given by $\pileType$.
The right-hand side of the recurrence is determined strictly
by the direction of change at each place in the permutation, \ie
whether it is 
increasing or decreasing.
%an ascent or a descent.
Therefore we introduce the concept of a \emph{change profile}:
\begin{defn}
The \emph{change profile}
$\profileFn(\perm) =: \profile$
of a permutation $\perm \in [n]!$
is the sequence
on the alphabet 
$$\profileAlphabet \dfneq \{a[scent], d[escent]\}$$
defined for all $s \in [n - 1]$ by
\[
%\text{$\profile(s) = a$ if $\perm(s+1) > \perm(s)$ or else $\profile(s) = d$.}
\profile(s) = \begin{cases}
a,	& \perm(s+1) > \perm(s),	\\
d,	&
%\text{otherwise.}
\perm(s+1) < \perm(s)
.
\end{cases}
\]
\end{defn}
For example,
the change profile
of $\perm = 4 8 7 5 3 1 2 6$
is $\profile = addddaa$.

If $\profile$ is the change profile of a permutation $\perm$, then
the recurrence of~\eqref{eq:min-piles-mixed} can be written as
\begin{align}
\label{eq:trajectory-regular}
\pile^*(s+1) =
\pile^*(s)
+ \left[
	\pileType(\pile^*(s)) = \qsym, \,
	\profile(s) = d
\right]
+ \left[
	\pileType(\pile^*(s)) = \ssym, \,
	\profile(s) = a
\right]
.
\end{align}
\eqref{eq:trajectory-regular} has a useful interpretation
as the state update equation governing
a deterministic finite automaton (DFA) over the alphabet $\profileAlphabet$, where
$s$ represents discrete time, and
$\pile^*$ represents the trajectory of the state under input $\profile$.
\eqref{eq:trajectory-regular} describes the dynamics of a DFA
defined entirely by $\pileType$:
It is made of a chain of states, %\ie
one of a corresponding type for each type assignment of $\pileType$.
For example, $\pileType = \qsym\qsym\ssym\qsym$ defines a chain of five ($5$) nodes, shown in Fig.~\ref{fig:dfa-example-1}.

\begin{figure}[h!]
\centering
\begin{tikzpicture}

	\tikzset{
		->, %makes the edges directed
		>=stealth', %makesthearrowheadsbold
		node distance=1.5cm, %specifiestheminimumdistancebetweentwonodes.Changeifnecessary. 
		every state/.style={thick,fill=gray!10}, %setsthepropertiesforeach’state’node 
		initial text=$ $, %setsthetextthatappearsonthestartarrow 
	}
	
	\makecommand{\qnode}[1]{
	\node[state, accepting, right of ]
	}

	\node[state,initial,accepting] (q1) [label=below:$\qsym$] {$1$};
	\node[state,accepting,right of=q1] (q2) [label=below:$\qsym$] {$2$};
	\node[state,accepting,right of=q2] (q3) [label=below:$\ssym$] {$3$};
	\node[state,accepting,right of=q3] (q4) [label=below:$\qsym$] {$4$};
	\node[state,right of=q4] (end)
	%[label=below:$\endState$]
	{$5$};

	\draw 
		(q1) edge[loop above] node{$a$} (q1) 
		(q1) edge[above] node{$d$} (q2)
		(q2) edge[loop above] node{$a$} (q2) 
		(q2) edge[above] node{$d$} (q3)
		(q3) edge[loop above] node{$d$} (q3) 
		(q3) edge[above] node{$a$} (q4)
		(q4) edge[loop above] node{$a$} (q4) 
		(q4) edge[above] node{$d$} (end)
		(end) edge[loop right] node{$a$, $d$} (end)
;

\end{tikzpicture}
\caption{The deterministic automata (chain), $\dfaFn(\qsym\qsym\ssym\qsym)$.}
\label{fig:dfa-example-1}
\end{figure}

Queue states (piles) consume arbitrarily long ascending ``runs'' of a permutation,
but no descents, while stack states consume descending runs but no ascents:
Accordingly, in Fig.~\ref{fig:dfa-example-1},
the $\qsym$ nodes have $a$ as a self loop while $d$ transitions to the next node in the chain;
conversely,
the $\ssym$ node has $d$ as a self-loop and 
$a$ transitions to the next node.
A final terminal node, appended to the end of the chain, returns to itself on either $a$ or $d$.

We have seen that sort is feasible on $\pileType$
if and only if~\eqref{eq:min-piles-mixed} has a solution.
The trajectory of the change profile $\profileFn(\perm)$
starting from the beginning of the chain ($\pile^*(1) = 1$)
is a solution to~\eqref{eq:min-piles-mixed} as long as it
ends in one of the pile states (any but the terminal node).
%Otherwise, \ie, 
If the terminal state is reached, however, then
\eqref{eq:min-piles-mixed} has no solution;
$\pileType$ is too short to sort $\perm$.
Accordingly,
we mark only the pile states as \emph{accepting} states (double outline in Fig.~\ref{fig:dfa-example-1}),
but not the terminal state.
Then we obtain the following key results:
\begin{prop}
\label{prop:sort-by-automata}
A permutation $\perm$ can be sorted
by a single-round shuffle with the pile type assignments $\pileType$ if and only if
its change profile $\profileFn(\perm)$ is accepted 
by the corresponding automaton, which we denote $\dfaFn(\pileType)$.
%($\profile$ is therefore a member of an associated regular language.)
%(The equivalence between regular languages and finite automata is well known.)
\end{prop}
\begin{corollary}
A permutation $\perm$
can be sorted by a multi-round shuffle with the type schedule $\typeSchedules$ if and only if
$\profileFn(\perm)$ is accepted by $\dfaFn(\typeSchedules) \dfneq \dfaFn(\vrPileType_1)$
($\vrPileType_1$ of Lemma~\ref{prop:hetero-multi-sort}).
\end{corollary}
The proof of the corollary is a trivial application of Lemma~\ref{prop:hetero-multi-sort} to Prop.~\ref{prop:sort-by-automata}.

\makecommand{\dynamicsFn}{f}
\makecommand{\localStateVar}{\pile^*}
%\makecommand{\localStateVar}{\pile}

For completeness, we give a formal definition of $\dfaFn$:
\begin{defn}[Sort certificate chain]
The ``chain'' automaton $\dfaFn(\pileType)$ of a types assignment $\pileType$
on $\numpile$ piles
is defined over the alphabet $\profileAlphabet$ as follows:
It is on states $[\numpile + 1]$,
starting from $\pile^*(1) = 1$, and accepting the states $[\numpile]$.
It is governed by 
$\pile^*(s + 1) = \dynamicsFn_{\dfaFn(\pileType)}(\pile^*(s), \profile(s))$,
where
\[
	\dynamicsFn_{\dfaFn(\pileType)}(\localStateVar, \profile) = \begin{cases}
		\localStateVar + \left[
			\pileType(\localStateVar) = \qsym, \,
			\profile = d
		\right]
		+ \left[
			\pileType(\localStateVar) = \ssym, \,
			\profile = a
		\right]
		& \localStateVar \leq \numpile
		\\
		\numpile + 1 & \localStateVar = \numpile + 1
		.
	\end{cases}
\]
\end{defn}

\subsection{Objective reformulation}

\makecommand{\questionCons}{{\mathcal Q}}
\makecommand{\allowedSched}{\sortScenario}

\makecommand{\paramVar}{\theta}
\makecommand{\paramSet}{\Theta}

\makecommand{\word}{W}
%\makecommand{\alphabet}{\Sigma}
\makecommand{\availableChains}{\Pi}
\makecommand{\dfa}{\pi}

\makecommand{\sortScript}{{\rm SORT}}
\makecommand{\chainScript}{{\rm CHAIN}}

Prop.~\ref{prop:sort-by-automata} establishes a useful correspondence
between type schedules and certain chain-like discrete automata.
It allows us to reframe questions of sort feasibility---%
such as Prob.~\ref{prob:repeated-round} and Prob.~\ref{prob:variable-round}---%
in terms of change profiles and accepting chains.

The sort feasibility problems of this paper all have instances (questions)
expressible in the following form.
\begin{defn}%[Sorting pile types]
$\questionCons_\sortScript(\perm, \allowedSched)$~---~%
Given a permutation $\perm$ and a set $\allowedSched$ of permissible type schedules,
is there a legal shuffle $(\typeSchedules, \schedules)$---%
in the sense that $\typeSchedules \in \allowedSched$---%
which
sorts $\perm$?
(The definition is in terms of multi-round shuffle without loss of generality.)
\end{defn}
%The constraints of an instance of shuffle are fully captured by the set $\allowedSched$.
%
\begin{defn}
We will call the set of all such questions $\sortScript$;
an impractically large problem
which
includes both
Prob.~\ref{prob:repeated-round} and Prob.~\ref{prob:variable-round}
as subsets.
\end{defn}
Prop.~\ref{prop:sort-by-automata} motivates questions of the related form:
\begin{defn}%[Accepting chains]
$\questionCons_\chainScript(\profile, \allowedSched)$~---~%
Given a change profile $\profile$ and a set $\allowedSched$ of permissible type schedules,
is there a schedule $\typeSchedules\in\allowedSched$
such that the chain DFA $\dfaFn(\typeSchedules)$ accepts $\profile$?
\end{defn}
\begin{defn}
We will call the set of all such questions $\chainScript$.
\end{defn}

\makecommand{\clauseVar}{\phi}
\makecommand{\formulaVar}{\clauseVar}

In the sequel, we analyze a sequence of pile shuffle sort feasibility problems,
ultimately leading to 
%variable-round Dealer's choice sort, \ie,
Prob.~\ref{prob:variable-round}.
For each such variant $\probVar \subset \sortScript$,
we can reformulate it as
% (reduce it to) 
an accepting chain problem,
\[
\probVar'
=
\left\{
\questionCons_\chainScript \left(
  \profileFn(\perm), \allowedSched
\right)
\, : \,
\questionCons_\sortScript(\perm, \allowedSched) \in \probVar
\right\}
.
\]
$\probVar$ and $\probVar'$ are co-reducible in that
$\probVar'$ can be reduced back to $\probVar$
by any means of obtaining a permutation $\perm$ with a given change profile.
Then our objective will become,
for each such pair,
to develop a polynomial-time reduction of ${\rm SAT}$ to $\probVar'$,
proving that $\problemVar'$ (and therefore $\problemVar$) is \NPHard/.
$\SAT$ is the set of all questions of the following form:
\begin{defn}%[Satisfying variable assignments]
$\questionCons_{{\rm SAT}}(\formulaVar)$~---~%
Given a CNF formula $\formulaVar$,
is there a satisfying assignment of its variables, \ie, $\exists \, \bf x \in \SAT(\phi)$?
\end{defn}
We note that
sort feasibility is generally no harder than NP since
the type schedule $\typeSchedules$
%of all (real) pile types, 
is a checkable certificate.

%%
%To reduce $\SAT$ to $\questionSet'$
%we should demonstrate a polynomial-time algorithm which,
%for any instance $\questionCons_\SAT(\phi) \in \SAT$,
%produces a corresponding instance $\questionCons_\chainScript(\profile, \sortScenario) \in \problemVar'$
%which is affirmitive if and only if $\phi$ is satisfiable.
%%We denote the condition by ${\bf x} \in \SAT(\formulaVar)$.
%%

%\emph{Permutation synthesis:}
%Many permutations will share a given change profile in general, but
%it is easy to obtain one:
%A simple construction given a profile $\profile$ is as follows:
%Start with an empty list.
%Add unity ($1$) to it. Then, for every step $s \in [n-1]$,
%if $\profile(s) = a$, then append $(s+1)$ to the list, otherwise prepend it to the list.
%Finally, invert the list as if it were a permutation sequence $\perm^{-1}$ to obtain $\perm$.

\subsection{Trajectories of change profiles on chain automata}

\makecommand{\word}{W}

Developing our reductions of $\SAT$
to a sequence of accepting-chain problems
will require detailed analysis of the trajectories of certain change profiles
on various chain automata.
Therefore, to end this section we provide a small amount of mathematical notation
for categorizing strings according to their trajectories on certain chains.

Let $\dynamicsFn_\dfaVar(k, \word)$
denote the state
reached after a chain $\dfaVar$
starting from
state $k$
consumes a string $\word$.
We will denote by $k_1 \to_\dfaVar k_2$ the set of all strings
whose trajectories starting from $k_1$ on $\dfaVar$ end at $k_2$,
\ie,
all $\word$ such that $\dynamicsFn_\dfaVar(k_1, \word) = k_2$.
The subscript may be omitted when the chain $\dfaVar$ is implicit.
We will denote by $k_1 \to_\dfaVar \, \geq k_2$ the strings
for which $\dynamicsFn_\dfaVar(k_1, \word) \geq k_2$,
\ie, whose trajectories starting from $k_1$ end \emph{no earlier} than $k_2$.
(By definition,
$k_1 \to_\dfaVar k_2 \implies k_1 \to_\dfaVar \, \geq k_2$.)
Finally,
we denote by
$\geq k_1 \to_\dfaVar \, \geq k_2$ the strings
where $k \geq k_1 \implies \dynamicsFn_\dfaVar(k, \word) \geq k_2$,
\ie, whose trajectories starting no earlier than $k_1$ end no earlier than $k_2$.
(By definition,
$\geq k_1 \to_\dfaVar \, \geq k_2 \implies k_1 \to_\dfaVar \, \geq k_2$.)

The change profiles of our reductions will be constructed by the concatenation
of a number of \emph{words} (strings), or \emph{gadgets}.
%, \ie, $\profile = W_1 W_2 \ldots W_n$.
%
Each word will represent a claim about the chain
which is tested when the word is played over a segment of the chain,
based on where its trajectory starts.
Information about the result of that test is encoded in the position where the trajectory of the word ends,
which is passed forward as the starting position for the next word/test.
%%
%The final word of the profile tests the claim ``the chain accepts the change profile'';
%tautologically,
%the trajectory ends in some accepting state if the claim is true, whereas
%it reaches the non-accepting end state if the claim is false.
%
As we compose change profiles from smaller strings,
we rely on a key property:
% of trajectories on chain automata.
% for composition of the basic gadgets.
%
\begin{lemma}[Non-backtracking lemma]
\label{lemma:no-backtrack}
The trajectory 
of a string $\word \in \profileAlphabet^*$
cannot end earlier
on any $\dfaVar \in \dfaFn(\arbpile^*)$
by starting later, \ie,
$k_1 \to_\dfaVar \geq k_2
\implies
\geq k_1 \to_\dfaVar \geq k_2
$.
\end{lemma}
\begin{proof}
There is a fairly simple proof by contradiction.
Suppose $\word \in k'_1 \to k'_2$
for some $k_1' \geq k_1$, but $k_2' < k_2$.
Then the trajectories of $\word$ from $k_1$ and $k_1'$ cross each other.
Since the maximum step size on any chain is unity (1), then 
there must be a partition of $W$ into $W_1 W_2$, where,
after $W_1$, the trajectories from both $k_1$ and $k'_1$ simultaneously visit 
the same position $k'' \in [k'_1, k'_2]$.
Then $W_2$ goes from $k''$ to two different endpoints,
drawing the contradiction
%which is impossible,
since 
chains are deterministic.
\end{proof}

\section{Factored instances and our proof strategy}
\label{sec:proof-strategy}

\makecommand{\assignvec}{{\mathbf x}}
\makecommand{\scenarioParams}{\theta}

%\makecommand{\problemVar}{{\mathbf P}}
\makecommand{\problemVar}{\probVar}
%\makecommand{\designAlign}{\Phi}
%\makecommand{\designAlign}{\pileType^\dagger}
\makecommand{\designAlign}{\ddot\pileType}
%\makecommand{\designAlign}{\mathring{\pileType}}

%
As previously stated, 
we prove our main result
through reductions from the \NPHard/ Boolean satisfiability problem (SAT)
to a series of pile shuffle sort feasibility problems.
Given the framework introduced in Sections~\ref{sec:algebra} and~\ref{sec:chains},
we introduce the problems used to build our proof in this section.

\subsection{Factored chains}

The forthcoming variants of sort feasibility are distinguished
by constraints imposed in various rounds of shuffle,
ultimately captured by the allowable type schedules (the set $\sortScenario$) in each problem instance.
For sets 
%of type schedules
that factor
as the concatenation of $n$ rounds of shuffle
with the allowable type assignments given
in each round $k \in [n]$
independently,
\ie,
\[
	\sortScenario = \left\{
		\left( \pileType_1, \ldots, \pileType_n \right)
		\, : \,
		\pileType_1 \in \sortScenario_1,
		\ldots,
		\pileType_n \in \sortScenario_n
	\right\}
	,
\]
we will express them
with the notation
$\sortScenario \dfneq \sortScenario_1 \backslash \ldots \backslash \sortScenario_n$.
The notation is inspired by the fact that
$\dfaFn(\typeSchedules) = \dfaFn\left( \pileType_1 \backslash \ldots \backslash \pileType_n \right)$
for every
$
\typeSchedules
=
\left( \pileType_1, \ldots, \pileType_n \right)
%=
\in \sortScenario_1 \backslash \ldots \backslash \sortScenario_n$.
We refer to problem instances that can be represented in this way as \emph{factored} instances
and
to the corresponding chains as factored chains.
We note that repeated-round sort feasibility (Prob.~\ref{prob:repeated-round}) is
\[
%	\sortScript\mhyphen\left(\backslash \arbpile^\numpile\right)^\numphase
%	=
	\left\{
	\questionCons_\sortScript(
		\perm, \,
		\overbrace{
			\arbpile^m \backslash \ldots \backslash \arbpile^m
		}^{\times \numphase}
	)
	\, : \,
	\perm \in [\nats_0]!,
	\quad
	\numphase \geq 0,
	\,
	\numpile \geq 1
	\right\}
	,
\]
and the variable-round problem (Prob.~\ref{prob:variable-round}) is
\[
%	\sortScript\mhyphen\arbpile^{m_1} \backslash \ldots \backslash \arbpile^{m_\numphase}
%	=
	\left\{
	\questionCons_\sortScript\left(
		\perm, \arbpile^{m_1} \backslash \ldots \backslash \arbpile^{m_\numphase}
	\right)
	\, : \,
	\perm \in [\nats_0]!, \quad
	\numphase \geq 0,
	\,
	\numpile_1 \geq 1, \ldots, \numpile_\numphase \geq 1
	\right\}
	.
\]

The problems we present next
are motivated
by variable-round sort feasibility (Prob.~\ref{prob:variable-round}),
but allow us to begin with simpler reductions by
imposing additional constraints in certain rounds of shuffle.
We limit our attention to problems in three ($3$) rounds,
which is sufficient to demonstrate \NPHard/ness more generally.

%%
%We refer to the three rounds, in respective order, as
%the \emph{alignment} round, the \emph{assignment} round, and the \emph{collection} round.
%%

\subsection{Outline}

\makecommand{\Rom}[1]{{\rm \MakeUppercase{\romannumeral #1}}}
\makecommand{\probAlias}[1]{\questionSet_\Rom{#1}}

We will begin progress toward our main result
with a reduction from $\SAT$
to the following artificially constrained
sort feasibility problem.
Introducing
our paper's ``\emph{key}'' type sequence
$$\designAlign = \qsym\qsym\qsym\qsym\ssym\ssym,
$$
our first problem, denoted $\probAlias{1}$, has allowable schedules of the form
$\designAlign \backslash \arbpile^{m_2} \backslash \qsym^{m_3}$,
with $m_2$ and $m_3$ as problem parameters:
An instance 
%of $\probAlias{1}$
%($\probAlias{1}'$)
permits any type schedule in exactly three rounds such that:
(1) the first round type assignment is precisely $\designAlign$, using $|\designAlign| = 6$ piles,
(2) the second round is on $\numpile_2$ piles of any combination of types, and
(3) the third round is on $\numpile_3$ queues.
\[
	\makecommand{\theWidth}{m}
	\probAlias{1}
	\dfneq
	\left\{
	\questionCons_\sortScript\left(
		\perm,
		\designAlign \backslash \arbpile^{\theWidth_2} \backslash \qsym^{\theWidth_3}
	\right)
	\, : \,
	\perm \in [\nats_0]!, \quad
	\theWidth_2, \theWidth_3 \geq 1
	\right\}
	.
\]
(Our factor notation inherits from string notation
that
an assignment $\pileType$ may stand in as the singleton set of itself.)

The details of $\probAlias{1}$
% ($\probAlias{1}$') 
and other problems of the paper are summarized
in Table~\ref{table:all-problems}.
Each row of the table corresponds to a co-reducible pair of problems
in a function-parametrized family
defined by
\begin{equation}
\label{eq:func-sort}
\makecommand{\theFunc}{f}
\makecommand{\theParam}{\theta}
\makecommand{\theDomain}{\fnDomain}
\sortScript\mhyphen\theFunc
\dfneq \Big\{
	\questionCons_\sortScript(\perm, \theFunc(\theParam))
	\, : \,
	\perm \in [\nats_0]!
	, \,
	\theParam \in \theDomain(\theFunc)
\Big\}
%,
\end{equation}
and
\makecommand{\theFunc}{f}
\begin{equation}
\label{eq:func-chain}
\makecommand{\theParam}{\theta}
\makecommand{\theDomain}{\fnDomain}
\chainScript\mhyphen\theFunc
\dfneq \Big\{
	\questionCons_\chainScript(\profile, \theFunc(\theParam))
	\, : \,
	\profile \in \profileAlphabet^*
	, \,
	\theParam \in \theDomain(\theFunc)
\Big\}
.
\end{equation}
The domain and factor notation
of the parametrizing function $\theFunc$
are given in the second and third columns of Table~\ref{table:all-problems}, respectively.
The first column lists problem aliases used in the paper.
For example, $\probAlias{1}$---%
and its sister $\probAlias{1}' \subset \chainScript$---%
correspond to the first row of the table.
%
%We will denote by $\probAlias{1}'$ the sister of $\probAlias{1}$ in .
%
\begin{table}[]
%\makecommand{\theFunc}{\sortScenario}
\makecommand{\theFunc}{f}
%\makecommand{\theParam}{\theta}
\makecommand{\theDomain}{\fnDomain}
\makecommand{\theProbs}[1]{$\probAlias{#1}$ ($\probAlias{#1}'$)}

\centering
\begin{tabular}{|r|c|c|}
\hline
\text{
%Variant -
$\sortScript$-$\theFunc$ ($\chainScript$-$\theFunc$) Alias
}
	%& Param. space - $\Theta$
	& $\fnDomain(\theFunc)$
		& $\theFunc$ %(\theta)$
\\ \hline\hline
\theProbs{1}	& \multirow{3}{*}{$m_2, m_3 \geq 1$}	
	& $\designAlign \backslash \arbpile^{m_2} \backslash \qsym^{m_3}$
	\\ \cline{1-1} \cline{3-3} 
\theProbs{2}	&                        
	& $\designAlign \backslash \arbpile^{m_2} \backslash \arbpile^{m_3}$
	\\ \cline{1-1} \cline{3-3} 
\theProbs{3}	&
	& $\arbpile^6 \backslash \arbpile^{m_2} \backslash \arbpile^{m_3}$
	\\ \hline
\theProbs{4}	& $\numpile \geq 1$
	& $\arbpile^\numpile \backslash \arbpile^\numpile \backslash \arbpile^\numpile$
	\\ \hline
\theProbs{5}	& $k \geq 1$
	& $\designAlign^k \backslash \arbpile^{6k} \backslash \arbpile^{6k}$
	\\ \hline
\hline
\text{Prob.~\ref{prob:repeated-round}}
	& $\numphase \geq 0$, $\numpile \geq 1$
	& $\overbrace{
			\arbpile^m \backslash \ldots \backslash \arbpile^m
		}^{\times \numphase}$
	\\ \hline
\text{Prob.~\ref{prob:variable-round}}
	& $\numpile_1, \ldots, \numpile_\numphase \geq 1$
	& $\arbpile^{\numpile_1} \backslash \ldots \backslash \arbpile^{\numpile_\numphase}$
	\\ \hline
\end{tabular}
\caption{Problems of interest in pile shuffle sort feasibility.}
\label{table:all-problems}
\end{table}

Although $\probAlias{1}$ is instrumental
in the development of our thesis,
it is not of general interest in itself, because
the first and third rounds are artificially restricted.
Therefore, subsequently,
we work to remove its restrictions to generalize our results.
Our first relaxation is
to remove the third-round constraint, by reducing SAT to sort feasibility on permissible sets
of the form $\designAlign \backslash \arbpile^{m_2} \backslash \arbpile^{m_3}$
($\probAlias{2}$ in Table~\ref{table:all-problems}).
Finally,
we relax the first-round restriction by reducing SAT
to sort feasibility on arbitrary pile types in all three rounds, via
permissible sets of the form $\arbpile^6 \backslash \arbpile^{m_2} \backslash \arbpile^{m_3}$
($\probAlias{3}$).
Our main result, Theorem~\ref{thm:variable-round-hard}, follows from the last reduction
since variable-round sort feasibility---%
Prob.~\ref{prob:variable-round}, also shown in Table~\ref{table:all-problems}---%
contains $\probAlias{3}$ as a strict subset.

After proving Theorem~\ref{thm:variable-round-hard},
we will discuss potential approaches for reducing $\SAT$
to
sort feasibility on type schedules of the form
$\arbpile^m \backslash \arbpile^{m} \backslash \arbpile^{m}$
($\probAlias{4}$),
or shuffle in three identical (i.e., repeated) rounds---a subset of Prob.~\ref{prob:repeated-round}---%
which would prove Conjecture~\ref{conj:repeated-hard}.

\subsection{Trajectories on factored chains}

We have narrowed our focus to sort on factored instances,
in particular
where the number of piles in each round $\phase\in [\numphase]$ is fixed;
that is,
it is the same across all $\pileType \in \sortScenario_\phase$.
In examining the trajectories of change profiles on such instances' chains,
we will find it useful to group their non-terminal states into ``frames'':
Let us denote the number of piles in round $\phase$ by $\| \sortScenario_\phase \|$.
Borrowing music terminology, we refer to the individual states of a chain as \emph{beats}.
Every $\|\sortScenario_1\|$ beats will form a \emph{measure}, and
every $\|\sortScenario_2\|$ measures create a \emph{phrase}.
Thus the chain will consist of $\|\sortScenario_3\|$ phrases followed by the terminal state, which remains ungrouped.
(In the two-round case, we simply group all the measures into a single phrase.)
We illustrate with the example of the chain $\qsym\qsym\ssym \backslash \qsym\ssym \backslash \qsym\qsym$ in Fig.~\ref{fig:chain-parts}.
\makecommand{\exampleState}{0 \backslash 1 \backslash 1}
\begin{figure}[h!]
\includegraphics[width=\linewidth]{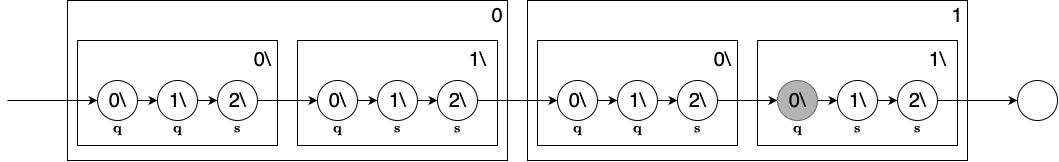}
\caption{Grouping of the chain $\qsym\qsym\ssym \backslash \qsym\ssym \backslash \qsym\qsym$
into two phrases, each of two measures, each of three beats.
A beat with coordinates $0 \backslash 1 \backslash 1$ is shaded.
}
\label{fig:chain-parts}
\end{figure}

This hierarchical grouping of states allows us
to describe positions on a chain using a coordinate notation:
%For example,
We will write $k_1 \backslash k_2 \backslash k_3$
to denote the $k_1$-th beat of the $k_2$-th measure of the $k_3$-th phrase, a short-hand for the position
$k_1 + \|\sortScenario_1\| \, (k_2 + \|\sortScenario_2\| \, k_3)$.
For example, the shaded state in Fig.~\ref{fig:chain-parts} can be written as $\exampleState$.
We use zero-based numbering in this setting, so that origins coincide, \ie, $0 \backslash 0 \backslash 0 = 0$,
and
we use as many coordinates as we have levels of grouping, typically two or three.

\makecommand{\profile}{\delta}
\makecommand{\profileAlphabet}{\Delta}
\makecommand{\step}{b} % "bima" (roughly) is Greek for "step"

\section{Elements of our reductions: Low-level gadgets}
\label{sec:gadgets}

\makecommand{\clause}{\phi}

% Segments
\makecommand{\gtClause}{{clause}}
\makecommand{\gtStart}{start\mhyphen clause}
\makecommand{\gtTest}{test}
\makecommand{\gtPos}{pos}
\makecommand{\gtNeg}{neg}
\makecommand{\gtDk}{dk}
\makecommand{\gtLastTest}{{end\mhyphen \gtTest}}
\makecommand{\gtLastPos}{endpos}
\makecommand{\gtLastNeg}{endneg}
\makecommand{\gtLastDk}{enddk}

\makecommand{\gtNext}{next}

\makecommand{\formSeg}{{formula}}

\makecommand{\gtQ}{{force\mhyphen q}}
%\makecommand{\noInvertSeg}{{noinv}}
%\makecommand{\noInvertSeg}{{req\mhyphen q}}
%\makecommand{\noInvertSeg}{{need\mhyphen q}}

\makecommand{\gtAlign}{align}

\makecommand{\gtPass}{pass}

% Beats
\makecommand{\bStart}{start}
\makecommand{\bChainDisq}{chain\mhyphen disq}
\makecommand{\bActd}{actd}
\makecommand{\bNActd}{\overline{actd}}
\makecommand{\bClauseDisq}{clause\mhyphen disq}
\makecommand{\bEnd}{end}

% clause def for bulk of paper
\makecommand{\clauseSeg}{\gtClause}
\makecommand{\mainClause}{\clauseSeg}

% LEGACY
\makecommand{\chainDisq}{\bChainDisq}
\makecommand{\litSeg}{\gtTest}
\makecommand{\lastLitSeg}{\gtLastTest}

We embark on the proof of Theorem~\ref{thm:variable-round-hard}
by presenting the main gadget
used throughout our reductions
to represent the clauses of a CNF formula.
In the sequel
we will demonstrate how to combine instances of the clause gadget to represent whole formulae.
\begin{defn}[Clause gadget]
\label{def:clause-word}
Given a CNF clause $\clause_j$ 
on $n \geq 1$ variables,
we define an embedding of the clause
within $\profileAlphabet^*$
by
\[
\mainClause(\clause_j)
\dfneq
\gtStart
\, \litSeg_1(\clause_j)
\, \litSeg_2(\clause_j)
\ldots \litSeg_{n-1}(\clause_j)
\, \lastLitSeg_n(\clause_j)
,
\]
where:
\begin{flushleft}
\[
\begin{array}{ll}

\gtStart & \dfneq daaaadda, \\

\litSeg_i(\clause_j) &\dfneq \begin{cases}
\gtPos	& x_i \in \clause_j,
\\
\gtNeg	& \neg x_i \in \clause_j,
\\
\gtDk	& \textrm{otherwise},
\end{cases}
\\

\gtPos &\dfneq aaddddadaaadad, \\
\gtNeg &\dfneq daddddadaaadad, \\
\gtDk &\dfneq adadddadaaadad, \\

\lastLitSeg_i(\clause_j) &\dfneq \begin{cases}
\gtLastPos	& x_i \in \clause_j,
\\
\gtLastNeg	& \neg x_i \in \clause_j,
\\
\gtLastDk	& \textrm{otherwise},
\end{cases}
\\

\gtLastPos &\dfneq aaddddaddddaaaaddaaaadddda, \\

\gtLastNeg &\dfneq daddddaddddadaaadaaaadaaaadddda,	\\

\gtLastDk &\dfneq daaddddaddddaaaaddaaaadddda.

\end{array}
\]
\end{flushleft}

\end{defn}

In the rest of this section,
%Next, 
we characterize the trajectories of $\mainClause$ and its parts
on chains in the family $\designAlign \backslash \arbpile^*$.
(We omit the $\dfaFn$ in $\dfaFn(\sortScenario_1 \backslash \ldots \backslash \sortScenario_m)$
when context makes it implicit.)
If a chain 
%$\pi$
is in $\pileType \backslash \arbpile^*$
for some types assignment $\pileType$, then
we say it is \emph{aligned} to $\pileType$.
%, or $\pileType$-aligned.
%
Chains aligned to the word
$\designAlign = \qsym\qsym\qsym\qsym\ssym\ssym$
% is a special type sequence which,
% along with its dual $\designAlign^\dual$,
provide useful interactions with the components of $\mainClause$.
We are particularly interested in chains which
contain the products of composing $\designAlign$
with embeddings of variable assignments
defined as follows:
%
%\makecommand{\assignEmbed}{\pileType}
\makecommand{\assignEmbed}{\Psi}
\begin{defn}[Assignment embeddings]
%Abusing notation, we define
Let
$\assignEmbed(\top) \dfneq \qsym$ and 
$\assignEmbed(\bot) \dfneq \ssym$.
For any variable assignment $\assignvec = x_1 \ldots x_n \in \{\top, \bot\}^n$,
we define $\assignEmbed(\assignvec) \dfneq \assignEmbed(x_1) \ldots \assignEmbed(x_n)$.
\end{defn}

The key result about $\mainClause$ is the following technical lemma,
which places
its instances 
among various word categories
(based on their trajectories)
using the previously-defined arrow notation.
The lemma refers to three specific beats by aliases
$\bStart \dfneq 0$, $\bChainDisq \dfneq 1$, and $\bEnd \dfneq 5$.

\begin{lemma}
\label{lemma:clause-gadget}
\makecommand{\word}{\mainClause(\clause_j)}
Given a CNF clause $\clause_j$ on $n \geq 1$ variables,
a prefix length $k_1 \geq 0$ (in measures),
and a suffix length $k_2 \geq 1$,
the following statements hold
for every
chain automaton
$\designAlign \backslash A_0 A_1 A_2 \in \designAlign \backslash \arbpile^{k_1} \arbpile^{n+1} \arbpile^{k_2}$:
\begin{enumerate}

\item
if
$A_1 = \assignEmbed(\assignvec) \, \qsym$
for some $\assignvec \in \SAT(\clause_j)$, then
$\word \in \bStart \backslash k_1 \to \bEnd \backslash (k_1 + n)$;

\item otherwise,
$\word \in \bStart \backslash k_1 \to \, \geq \bChainDisq \backslash (k_1 + n + 1)$;

\item
$
\word \in
\bChainDisq \backslash k_1 \to \, \geq \bChainDisq \backslash (k_1 + n + 1)$, unconditionally.

\end{enumerate}

\end{lemma}
We will prove the lemma at the end of the section
after developing additional machinery.

\makecommand{\theClause}{\mainClause(\clause_j)}

Lemma~\ref{lemma:clause-gadget}
%The lemma
is concerned with
segments
of $n+1$ measures in length
($\designAlign \backslash A_1$)
within $\designAlign$-aligned chains,
and with
two starting positions and two ending positions relative to them.
Its three claims are summarized graphically in Fig.~\ref{fig:traj-clause},
where
the word $\theClause$ and any conditions of its membership in a given category are shown
next to the corresponding arrow.
For a segment starting after $k_1$ measures,
the starting positions are $\bStart \backslash k_1$ and $\bChainDisq\backslash k_1$,
shown on the left-hand side of Fig.~\ref{fig:traj-clause}.
These are
the first and second beats, respectively, of the segment's first measure.
The ending positions,
shown on the right-hand side, are
$\bEnd \backslash (k_1 + n)$ and $\bChainDisq\backslash (k_1 + n + 1)$.
$\bEnd \backslash (k_1 + n)$ is the last beat of the last measure---the very last position---of the segment,
while
$\bChainDisq\backslash (k_1 + n + 1)$ is the second beat of the \emph{next} measure,
%\ie, 
beyond the segment and into the suffix.

\tikzset{
  dotnode/.style={
    circle, fill, inner sep=0pt, minimum size=4pt
  }
}

\tikzset{
myarrow/.style={>=latex}
}

\begin{figure}[h!]
\centering
\begin{tikzpicture}

	\makecommand{\theWord}{\gtClause(\clause_j)}
	\makecommand{\theWidth}{7}

  % Define the four corner nodes
  \node[dotnode, label=below left:{$\bStart \backslash k_1$}] (A) at (0, 4) {};
  \node[dotnode, label=below:{$\geq \bChainDisq \backslash k_1$}] (B) at (0, 0) {};
  \node[dotnode, label=below:{$\bEnd \backslash (k_1 + n)$}] (C) at (\theWidth, 4) {};
  \node[dotnode, label=below:{
		$\geq \bChainDisq \backslash (k_1 + n + 1)$
  }] (D) at (\theWidth, 0) {};

  % Draw an arrow from A to D with a label halfway
	\draw[->, myarrow] (A) -- node[midway, above, sloped]
	{$\theWord$ if $A_1 \in \assignEmbed\left( \SAT(\clause_j) \right) \, \qsym$}
	(C)
	;
	\draw[->, myarrow] (A) -- node[midway, below, sloped] {
	} (C);
	\draw[->, myarrow] (A) -- node[midway, above, sloped] {$\theWord$, otherwise} (D);
	\draw[->, myarrow] (B) -- node[midway, above, sloped] {$\theWord$ [unconditionally]} (D);
\end{tikzpicture}
\caption{Types of trajectories of $\theClause$ on $\designAlign \backslash A_0 A_1 A_2$.}
\label{fig:traj-clause}
\end{figure}

According to Lemma~\ref{lemma:clause-gadget} (first claim) if
the segment $\designAlign \backslash A_1$ encodes a satisfying assignment of $\clause_j$, then
a trajectory of $\theClause$ 
starting from $\bStart \backslash k_1$
will end at $\bEnd \backslash (k_1 + n)$.
Otherwise (second claim)
it will end 
no earlier than 
$\bChainDisq\backslash (k_1 + n + 1)$.
An important observation then is that
as a standalone chain,
\ie, with $k_1 = k_2 = 0$,
$\designAlign \backslash A_1$
%(cropped)
accepts $\theClause$
if and only if it
%$A_1$
encodes a satisfying assignment of $\clause_j$.
Meanwhile, the trajectories of $\theClause$
starting from $\bChainDisq\backslash k_1$---%
or any later, given Lemma~\ref{lemma:no-backtrack}---%
reach at least $\bChainDisq\backslash (k_1 + n + 1)$
%or further,
regardless of the segment's composition.

The layout
of claims in Fig.~\ref{fig:traj-clause}
evokes the form of
%illustrates how an instance of $\gtClause$ acts as 
a conjunctive \emph{gate},
shown graphically in Fig.~\ref{fig:conj-gate},
where
a proposition
encoded by the starting position ($A$)
is combined with a proposition
encoded by the action of the word itself ($B$).
In the sequel we will show how to connect instances of $\gtClause$ together in series
to generate a test for an entire formula.
\begin{figure}[h!]
\centering
\begin{tikzpicture}

	\makecommand{\theWord}{clause(\clause_j)}
	
  % Define the four corner nodes
	\node[dotnode, label=left:{$A$}] (A) at (0, 2) {};
	\node[dotnode, label=left:{$\overline A$}] (B) at (0, 0) {};
	\node[dotnode, label=right:{$AB$}] (C) at (2, 2) {};
	\node[dotnode, label=right:{$\overline{AB}$}] (D) at (2, 0) {};

  % Draw an arrow from A to D with a label halfway
	\draw[->, myarrow] (A) -- node[midway, above, sloped] {$B$} (C);
	\draw[->, myarrow] (A) -- node[pos=.6, anchor=south west] {$\overline B$} (D);
	\draw[->, myarrow] (B) -- node[midway, above, sloped] {} (D);
\end{tikzpicture}
\caption{Structure of a conjunctive gate gadget (``AND'')}
\label{fig:conj-gate}
\end{figure}
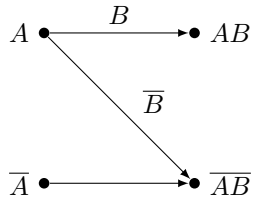

Lemma~\ref{lemma:clause-gadget}
demonstrates how
the starting positions of words in a change profile
influence their trajectories.
% are instrumental
%in passing information forward from word to word along a chain.
%
Moreover, 
it alludes to a semantic significance
to certain beats in certain measures.
%of chains.
We have already introduced the beat aliases
$\bStart$,
$\bEnd$,
and $\bChainDisq$ (or, ``chain disqualified'').
In the sequel we will encounter three more aliases---%
$\bActd$ (``clause activated''),
$\bNActd$ (``clause not activated''),
and $\bClauseDisq$ (``clause disqualified'').
To assist our intuition in the sequel,
we will refer to the aforementioned beats throughout the analysis
by the given aliases, collected in Table~\ref{table:state-semantics} below.
%
%Additionally, 
Table~\ref{table:state-semantics} associates a color with each such beat, 
which enriches some visual aids in the sequel,
such as Fig.~\ref{fig:design-chain} showing the chain $\designAlign$.

\begin{table}[h!]
\centering
\begin{tabular}{|r|l|c|}
\hline
\textbf{Beat}							  & \textbf{Alias}        & \textbf{Color} \\ \hline\hline
$0\backslash\ldots$                       & $\bStart\backslash\ldots$               & blue           \\ \hline
$1\backslash\ldots$                       & $\bChainDisq\backslash\ldots$ & red            \\ \hline
$2\backslash\ldots$                       & $\bActd\backslash\ldots$                & green          \\ \hline
$3\backslash\ldots$                       & $\bNActd\backslash\ldots$     & yellow         \\ \hline
$5\backslash\ldots$                       & $\bEnd\backslash\ldots$, $\bClauseDisq \backslash\ldots$         & gold/orange    \\ \hline
\end{tabular}
\caption{Beat aliases on each measure of a $\designAlign$-aligned chain.}
\label{table:state-semantics}
\end{table}
\begin{figure}[h!]
\centering
\begin{tikzpicture}

	\tikzset{
		->, %makes the edges directed
		>=stealth', %makesthearrowheadsbold
		node distance=1.5cm,
		every state/.style={thick,fill=gray!10,fill opacity=.2, text opacity=1}, 
		initial text=$ $,
	}
	
	\makecommand{\qnode}[1]{
	\node[state, accepting, right of ]
	}

	\definecolor{blue}{HTML}{4682B4}
	\definecolor{red}{HTML}{DC143C}
	\definecolor{green}{HTML}{32CD32}
	\definecolor{yellow}{HTML}{FFD700}
	\definecolor{gold}{HTML}{FFA500}
	
	\node[state,initial,accepting,fill=blue] (q1) [label=below:$\qsym$] {$0$};
	\node[state,accepting,fill=red,right of=q1] (q2) [label=below:$\qsym$] {$1$};
	\node[state,accepting,fill=green,right of=q2] (q3) [label=below:$\qsym$] {$2$};
	\node[state,accepting,fill=yellow,right of=q3] (q4) [label=below:$\qsym$] {$3$};
	\node[state,accepting,right of=q4] (q5) [label=below:$\ssym$] {$4$};
	\node[state,accepting,fill=gold,right of=q5] (q6) [label=below:$\ssym$] {$5$};
	\node[state,right of=q6] (end) {$6$};

	\draw 
		(q1) edge[loop above] node{$a$} (q1) 
		(q1) edge[above] node{$d$} (q2)
		(q2) edge[loop above] node{$a$} (q2) 
		(q2) edge[above] node{$d$} (q3)
		(q3) edge[loop above] node{$a$} (q3) 
		(q3) edge[above] node{$d$} (q4)
		(q4) edge[loop above] node{$a$} (q4) 
		(q4) edge[above] node{$d$} (q5)
		(q5) edge[loop above] node{$d$} (q5) 
		(q5) edge[above] node{$a$} (q6)
		(q6) edge[loop above] node{$d$} (q6) 
		(q6) edge[above] node{$a$} (end)
		(end) edge[loop above] node{$a$, $d$} (end)
;

\end{tikzpicture}
\caption{The chain $\designAlign$. Node color indicates beat significance from Table~\ref{table:state-semantics}.}
\label{fig:design-chain}
\end{figure}

We prepare for the proof of Lemma~\ref{lemma:clause-gadget} 
by examining the trajectories of each word in the clause gadget independently.
We start with the word $\gtStart$,
which prepares a chain to receive tests of variable assignments.
\begin{lemma}[Trajectories of $\gtStart$]
\label{lemma:start-clause}
The following statements hold
for every $n \geq 1$,
every chain $\pi \in \designAlign \backslash \arbpile^{n}$,
and for every measure $k \in [n]$:
\begin{itemize}
\makecommand{\theWord}{{\gtStart}}

\item
$\theWord \in \bStart \backslash k \to \bNActd \backslash k$, \ie,
a start state becomes a clause not activated state;

\item
$\theWord \in \, \geq \bChainDisq \backslash k \to \, \geq \bClauseDisq \backslash k$, \ie,
chain disqualification becomes clause disqualification.
\end{itemize}

\end{lemma}
The claims of Lemma~\ref{lemma:start-clause} are summarized in the previous fashion in Fig.~\ref{fig:traj-start-clause}.
\begin{figure}[h!]

\makecommand{\theWord}{{\gtStart}}

\centering
\begin{tikzpicture}
  % Define the four corner nodes
  \node[dotnode,
	label=left:{$\bStart \backslash k$}] (A) at (0, 1.5) {};
  \node[dotnode,
	label=left:{$\geq \bChainDisq \backslash k$}] (B) at (0, 0) {};
  \node[dotnode,
	label=right:{$\bNActd \backslash k$}] (C) at (4, 1.5) {};
  \node[dotnode,
	label=right:{$\geq \bClauseDisq \backslash k$}] (D) at (4, 0) {};
  % Draw an arrow from A to D with a label halfway
	\draw[->, myarrow] (A) -- node[midway, above, sloped] {$\theWord$} (C);
	\draw[->, myarrow] (B) -- node[midway, above, sloped] {$\theWord$} (D);
\end{tikzpicture}
\caption{
Types of trajectories of $\theWord$ on $\designAlign \backslash \arbpile^*$.
%Preparing a chain to test variable assignments.
}
\label{fig:traj-start-clause}
\end{figure}

\begin{proof}

We use the opportunity
of proving Lemma~\ref{lemma:start-clause}
to introduce a visual aid used for other more complex gadgets in the sequel.
\begin{figure}[]
	\centering
	\hfill
	\begin{subfigure}[t]{.4\linewidth}
		\includegraphics[width=\linewidth]{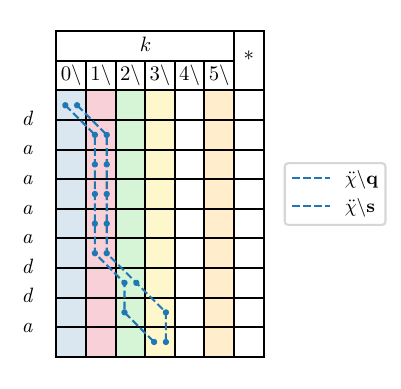}
		\caption{$\gtStart$ trajectories from $\bStart \backslash k$}
		\label{fig:start-clause-0}
	\end{subfigure}
	\hfill
	\begin{subfigure}[t]{.4\linewidth}
		\includegraphics[width=\linewidth]{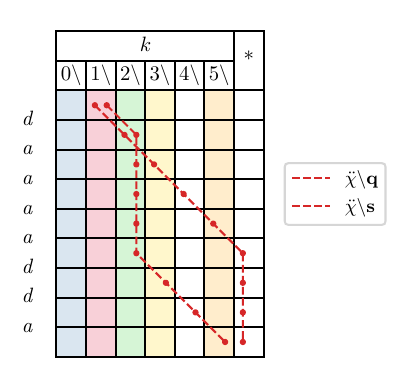}
		\caption{$\gtStart$ trajectories from $\bChainDisq \backslash k$}
		\label{fig:start-clause-1}
	\end{subfigure}
	\hfill
\caption{Trajectories of $\gtStart$
from two starting positions
on the $k$-th measure
of the chain automata in $\designAlign \backslash \arbpile^*$.}
\label{fig:start-clause}
\end{figure}
Figure~\ref{fig:start-clause} shows trajectories created by the word $\gtStart$
on the $k$-th measure
on every permissible chain $\pi \in \designAlign \backslash \arbpile^{*}$.
Fig.~\ref{fig:start-clause-0} shows the trajectories that start from the state $\bStart \backslash k$, while
Fig.~\ref{fig:start-clause-1} shows the trajectories that start from the state $\bChainDisq \backslash k$.
Only the $k$-th measure is shown in each case, because it contains the relevant parts of all trajectories.
(Some later diagrams require additional measures to be framed.)
Over just the $k$-th measure,
every chain looks either like $\designAlign \backslash \qsym$, or like $\designAlign \backslash \ssym$, so
only those two \emph{classes} of trajectories can be distinguished in the two diagrams.

A marker occupying the grid cell in row $t$ and column $\pile$ indicates that
the corresponding trajectories visit state $\pile$ at time $t$, i.e.,
after the first $t$ symbols of the input are consumed.
The letters of $\gtStart$ are written down the left-hand side of the grid,
each next to its associated time step.
The top-left corner of the grid is $(t, \pile) = (0, 0 \backslash k)$.
Each column associated with a beat from Table~\ref{table:state-semantics} has been highlighted in the associated color.
The final column ($*$) represents the union of all states
numbered higher than those in the displayed frame,
including the non-accepting terminal state.
The final row shows the final state reached by the end of the word.

We obtain the lemma since
both kinds of trajectory from $\bStart \backslash k$ (Fig.~\ref{fig:start-clause-0})
terminate at $\bNActd \backslash k$,
while
both kinds of trajectory from $\bChainDisq \backslash k$ (Fig.~\ref{fig:start-clause-1})
terminate at or beyond $\bClauseDisq \backslash k$.
\end{proof}

Next we examine the trajectories of the words $\gtPos$, $\gtNeg$, and $\gtDk$.
The words $\gtPos$ and $\gtNeg$ are meant to activate a clause
if the variable under test is true ($\top$) or false ($\bot$) respectively.
The word $\gtDk$ (``don't care'') simply passes forward the latest activation status of the clause,
regardless of the variable assignment.
\begin{lemma}[Trajectories on variable assignments]
\label{lemma:lit-gadget}
The following statements hold
for every $n \geq 1$,
every chain $\designAlign \backslash \pileTypeSingle_0 \ldots \pileTypeSingle_n \in \designAlign \backslash \arbpile^{n + 1}$,
and
for each measure $k \in [n]-1$:
\begin{itemize}

\item
Activation \emph{may} occur; i.e.,
$\bNActd \backslash k \to \, \bActd \backslash (k+1)$
contains
$\gtPos$ if $\pileTypeSingle_k = \qsym$, and 
$\gtNeg$ if $\pileTypeSingle_k = \ssym$;

\item otherwise, the unactivated state propagates, i.e.,
$\bNActd \backslash k \to \, \bNActd \backslash (k+1)$
contains
$\gtDk$,
$\gtPos$ if $\pileTypeSingle_k \neq \qsym$, and 
$\gtNeg$ if $\pileTypeSingle_k \neq \ssym$%
;

\item
the activated state propagates, i.e.,
$\bActd \backslash k \to \bActd \backslash (k+1) \supset \{ \gtPos, \gtNeg, \gtDk \}$;

\item
disqualification also propagates:
$\bClauseDisq \backslash k \to \, \geq \bClauseDisq \backslash (k+1)
\supset \{ \gtPos, \gtNeg, \gtDk \}$%
.

\end{itemize}
\end{lemma}

\begin{proof}
The proof is again exhaustive and can be verified with the diagrams
for $\gtPos$ and $\gtNeg$ in Fig.~\ref{fig:pos-neg-traj}, and
for $\gtDk$ in Fig.~\ref{fig:dk-traj}.
The figures frame two measures of $\designAlign$ instead of one.
\end{proof}

\begin{figure}[p]
	% act=2, nact=3, disq=5
	\makecommand{\theFrac}{.48}
	\centering

	\begin{subfigure}[t]{\theFrac\linewidth}
		\includegraphics[width=\linewidth]{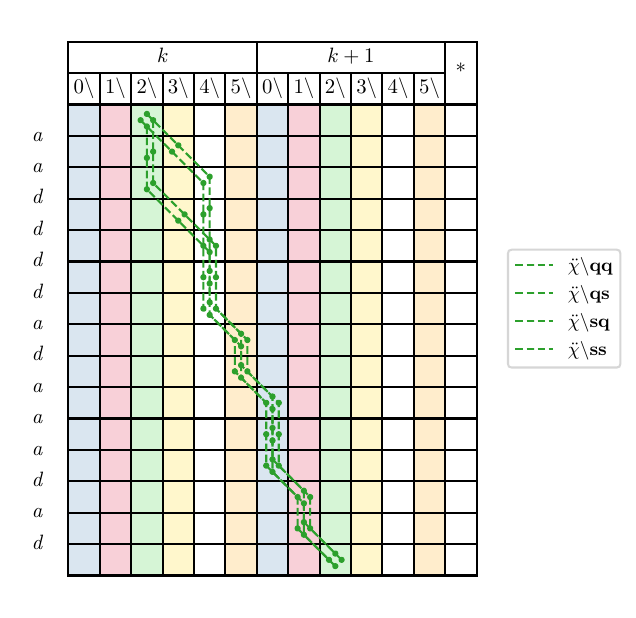}
		\caption{$\gtPos$ starting activated}
	\end{subfigure}
	\begin{subfigure}[t]{\theFrac\linewidth}
		\includegraphics[width=\linewidth]{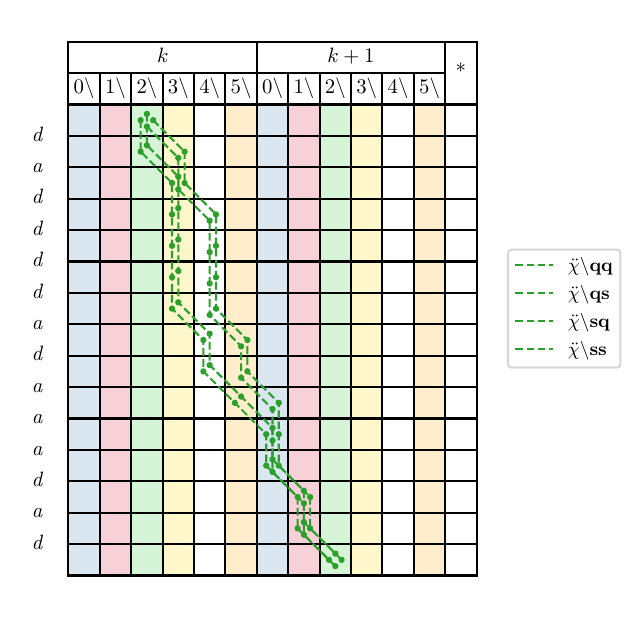}
		\caption{$\gtNeg$ starting activated}
	\end{subfigure}

	\begin{subfigure}[t]{\theFrac\linewidth}
		\includegraphics[width=\linewidth]{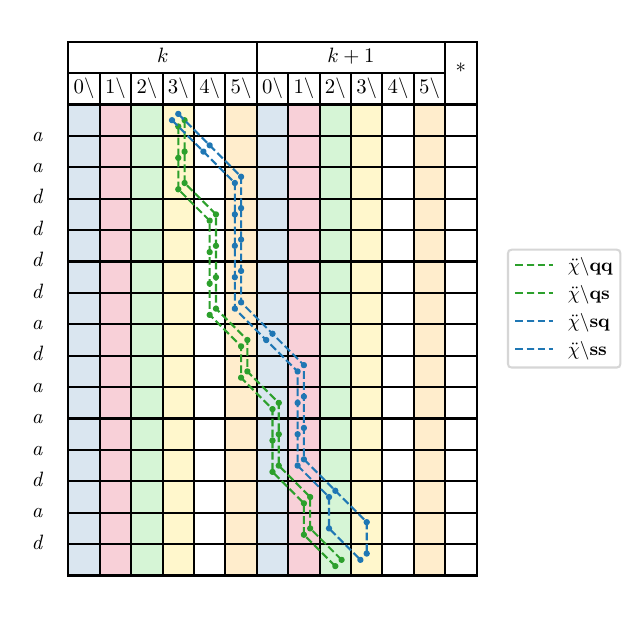}
		\caption{$\gtPos$ starting not activated}
	\end{subfigure}
	\begin{subfigure}[t]{\theFrac\linewidth}
		\includegraphics[width=\linewidth]{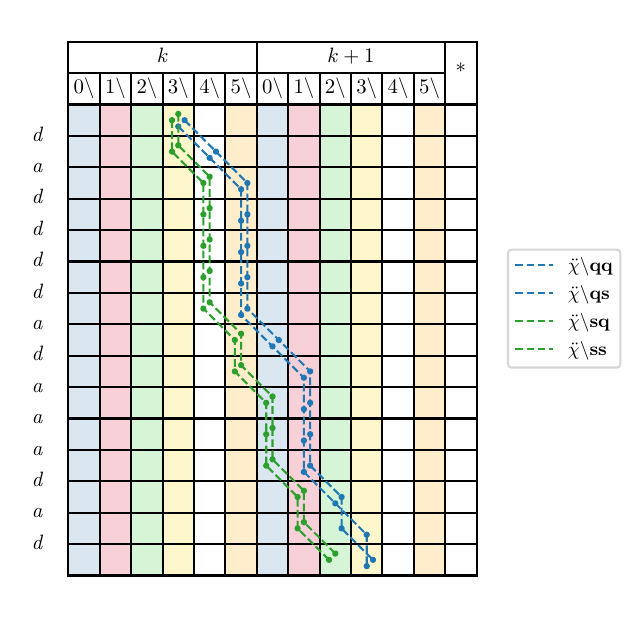}
		\caption{$\gtNeg$ starting not activated}
	\end{subfigure}

	\begin{subfigure}[t]{\theFrac\linewidth}
		\includegraphics[width=\linewidth]{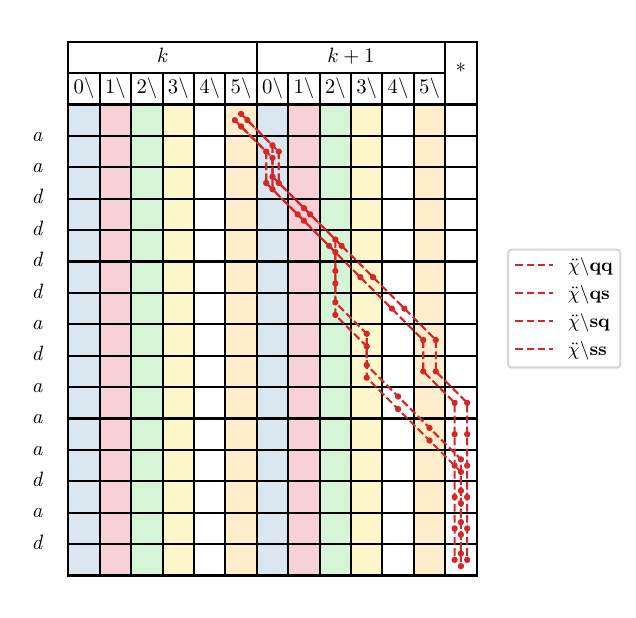}
		\caption{$\gtPos$ starting disqualified}
	\end{subfigure}
	\begin{subfigure}[t]{\theFrac\linewidth}
		\includegraphics[width=\linewidth]{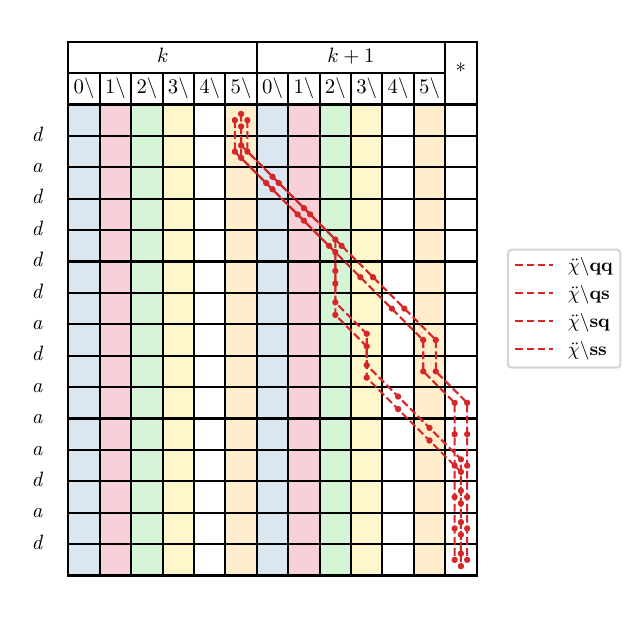}
		\caption{$\gtNeg$ starting disqualified}
	\end{subfigure}

\caption{Trajectories of $\gtPos$ and $\gtNeg$
from three starting positions
on measure $k$
of the chain automata in $\designAlign \backslash
%\arbpile^{n+1}
\arbpile^*
$.
}
\label{fig:pos-neg-traj}
\end{figure}

Once again we summarize the claims of the lemma graphically, in Fig.~\ref{fig:test-variable},
and
Fig.~\ref{fig:gate-lit} shows a corresponding gate structure
in terms of the relations between encoded propositions.
Comparing them demonstrates the role of instances of $\gtTest$
within an instance of $clause$,
of populating the disjunctive clause at the end of a conjunctive formula.

%\clearpage

\begin{figure}[t]
% act=2, nact=3, disq=5
\makecommand{\theFrac}{.3}
\centering
	\begin{subfigure}[t]{\theFrac\linewidth}
		\includegraphics[width=\linewidth]{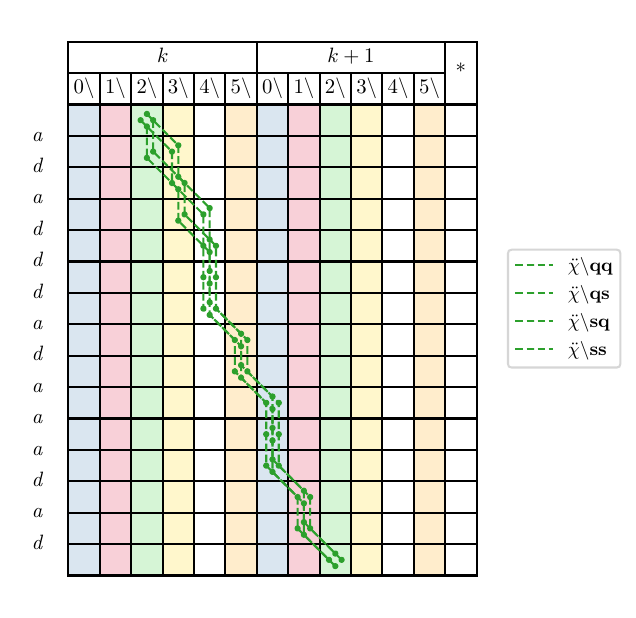}
		\caption{$\gtDk$ starting activated}
	\end{subfigure}
	\begin{subfigure}[t]{\theFrac\linewidth}
		\includegraphics[width=\linewidth]{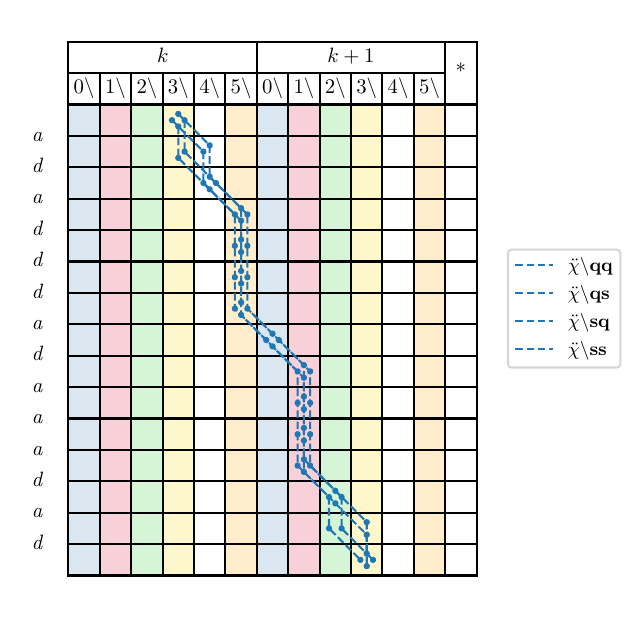}
		\caption{$\gtDk$ starting not activated}
	\end{subfigure}
	\begin{subfigure}[t]{\theFrac\linewidth}
		\includegraphics[width=\linewidth]{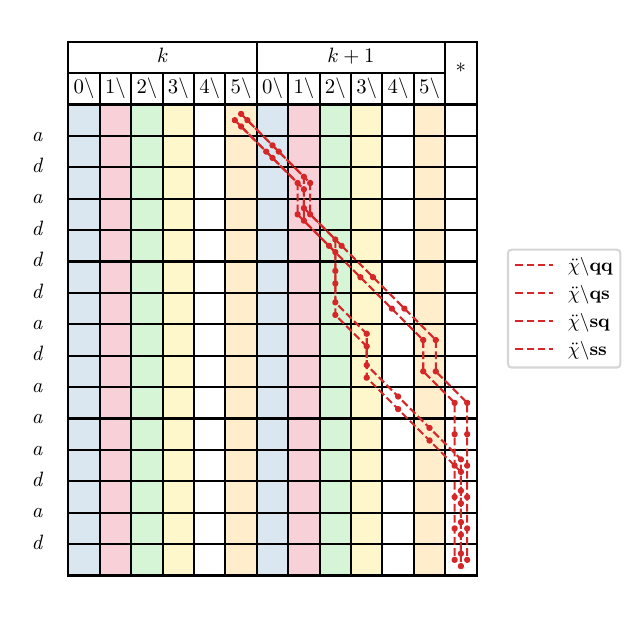}
		\caption{$\gtDk$ starting disqualified}
	\end{subfigure}
\caption{Trajectories of $\gtDk$
from three starting positions
on measure $k$
of the chain automata in $\designAlign \backslash
% \arbpile^{n+1}$.
\arbpile^*$.
}
\label{fig:dk-traj}
\end{figure}

\begin{figure}[t]
\centering
\begin{tikzpicture}
  % Define the four corner nodes

%	\draw[->, myarrow] (A) -- node[midway, above, sloped] {$\theWord$} (C);
	\node[dotnode, label=left:{$\bActd \backslash k$}] (A) at (0, 4) {};
	\node[dotnode, label=left:{$\bNActd \backslash k$}] (B) at (0, 2) {};
	\node[dotnode, label=below:{$\geq \bClauseDisq \backslash k$}] (C) at (0, 0) {};
	
	\node[dotnode, label=right:{$\bActd \backslash (k+1)$}] (D) at (6, 4) {};
	\node[dotnode, label=right:{$\bNActd \backslash (k+1)$}] (E) at (6, 2) {};
	\node[dotnode, label=below:{$\geq \bClauseDisq \backslash (k+1)$}] (F) at (6, 0) {};

  % Draw an arrow from A to D with a label halfway
	\draw[->, myarrow] (A) -- node[midway, above] {
	%$\{ \gtPos, \gtNeg, \gtDk \}$
	$\gtPos, \gtNeg, \gtDk$
	} (D); 
	\draw[->, myarrow] (C) -- node[midway, above] {
	%$\{ 
	$\gtPos, \gtNeg, \gtDk$
	% \}$
	} (F); 
	\draw[->, myarrow] (B) -- node[midway, anchor=north west, pos=.6] {
	$\gtPos$ if $\pileTypeSingle_k = \qsym$, $\gtNeg$ if $\pileTypeSingle_k = \ssym$
	} (D);
	\draw[->, myarrow] (B) -- node[midway, below, align=center] {
	$\gtDk$, 
	$\gtPos$ if $\pileTypeSingle_k \neq \qsym$,
	$\gtNeg$ if $\pileTypeSingle_k \neq \ssym$
	} (E);
\end{tikzpicture}
\caption{
Types of trajectories of $\litSeg$ on $\designAlign \backslash
%\arbpile^{n+1}$.
\arbpile^*$.
%Preparing a chain to test variable assignments.
}
\label{fig:test-variable}
\end{figure}

\begin{figure}[h!]
\centering
\begin{tikzpicture}

  % Define the four corner nodes
	\node[dotnode, label=left:{$AB$}] (A) at (0, 3) {};
	\node[dotnode, label=left:{$A \overline{B}$}] (B) at (0, 1) {};
	\node[dotnode, label=left:{$\overline A$}] (C) at (0, 0) {};
	
	\node[dotnode, label=right:{$A(B+C)$}] (D) at (3, 3) {};
	\node[dotnode, label=right:{$A \overline{B+C}$}] (E) at (3, 1) {};
	\node[dotnode, label=right:{$\overline A$}] (F) at (3, 0) {};

  % Draw an arrow from A to D with a label halfway
	\draw[->, myarrow] (A) -- node[midway, above, sloped] {} (D); 
	\draw[->, myarrow] (C) -- node[midway, above, sloped] {} (F); 
	\draw[->, myarrow] (B) -- node[midway, above] {	$C$	} (D);
	\draw[->, myarrow] (B) -- node[midway, above, sloped] {$\overline C$	} (E);
\end{tikzpicture}
\caption{Structure of a conditioned disjunction gate (``OR'').}
\label{fig:gate-lit}
\end{figure}
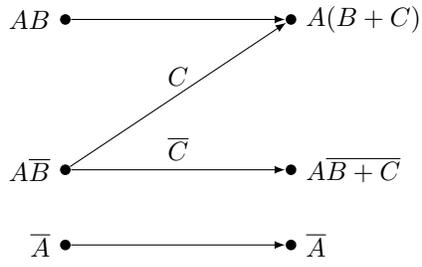

The words $\gtLastPos$, $\gtLastNeg$, and $\gtLastDk$ have the same basic function, but
they consume an additional measure of the chain.
Also they are used only as the final word in a clause,
because---like $\gtStart$---%
they change the testing ``mode'', or
the set of beats that hold significance in the following measure(s).
\begin{lemma}[Trajectories on last assignments]
\label{lemma:endlit-gadget}
\makecommand{\pt}{\pileTypeSingle}
The following hold
for every $n \geq 1$, for every chain $
\designAlign \backslash \pt_0 \ldots \pt_{n+1} \in 
\designAlign \backslash \arbpile^{n+2}$,
and each measure $k \in [n]-1$:
\begin{itemize}

\item
$\{ \gtLastPos, \gtLastNeg, \gtLastDk \}
\subset \bActd \backslash k \to \bEnd \backslash (k+1)$
if $\pt_{k+1} = \qsym$;
otherwise, they are in
$
%\{ \gtLastPos, \gtLastNeg, \gtLastDk \} \subset
\bActd \backslash k \to \bChainDisq \backslash (k+2)$%
;

\item
$\bNActd \backslash k \to \, \bEnd \backslash (k+1)$
contains
$\gtLastPos$ if $\pt_k \pt_{k+1} = \qsym\qsym$, and
$\gtLastNeg$ if $\pt_k \pt_{k+1} = \ssym\qsym$;
the complement are in
$\bNActd \backslash k \to \, \geq \bChainDisq \backslash (k+2)$,
including $\gtLastDk$ unconditionally%
;

\item
%clause disqualification states move to chain disqualification states, i.e.,
$\{ \gtLastPos, \gtLastNeg, \gtLastDk \}
\subset
\bClauseDisq \backslash k \to \, \geq \bChainDisq \backslash (k+2)$%
.

\end{itemize}
\end{lemma}
\begin{proof}
The proof is again exhaustive and can be verified
with the trajectory plots for $\gtLastPos$, $\gtLastNeg$, and $\gtLastDk$,
which are
in Figures~\ref{fig:endtest-actd}, Fig.~\ref{fig:endtest-nactd}, and~\ref{fig:endtest-disq} in Appendix~\ref{sec:more-traj}.
\end{proof}

The claims of Lemma~\ref{lemma:endlit-gadget} are graphed in Fig.~\ref{fig:endlit-arrows}.
They form a compound gate
whose first stage has the same essential structure as $\gtTest$ in Fig.~\ref{fig:test-variable}.
\begin{figure}[h!]
\centering
\begin{tikzpicture}

	\makecommand{\xOne}{4}
	\makecommand{\xTwo}{5.5}
	\makecommand{\xThree}{7}

	%\makecommand{\allThree}{$\gtLastPos, \gtLastNeg, \gtLastDk$}
	\makecommand{\allThree}{$\gtLastTest$ (any)}
	
  % Define the four corner nodes
    \node[dotnode, label=left:{$\bActd \backslash k$}] (A) at (0, 4) {};
    \node[dotnode, label=left:{$\bNActd \backslash k$}] (B) at (0, 2) {};
    \node[dotnode, label=below:{$\geq \bClauseDisq \backslash k$}] (C) at (0, 0) {};

	\node[dotnode] (AAA) at (\xOne, 4) {};  % actd
	\node[dotnode] (BBB) at (\xOne, 2) {};  % not actd
		
	\node[dotnode] (AA) at (\xTwo, 4) {};  % actd
	\node[dotnode] (BB) at (\xTwo, 0) {};  % not actd

    \node[dotnode,
	label=right:{$\bEnd \backslash (k+1)$}] (D) at (\xThree, 4) {};
    \node[dotnode,
	label=right:{$\geq \bChainDisq \backslash (k+2)$}] (E) at (\xThree, 0) {};

	\draw[->, myarrow] (A) -- node[midway, above, sloped] {
		\allThree
	} (AAA); 
	\draw[->, myarrow] (B) -- node[align=center, anchor=north west, pos=.7
	]{
	$\gtLastPos$ if $\pileTypeSingle_k = \qsym$,
	\\
	$\gtLastNeg$ if $\pileTypeSingle_k = \ssym$
	} (AAA);
	\draw[->, myarrow] (B) -- node[midway, below, align=center] {
	$\gtLastDk$,
	\\
	$\gtLastPos$ if $\pileTypeSingle_k \neq \qsym$,
	\\
	$\gtLastNeg$ if $\pileTypeSingle_k \neq \ssym$
	} (BBB);

	\draw[->, myarrow] (BBB) -- (BB);

	\draw[->, myarrow] (C) -- node[midway, above, sloped] {
		\allThree
	} (BB);
	\draw[->, myarrow] (BB) -- node[midway, above, sloped] {} (E);

	\draw[->, myarrow] (AAA) -- node[midway, above, sloped] {} (AA);
	\draw[->, myarrow] (AA) -- node[midway, above, sloped] {$x_{k+1} = \qsym$} (D);
	\draw[->, myarrow] (AA) -- node[midway, above, sloped] {$x_{k+1} \neq \qsym$} (E);
\end{tikzpicture}
\caption{Types of trajectories of $\lastLitSeg(\clause_j)$ on $\designAlign \backslash
% \arbpile^{n+2}$.}
\arbpile^*$.}
\label{fig:endlit-arrows}
\end{figure}

At last we have the tools needed for the proof of Lemma~\ref{lemma:clause-gadget}.
\begin{proof}[Proof of Lemma~\ref{lemma:clause-gadget}]
We start with trajectories from $\bChainDisq \backslash k_1$:
According to Lemma~\ref{lemma:start-clause},
after $\gtStart$ the state will be $\tilde k_0 \geq \bClauseDisq \backslash k_1$.
Then, according to Lemma~\ref{lemma:lit-gadget},
for each $i \in [n-1]$,
the application of $\litSeg_i(\clause_j)$ will transit from state $\tilde k_{i-1}$ to $\tilde k_{i}$, where
$\tilde k_i \geq \bClauseDisq \backslash (k_1 + i)$.
Finally, according to Lemma~\ref{lemma:endlit-gadget},
$\lastLitSeg_{n}(\clause_j)$ transits
from $\tilde k_{n-1} \geq \bClauseDisq \backslash (k_1 + n - 1)$
to $\tilde k_n
\geq \bChainDisq \backslash (k_1 + n + 1)$.

Next we consider trajectories from $\bStart \backslash k_1$
when $A_1 = \assignEmbed(\assignvec) \qsym$ for some $\assignvec \in SAT(\clause_j)$.
In this case
$\gtStart$ takes us to $\tilde k_0 = \bNActd \backslash k_1$.
Since ${\mathbf x} \in SAT(\clause_j)$, then in at least one place
either
$x_i = \top$ where $x_i \in \clause_j$, or
$x_i = \bot$ where $\neg x_i \in \clause_j$.
Then, respectively, either
$\assignEmbed(x_i) = \qsym$ and $\litSeg_i(\clause_j) = \gtPos$ or
$\assignEmbed(x_i) = \ssym$ and $\litSeg_i(\clause_j) = \gtNeg$.
%and .
We consider the first such $i$.
For each $i'=1 \ldots i-1$,
$\litSeg_{i'}(\clause_j)$ moves
from $\tilde k_{i'-1} = \bNActd \backslash (k_1 + i' - 1)$
to $\tilde k_{i'} = \bNActd \backslash (k_1 + i')$.
Then activation occurs:
If $i < n$, then
$\litSeg_i(\clause_j)$ moves
from $\tilde k_{i-1} = \bNActd \backslash (k_1 + i - 1)$
to $\tilde k_{i} = \bActd \backslash (k_1 + i)$.
Then for each $i' \in (i+1) \ldots n-1$,
$\litSeg_{i'}(\clause_j)$ moves
from $\tilde k_{i'-1} = \bActd \backslash (k_1 + i' - 1)$
to $\tilde k_{i'} = \bActd \backslash (k_1 + i')$.
Finally $\lastLitSeg_n(\clause_j)$ moves
from $\tilde k_{n-1} = \bActd \backslash (k_1 + n - 1)$
to $\bEnd \backslash (k_1 + n)$.
If $i = n$, then 
$\lastLitSeg_n(\clause_j)$ simply moves
from $\tilde k_{n-1} = \bNActd \backslash (k_1 + n - 1)$
to $\tilde k_{n} = \bEnd \backslash (k_1 + n)$.

Finally, we consider trajectories from
$\bStart \backslash k_1$
when $A_1$ is \emph{not} of the form $\assignEmbed(\assignvec) \qsym$
for some $\assignvec \in SAT(\clause_j)$.
Then either 
$A_1 = \assignEmbed(\assignvec) \qsym$ for ${\mathbf x} \notin SAT(\clause_j)$,
or else
$A_1 = \assignEmbed(\assignvec) \ssym$, \ie, it ends on a stack.
In the former case there is never any activation, and so
$\tilde k_{n-1} = \bNActd \backslash (k_1 + n-1)$, and
$\lastLitSeg_n(\clause_j)$ moves
to $\tilde k_n \geq \bChainDisq \backslash (k_1 + n + 1)$.
If $A_1$ ends on a stack, then regardless of whether activation occurred,
since $\tilde k_{n-1} \geq \bActd \backslash (k_1 + n-1)$
then $\lastLitSeg_n(\clause_j)$ moves
to $\tilde k_n \geq \bChainDisq \backslash (k_1 + n + 1)$.
\end{proof}

%\clearpage

\makecommand{\theInstances}{\designAlign \backslash \arbpile^{m_2} \backslash \qsym^{m_3}}
\makecommand{\theProbChain}{\chainScript\mhyphen\theInstances}

\section{Sort on $\theInstances$ is \NPHard/}
\label{sec:Q1}

Now we complete our first reduction of SAT to an accepting-chain problem,
specifically $\probAlias{1}'$,
and
thereby to the associated sort feasibility problem $\probAlias{1}$.
As shown in Table~\ref{table:all-problems},
their permissible sets are of the form $\theInstances$.

We begin by introducing a problem-specific embedding of formulae
as change profiles,
including a new trivial low-level gadget:

\begin{defn}
Given a CNF
formula $\clause$ of $m \geq 1$ clauses, we define
$$\formSeg_{\Rom{1}} (\clause)
\dfneq
\mainClause(\clause_1) \,
\gtNext \, \mainClause(\clause_2)
\ldots
\gtNext \, \mainClause(\clause_m)
,
$$
where
$\gtNext \dfneq a$.
\end{defn}

\begin{lemma}
\label{lemma:next-clause}
For every chain
$\designAlign \backslash x_0 \ldots x_n 
\in \designAlign \backslash \arbpile^{n+1}$,
$n \geq 1$,
for each measure $k \in [n]-1$,
if $x_k = \qsym$, then
$\gtNext \in \bEnd \backslash k \to \bStart \backslash (k+1)$.
\end{lemma}

\begin{proof}
The proof is again exhaustive and can be verified (easily)
with the trajectory plots for $\gtNext = a$ in Fig.~\ref{fig:next-proof}.
\end{proof}

\begin{figure}[h!]
	\centering
	\includegraphics[width=.75\linewidth]{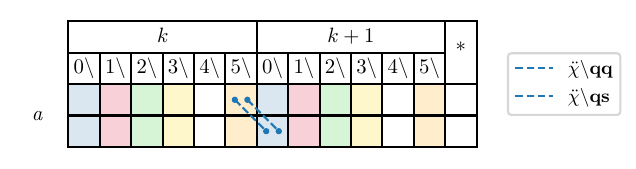}
\caption{
Trajectories of $\gtNext$
on chain automata
of the form $\designAlign \backslash x_0 \ldots x_n$
(only $x_k = \qsym$)
starting from $\bEnd \backslash k$.
}
\label{fig:next-proof}
\end{figure}

The figure demonstrates the function of the word $\gtNext$,
which is to move from the end of one measure to the start of the next one.
It provides us a key adaptation of Lemma~\ref{lemma:clause-gadget}.

\begin{lemma}
\label{lemma:clause-plus-next}
Given a CNF clause $\clause_j$ on $n \geq 1$ variables,
the following statements hold
for every
$k_1 \geq 0$, $k_2 \geq 1$,
and for every chain
$\designAlign \backslash A_0 A_1 A_2 \in \designAlign \backslash \arbpile^{k_1} \arbpile^{n+1} \arbpile^{k_2}$:
\begin{itemize}

\item
if $A_1 = \assignEmbed(\assignvec) \qsym$ for some $\assignvec \in SAT(\clause_j)$, then
$\mainClause(\clause_j) \, \gtNext
\in \bStart \backslash k_1 \to \bStart \backslash (k_1 + n + 1)$%
;

\item
otherwise, it is in
$\bStart \backslash k_1 \to \, \geq \bChainDisq \backslash (k_1 + n + 1)$;

\item
$\mainClause(\clause_j) \, \gtNext
\in \bChainDisq \backslash k_1 \to \, \geq \bChainDisq \backslash (k_1 + n + 1)$.
%, unconditionally.

\end{itemize}
\end{lemma}

\begin{proof}
The proof is a fairly straightforward combination of Lemmas~\ref{lemma:clause-gadget} and~\ref{lemma:next-clause}.
\end{proof}
We omit the graph of the lemma's claims, which would be nearly identical to that of Fig.~\ref{fig:traj-clause},
only with $\bStart \backslash (k_1 + n + 1)$ as the top-right corner.
(Which makes it chainable.)

Now we may consider trajectories of the formula gadget as a whole.
\begin{lemma}
\label{lemma:formula-gadget}
Given a formula $\clause$
of $m$ CNF clauses $\{ \clause_j \}_{j=1}^m$ on $n \geq 1$ variables,
the following statements hold
on all
$\designAlign \backslash A_1 \ldots A_{m} x \in \designAlign \backslash \left( \arbpile^{n+1} \right)^{m} \arbpile$%
:
\begin{enumerate}

\item
if for each $j \in [m]$,
$A_{j} = \assignEmbed(\assignvec_j) \qsym$
for some $\assignvec_j \in SAT(\clause_j)$,
then
$\formSeg_\Rom{1}(\clause)
\in \bStart \backslash 0
\to
\bEnd \backslash (m(n+1) - 1)
$%
;

\item
otherwise, 
$\formSeg_\Rom{1}(\clause)
\in \bStart \backslash 0
\to \, \geq \bChainDisq \backslash m(n+1)$.

%\item
%$\bChainDisq \backslash 0 \backslash 0 
%\to \geq \bChainDisq \backslash 0 \backslash m$
%unconditionally.\todo{do we care though?}
\end{enumerate}
\end{lemma}

\begin{proof}
We start with the first case.
Due to Lemma~\ref{lemma:clause-plus-next},
for each $j \in [m-1]$,
the segment $\clauseSeg(\clause_j) \, \gtNext$
moves
from state $\tilde k_{j-1} = \bStart \backslash (j-1)(n+1)$
to state $\tilde k_j = \bStart \backslash j(n+1)$.
Then,
due to Lemma~\ref{lemma:clause-gadget},
$\clauseSeg(\clause_m)$
moves
from state $\tilde k_{m-1} = \bStart \backslash (m-1)(n+1)$
to state $\tilde k_m = \bEnd \backslash (m(n+1)-1)$.

In the second case,
there is at least one clause $\clause_j$ for which
$A_{j} \neq \assignEmbed(\assignvec_j) \qsym$ for any $\assignvec_j \in SAT(\clause_j)$.
Then $\clauseSeg(\clause_j)$ moves
from state $\tilde k_{j-1} \geq \bStart \backslash (j-1)(n+1)$
to $\tilde k_j \geq \bChainDisq \backslash j(n+1)$.
Then for each $j' \in (j+1) \ldots m$ (none if $j = m$),
%the segment 
$\gtNext \, \clauseSeg(\clause_j)$ moves from 
state $\tilde k_{j'-1} \geq \bChainDisq \backslash (j'-1)(n+1)$
to state
$\tilde k_{j'} \geq \bChainDisq \backslash j' (n+1)$.
\end{proof}

\begin{lemma}
\label{lemma:red1-language}
Given a formula $\clause$ of $m\geq 1$ CNF clauses on $n \geq1$ variables,
a chain $\pi \in \designAlign \backslash \arbpile^{n+1} \backslash \qsym^m$ accepts
the change profile $\formSeg_{\Rom{1}}(\clause)$
if and only if
$\pi = \designAlign \backslash \assignEmbed(\assignvec) \qsym \backslash \qsym^m$ for some $\assignvec \in SAT(\clause)$.
\end{lemma}

\begin{varscope}
\makecommand{\word}{A}
\begin{proof}
By definition every chain $\pi \in \designAlign \backslash \arbpile^{n+1} \backslash \qsym^m$ can be expressed as
$\pi = \designAlign \backslash A \backslash \qsym^m$ for some $A \in \arbpile^{n+1}$.
Then by Lemma~\ref{lemma:concat-lemma},
$\pi$ is equal to the concatenation
of $m$ copies of $\designAlign \backslash A \backslash \qsym = \designAlign \backslash A$, \ie,
$\pi
= \left(\designAlign \backslash A\right) \ldots \left(\designAlign \backslash A\right)
= \designAlign \backslash A \ldots A$.
Being one measure shorter than the chains 
of Lemma~\ref{lemma:formula-gadget},
$\pi$ cannot reach 
the state $\bChainDisq \backslash m(n+1)$
of the lemma's second case;
it reaches the terminal state instead.
Therefore
$\pi$ accepts the formula gadget if and only if
the condition of the first case holds.
That is exactly the condition of the present lemma:
$\word = \assignEmbed(\assignvec) \qsym$ for some $\assignvec$
such that $\assignvec \in SAT(\clause_j)$ for all $j$,
\ie, $\assignvec \in SAT(\clause)$.
\end{proof}
\end{varscope}

\makecommand{\theProblem}{\questionSet_{\Rom{1}}}
\makecommand{\theFormula}{\formSeg_{\Rom{1}}(\clause)}

\begin{corollary}
\label{cor:reduction-1}
For every instance $\questionCons_\SAT(\clause) \in \SAT$,
with formula $\clause$ of $m$ clauses on $n$ variables,
the instance
$\questionCons_\chainScript\left(
\theFormula,
\designAlign \backslash \arbpile^{n+1} \backslash \qsym^m
\right)$
is in $\theProblem'$
and has the same answer.
\end{corollary}
\begin{proof}
We can see by inspection that
the latter question is in $\theProblem'$ (Table~\ref{table:all-problems}).
We obtain the corollary from Lemma~\ref{lemma:red1-language} because
the permissible schedules 
$\designAlign \backslash \arbpile^{n+1} \backslash \qsym^m$
%by inspection,
contains the embedding
$\designAlign \backslash \assignEmbed(\assignvec) \qsym \backslash \qsym^m$
of every $\assignvec \in \{\top, \bot\}^n$.
Therefore, an accepting chain exists if and only if $\clause$ is satisfiable.
\end{proof}

\begin{lemma}
%Corollary~\ref{cor:reduction-1} allows us to prove that 
$\theProblem$ and $\theProblem'$ are \NPHard/.
\end{lemma}
\begin{proof}
Corollary~\ref{cor:reduction-1}
leaves only to prove that
its reduction is polynomial-time:
The output instance data includes the change profile $\theFormula$
and parameters $m$ and $n$ of the permissible set.
(The latter are negligible in transcription length.)
The change profile is composed of $m$ instances of $\gtClause$,
each of length $O(n)$, and so
can be transcribed in $O(mn)$ time.
The transcription length $|\clause|$ of a CNF formula $\clause$ 
of $m$ clauses on $n$ variables
is at least proportional to $m$ or $n$, whichever is larger,
so $O(mn) \subseteq O(|\clause|^2)$.
\end{proof}

%\clearpage
\section{Sort on $\designAlign \backslash \arbpile^{m_2} \backslash \arbpile^{m_3}$ is \NPHard/}

\label{sec:Q2}

Our first hardness proof, just completed,
depended critically on restricting the last round of shuffle
to be on queues only.
It allowed our proofs to exploit the fact that our shuffle was on the concatenation of $m$ copies of 
$\designAlign \backslash A \backslash \qsym = \designAlign \backslash A$
for a particular $A \in \arbpile^{n+1}$.
In this section we present a reduction that removes this constraint,
allowing the types of piles in the third last round of shuffle to be chosen arbitrarily.

Our next reduction requires a new gadget, which we present below.
\makecommand{\clauseType}[1]{\clauseSeg_\Rom{#1}}
\begin{defn}
\label{def:formula2}
%\makecommand{\theGadget}[1]{\gtQ \, \clauseSeg(#1)}
\makecommand{\theGadget}[1]{\clauseType{2}(#1)}
Let
\[
\formSeg_{\Rom{2}}(\clause) \dfneq 
\theGadget{\clause_1}
\,
\gtNext
\,
\theGadget{\clause_2}
\ldots
\gtNext
\,
\theGadget{\clause_m}
,
\]
where

$\clauseSeg_\Rom{2}(\clause_j) = \gtQ \, \clauseSeg(\clause_j)$,
with
$\gtQ \dfneq dadaaaadaaaaddddaaaa$.

\end{defn}
The word $\gtQ$ has the function
of forcing
its measure to be $\designAlign \backslash \qsym$
in accepting chains,
rather than $\designAlign \backslash \ssym$:
\begin{lemma}%[$\gtQ$]
\label{lemma:traj-force-q}

\makecommand{\pt}{\pileTypeSingle}

The following statements hold 
for every chain $
\designAlign \backslash \pt_0 \ldots \pt_{n} \in 
\designAlign \backslash \arbpile^{n+1}$,
$n\geq 1$,
for each measure $k \in [n]-1$:
\begin{itemize}
\item
if $x_k = \qsym$, then
$\gtQ \in \bStart \backslash k \to \bStart \backslash (k+1)$%
;

\item otherwise,
$\gtQ \in \bStart \backslash k \to \geq \bChainDisq \backslash (k+1)$;

\item
$\gtQ \in \bChainDisq \backslash k \to \, \geq \bChainDisq \backslash (k+1)$, unconditionally.
\end{itemize}
\end{lemma}

\begin{proof}
The proof is again exhaustive and can be verified
with the trajectory plots for $\gtQ$
in Fig.~\ref{fig:force-q-traj}.
\end{proof}
\begin{figure}[h!]
	\makecommand{\theFrac}{.49}
	\centering
	\begin{subfigure}[t]{\theFrac\linewidth}
		\includegraphics[width=\linewidth]{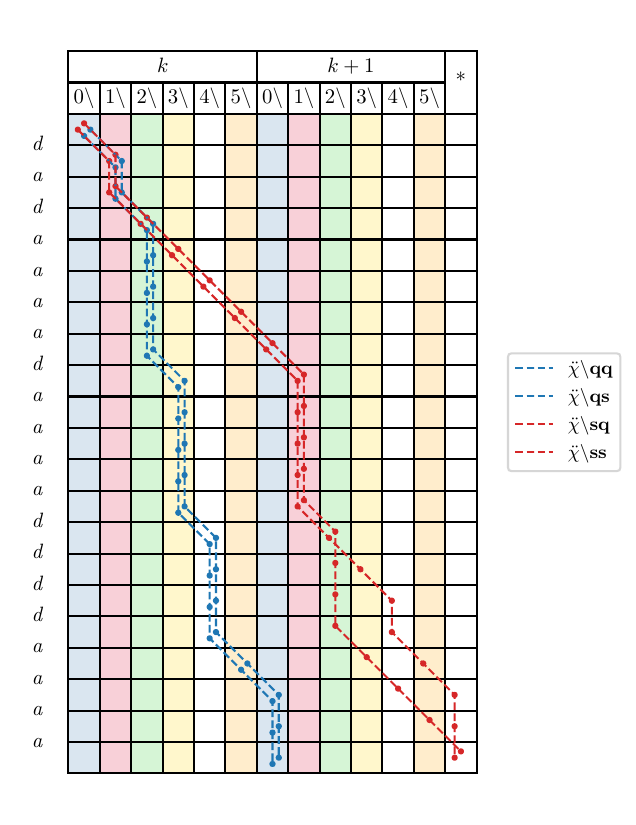}
		\caption{$\gtQ$ starting from $\bStart \backslash k$}
	\end{subfigure}
	\begin{subfigure}[t]{\theFrac\linewidth}
		\includegraphics[width=\linewidth]{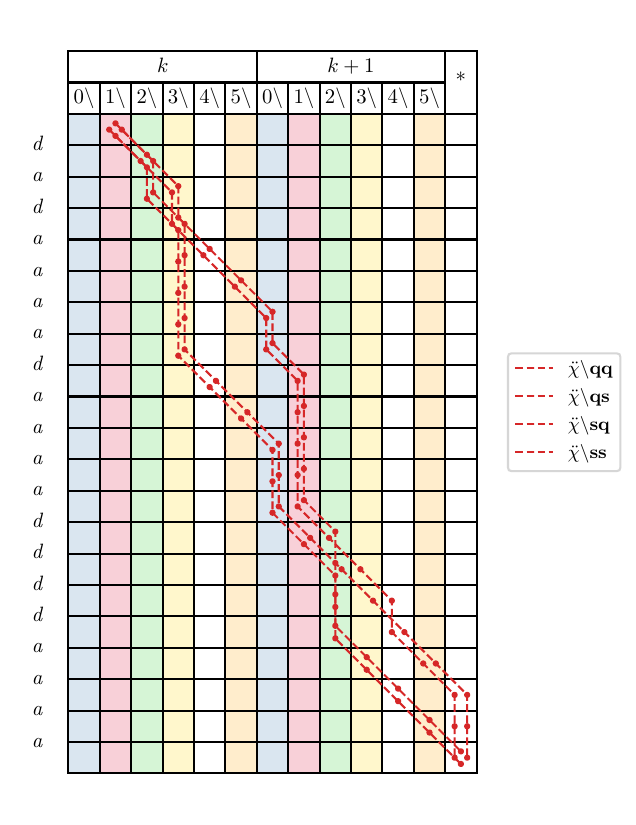}
		\caption{$\gtQ$ starting from $\bChainDisq \backslash k$}
	\end{subfigure}
\caption{
Trajectories of $\gtQ$
from two starting positions
on measure $k$
of the chain automata in $\designAlign \backslash 
%\arbpile^{n+1}$.
\arbpile^*$.
}
\label{fig:force-q-traj}
\end{figure}

\makecommand{\theFormula}{\formSeg_{\Rom{2}}(\clause)}
\makecommand{\theProblem}{\questionSet_\Rom{2}}

Lemma~\ref{lemma:traj-force-q} provides for a minor adaptation of Lemma~\ref{lemma:formula-gadget}
for $\formSeg_{\Rom{2}}$.
\begin{lemma}
\label{lemma:formula2}
Given a formula $\clause$
of $m$ CNF clauses $\{ \clause_j \}_{j=1}^m$
on $n \geq 1$ variables,
the following statements hold
for every chain
$\designAlign \backslash A_1 \ldots A_{m} x \in \designAlign \backslash \left( \arbpile^{n+2} \right)^{m} \arbpile$:
\begin{enumerate}
\item
if for each $j \in [m]$,
$A_j = \qsym \assignEmbed(\assignvec_j) \qsym$
for some $\assignvec_j \in SAT(\clause_j)$,
then
$\theFormula \in \bStart \backslash 0 \to \bEnd \backslash (m(n+2) - 1)$;

\item
otherwise,
$\theFormula \in \bStart \backslash 0 \to \, \geq \bChainDisq \backslash m(n+2)$.

\end{enumerate}
\end{lemma}
We leave the proof of the lemma as an exercise for the reader.
The logic is essentially the same as for Lemma~\ref{lemma:formula-gadget},
but 
with the placements of $\gtQ$ ensuring that
$A_j = \qsym \assignEmbed(\assignvec_j) \qsym$
rather than
$A_j = \assignEmbed(\assignvec_j) \qsym$
previously.
(We will soon exploit the symmetry
of having a $\qsym$ on each side
of an assignment embedding.)

For the next result we require an important property of factorable type schedules:
\begin{lemma}[Equivalence classes]
\label{lemma:equivs}
Given any two type schedules $\pileType_1$ and $\pileType_2$,
\[
\pileType_1 \backslash \pileType_2
=
\pileType_1^\dual \backslash {\invert \pileType_2}
.
\]
\end{lemma}
\begin{proof}
We offer a simple, purely algebraic proof based on previously-discussed properties of the two types of inverses:
\begin{align*}
\pileType_1 \backslash \pileType_2
&= \pileType_1 \backslash \qsym \backslash \pileType_2
= \pileType_1 \backslash (\ssym \backslash \ssym) \backslash \pileType_2
= (\pileType_1 \backslash \ssym) \backslash (\ssym \backslash \pileType_2)
= \pileType_1^\dual \backslash \invert \pileType_2
.
\end{align*}
\end{proof}

\begin{varscope}
\makecommand{\word}{A}
\begin{lemma}
\label{lemma:red2-language}
Given a formula $\clause$
of $m \geq 1$ CNF clauses
on $n \geq 1$ variables,
a chain $\pi \in \designAlign \backslash \arbpile^{n+2} \backslash \arbpile^m$ accepts
the change profile $\formSeg_{\Rom{2}}(\clause)$
if and only if
$\pi = \designAlign \backslash \qsym \assignEmbed(\assignvec) \qsym \backslash \qsym^m$ for some $\assignvec \in SAT(\clause)$.
\end{lemma}
\begin{proof}
By definition every chain $\pi \in \designAlign \backslash \arbpile^{n+2} \backslash \arbpile^m$ can be expressed by
$\pi = \designAlign \backslash A \backslash x_0 \ldots x_{m-1}$
for some $\word \in \arbpile^{n+2}$, $x_0 \ldots x_{m-1} \in \arbpile^m$.
Without loss of generality $x_0 = \qsym$, since otherwise, by Lemma~\ref{lemma:equivs},
$\pi$ can be equivalently expressed as $\designAlign \backslash A^\dual \backslash \invert x_0 \ldots \invert x_{m-1}$.
By Lemma~\ref{lemma:concat-lemma},
$\pi = 
\left( \designAlign \backslash A \backslash x_0 \right)
\ldots
\left( \designAlign \backslash A \backslash x_{m-1} \right)
$,
where
$\designAlign \backslash A \backslash x_j = \designAlign \backslash A$ if $x_j = \qsym$, otherwise, when $x_j = \ssym$,
$\designAlign \backslash A \backslash x_j = \designAlign \backslash A^\dual$.

%Similarly as in the proof of Lemma~\ref{lemma:red1-language},
Only the first case in Lemma~\ref{lemma:formula2} ends in
%a state strictly less than $0 \backslash 0 \backslash m = 6(n+2)m$,
%which is the non-accepting terminal state on the truncation represented by the DFA of $\pi$.
an accepting state on truncated $\pi$.
Therefore, the
acceptance criteria
become:
\begin{itemize}
\item wherever $x_j = \qsym$, $A = \qsym \assignEmbed(\assignvec) \qsym$ for some $\assignvec \in SAT(\clause_j)$; and

\item wherever $x_j = \ssym$, $A^\dual = \qsym \assignEmbed(\assignvec) \qsym$ for some $\assignvec \in SAT(\clause_j)$;
%---i.e., 
equivalently,
$A = \ssym \assignEmbed^\dual(\assignvec) \ssym$.
\end{itemize}
Clearly it is impossible for $A$ to be of 
%a mixture of the pile types among $x_0 \ldots x_{m-1}$ to satisfy 
both forms simultaneously.
Then since $x_0 = \qsym$,
acceptance can only occur if
%they are all $\qsym$.
$x_1 \ldots x_{m-1}$ % = \qsym^{m-1}$.
are all $\qsym$ as well.
In such case
the first condition is used uniformly for all $j$,
which obtains the lemma.
\end{proof}
\end{varscope}

As before, a forthcoming corollary demonstrates that $\probAlias{2}$ is \NPHard/ by reduction from \SAT.
We will omit the proofs of both results, which are almost identical to their $\probAlias{1}$ analogues.
\begin{corollary}
\label{cor:reduction-2}
For every instance $\questionCons_\SAT(\clause) \in \SAT$,
with formula $\clause$ of $m$ clauses on $n$ variables,
the instance
$\questionCons_\chainScript\left(
\theFormula,
\designAlign \backslash \arbpile^{n+2} \backslash \arbpile^m
\right)$
is in $\theProblem'$
and has the same answer.
\end{corollary}
\begin{lemma}
%Corollary~\ref{cor:reduction-1} allows us to prove that 
$\theProblem$ and $\theProblem'$ are \NPHard/.
\end{lemma}

%\clearpage
\section{%
Problem~\ref{prob:variable-round}
%Variable-round pile shuffle sort 
is \NPHard/ (via $\arbpile^6 \backslash \arbpile^{m_2} \backslash \arbpile^{m_3}$)}

\label{sec:variable-round-feasibility-complexity}
%\label{sec:Q3}

Finally, we relax the first-round constraint to show that
sort on permissible sets
of the form
$\arbpile^6 \backslash \arbpile^{m_2} \backslash \arbpile^{m_3}$
%with $m_2$ and $m_3$ as parameters.
($\probAlias{3}$)
is \NPHard/.
It follows then that
%variable-round sort feasibility 
%(Prob.~\ref{prob:variable-round}) 
Prob.~\ref{prob:variable-round}
is \NPHard/ also.

Unlike our previous reductions, which have remained more or less practical
in terms of the size of instances generated,
our next reduction is relatively impractical.
Nevertheless,
it demonstrates a polynomial-time reduction from SAT,
ultimately proving Theorem~\ref{thm:variable-round-hard}.

We begin by discussing the trajectories of a new gadget,
\[
\gtAlign \dfneq a^{24} d a^{18} d a^{18} d a^{18} d^{21} a d^{5} a^{17} d a^{5} d^{11} a d^{5} a d^{4} a d^{4} a^{4}
.
\]
The function of $\gtAlign$ is to
consume large segments of
(\ie, ``\emph{penalize}'')
permissible chains not contained by $\designAlign \backslash \arbpile^*$,
%(\ie, non- $\designAlign$-aligned chains),
so that they cannot accept any formula embedding.
Meanwhile, it must \emph{not} penalize any of the properly $\designAlign$-aligned assignment-embedding chains.

Unfortunately,
a single instance of $\gtAlign$ can only guarantee a fairly small penalty,
of a single extra beat per two measures consumed:

\begin{lemma}
\label{lemma:align-singleton}
\makecommand{\pt}{\pileTypeSingle}

The following statements hold
for every chain $\pi \in 
\arbpile^6 \backslash \arbpile^{n+2}$,
$n \geq 1$,
for each measure $k \in [n]-1$:
\begin{itemize}
\item
($\designAlign$-aligned)
if $
\pi = \designAlign \backslash \pt_0 \ldots \pt_{n+1}$, for $x_0 \ldots x_{n+1} \in \arbpile^{n+2}$:
\begin{itemize}
\item
if $x_k x_{k+1} = \qsym\ssym$,
then
$\gtAlign \in \bStart \backslash k \to \bStart \backslash (k + 2)$;

\item
otherwise,
$\gtAlign \in \bStart \backslash k \to \, \geq \bChainDisq \backslash (k + 2)$;

\item
$\gtAlign \in \bChainDisq \backslash k \to \, \geq \bChainDisq \backslash (k + 2)$, unconditionally;

\end{itemize}

\item
otherwise,
for all $k' \in [6] - 1$,
$\gtAlign \in k' \backslash k \to \, \geq (k' + 1) \backslash (k + 2)$.

\end{itemize}

\end{lemma}
The proof of the lemma is exhaustive
as for all of our other low-level gadgets, only now
due to the extra degrees of freedom in the first round,
there are thousands of trajectory classes to check.
Our proof-by-plot approach is not helpful in this case,
especially given
% for a variety of reasons, including
the word's length.
We omit the proof, which the author verified with computer assistance;
all classes can be checked within a few seconds on an M1 MacBook Pro.

We rely on repetition of the word $\gtAlign$ to produce a penalty
that is sufficiently large for our purposes.
In the next lemma
we partition chains 
%conceptually
logically
into four segments:
the \emph{prefix}, the \emph{guard}, the \emph{payload}, and the \emph{suffix}.
\begin{lemma}
\label{lemma:aligned-word}
The following statements hold
for any $k_1 \geq 0$, $n, k_2 \geq 1$, and
for every 
$\pi \in
\arbpile^6 \backslash \arbpile^{k_1} \arbpile^{2(6n+1)} \arbpile^n \arbpile^{k_2}$:

\begin{itemize}
%\makecommand{\nextBlock}{\left( k_1 + 2(6n+1) \right)}
\makecommand{\nextBlock}{( k_1 + 2(6n+1) )}

\item
($\designAlign$-aligned) if $\pi$ can be written as $\designAlign \backslash A_0 A_1 A_2 A_3$:
\begin{itemize}

\item
if $A_1 = (\qsym\ssym)^{6n+1}$, then
${\gtAlign}^{6n+1} \in \bStart \backslash k_1 \, \to \bStart \backslash \nextBlock$;
%\ie, from the start of $A_1$ to the start of $A_2$;

\item otherwise
${\gtAlign}^{6n+1} \in \bStart \backslash k_1 \, \to \, \geq \bChainDisq \backslash \nextBlock$;
%\ie, disqualified on the same measure;

\item
${\gtAlign}^{6n+1} \in \bChainDisq \backslash k_1 \, \to \, \geq \bChainDisq \backslash \nextBlock$;
%, unconditionally;

\end{itemize}

\item otherwise,
${\gtAlign}^{6n+1}
\in \bStart \backslash k_1 \, \to \, \geq \bChainDisq \backslash ( k_1 + 2(6n+1) + n )$.
%skipping $A_2$ entirely.

\end{itemize}
\end{lemma}
\begin{proof}
For each of the three aligned cases,
the claims are fairly easy to reason by repeated application of Lemma~\ref{lemma:align-singleton}.
For the final case we observe that
starting from $\bStart \backslash k_1$ and 
applying $6n+1$ copies of $\gtAlign$,
we find ourselves at a state
\begin{align*}
\geq \, (\bStart+6n+1) \backslash (k_1 + 2(6n+1))
&= (\bStart+1) \backslash \left( k_1 + 2(6n+1) + n \right)
\\ &= \bChainDisq \backslash \left( k_1 + 2(6n+1) + n \right)
.
\end{align*}
\end{proof}

The lemma demonstrates that if the candidate chain is aligned
and the guard has the necessary form, then
the trajectory of the input sequence moves 
from the start of the guard to the start of the payload.
If the guard is not of the correct form, however, then
the trajectory ends in a disqualified state relative to the payload instead.
Meanwhile, the trajectory
from a disqualified state
on the guard of any aligned chain
reaches a disqualified state with respect to the payload,
irrespective of the guard's composition.
Finally, if the chain is not $\designAlign$-aligned at all, then
a trajectory from the guard skips the payload entirely, landing in the suffix instead.

By prepending a sufficiently strong penalty to our clause gadget,
we can obtain our new reduction's version of Lemma~\ref{lemma:clause-gadget}:
\begin{lemma}
\label{lemma:force-align-block}

\makecommand{\theWord}{{\gtAlign}^{n''} \gtClause(\clause_j)}

Given a CNF clause $\clause_j$
on $n \geq 1$ variables,
let
$n' = n+1$ and
$n'' = 6n'+1$.
The following statements hold
for every
%$k_1 \in \nats_0$, $k_2 \in \nats_1$,
$k_1 \geq 0$, $k_2 \geq 1$,
and for every chain
$\pi \in
\arbpile^6 \backslash \arbpile^{k_1} \arbpile^{(2n'' + n')} \arbpile^{k_2}$:
\begin{itemize}
\item
($\designAlign$-aligned)
if $\pi$ can be written as $\designAlign \backslash A_0 A_1 A_2$,

\begin{itemize}

\item
if $A_1 = (\qsym\ssym)^{n''} \, \assignEmbed(\assignvec_j) \, \qsym$
for some $\assignvec_j \in \SAT(\clause_j)$, then
$\theWord \in \bStart \backslash k_1 \, \to \bEnd \backslash \left( k_1 + 2n'' + n' - 1 \right)$;
%\ie, from the start of $A_1$ to the very end of it;

\item otherwise
$\theWord \in \bStart \backslash k_1 \, \to \, \geq \bChainDisq \backslash (k_1 + 2n''+n')$;
%\ie, disqualified on the first measure of $A_2$ instead;

\item
$\theWord \in \bChainDisq \backslash k_1 \, \to \, \geq \bChainDisq \backslash (k_1 + 2n''+n')$, unconditionally;

\end{itemize}

\item otherwise,
$\theWord \in \bStart \backslash k_1 \, \to \, \geq \bChainDisq \backslash (k_1 + 2n''+n')$.

\end{itemize}

\end{lemma}

\begin{proof}

\makecommand{\gtOne}{\gtAlign^{n''}}
\makecommand{\gtTwo}{\gtClause(\clause_j)}
\makecommand{\stateVar}{{\tilde k}}

We start with the case that $\pi$ is unaligned,
in which by Lemma~\ref{lemma:aligned-word},
the word $\gtOne$ is itself in
$\bStart \backslash k_1 \, \to \, \geq \bChainDisq \backslash (k_1 + 2n''+n')$.
Then so are any of its continuations, \eg, by $\gtTwo$ in the present lemma.

Next, we examine the case that $\pi$ is aligned and the word is applied from
$\stateVar_0 \geq \bChainDisq \backslash k_1$.
Then
by Lemma~\ref{lemma:aligned-word},
$\gtOne$ brings the trajectory to
$\stateVar_1 \geq \bChainDisq \backslash (k_1 + 2n'')$, and
by Lemma~\ref{lemma:clause-gadget},
$\gtTwo$ brings it from there to
$\stateVar_2 \geq \bChainDisq \backslash (k_1 + 2n'' + n')$.

Finally, we examine the case that $\pi$ is aligned and the word is applied from
$\stateVar_0 = \bStart \backslash k_1$.
We partition $A_1$ as $A_{11}A_{12} \in \arbpile^{2n''} \arbpile^{n'}$,
so that
$A_1 = (\qsym\ssym)^{n''} \, \assignEmbed(\assignvec_j) \, \qsym$
if and only if
$A_{11} = (\qsym\ssym)^{n''}$ and $A_{12} = \assignEmbed(\assignvec_j) \, \qsym$.
If both are true, then
$\gtOne$ moves from $\stateVar_0$ to $\stateVar_1 = \bStart \backslash (k_1 + 2n'')$,
and
$\gtTwo$ moves from there to $\stateVar_2 = \bEnd \backslash (k_1 + 2n'' + n' - 1)$.
Otherwise, 
%We let the reader verify that, otherwise, the trajectory reaches a disqualifying state instead.
we let the reader verify that the trajectory reaches a disqualifying state instead.
\end{proof}

\begin{lemma}
%
%\makecommand{\clausePrefix}{{\gtAlign}^{n''}}
\makecommand{\clausePrefix}{{\gtAlign}^{6n+7}}
Given a formula $\clause$
of $m \geq 1$ CNF clauses
% $\{ \clause_j \}_{j=1}^m$
on $n \geq 1$ variables,
a chain $\pi \in 
\arbpile^6 \backslash \arbpile^{13n + 15} \backslash \arbpile^m$
accepts the change profile
$$\formSeg_{\Rom{3}}(\clause)
\dfneq
\gtClause_\Rom{3}(\clause_1)
\, \gtNext \, \gtClause_\Rom{3}(\clause_2)
\ldots
\gtNext \, \gtClause_\Rom{3}(\clause_m)
,
$$
with
$\gtClause_\Rom{3}(\clause_j)
\dfneq 
\clausePrefix
 \, \gtClause(\clause_j)
$,
if and only if
$\pi = \designAlign \backslash (\qsym\ssym)^{6n+7} \, \assignEmbed(\assignvec) \, \qsym \backslash \qsym^m$
for some $\assignvec \in SAT(\clause)$.
\end{lemma}
\begin{proof}[Proof Sketch]
The lemma is proved with the same technique as the proof of Lemma~\ref{lemma:red2-language},
but using Lemma~\ref{lemma:force-align-block}
instead of Lemma~\ref{lemma:formula2}.
%Critically, we observe that the assignment encoding chains of the lemma all start and end with $\qsym$.
\end{proof}

\makecommand{\theFormula}{\formSeg_{\Rom{3}}(\clause)}
\makecommand{\theProblem}{\questionSet_\Rom{3}}

\begin{corollary}
\label{cor:reduction-3}
For every instance $\questionCons_\SAT(\clause) \in \SAT$,
with formula $\clause$ of $m$ clauses on $n$ variables,
the instance
$\questionCons_\chainScript\left(
\theFormula,
\arbpile^6 \backslash \arbpile^{13 n + 15} \backslash \arbpile^m
\right)$
is in $\theProblem'$
and has the same answer.
\end{corollary}
As before, we omit the stereotypical proof of the corollary.
\begin{lemma}
$\theProblem$ and $\theProblem'$ are \NPHard/.
\end{lemma}
\begin{proof}
Although impractically large from the viewpoint of constructing (or sorting) a physical deck of cards,
$\theFormula$ can still be transcribed in $O(mn)$ time.
\end{proof}

%
%Moreover,
As previously observed,
on this occasion
the proof of Theorem~\ref{thm:variable-round-hard} follows directly:
\begin{proof}[Proof of Theorem~\ref{thm:variable-round-hard}]
Because $\theProblem$ is \NPHard/ and a strict subset of Prob.~\ref{prob:variable-round},
then Prob.~\ref{prob:variable-round} is also \NPHard/.
\end{proof}

%\clearpage

\makecommand{\alignedStaticRepeatedSort}{\designAlign^{k} \backslash \arbpile^{6k} \backslash \arbpile^{6k}}

\section{On Conjecture~\ref{conj:repeated-hard}}
%$\alignedStaticRepeatedSort$---a \NPHard/ static-table sort problem}
%\section{Sort on $\alignedStaticRepeatedSort$---a static table---is \NPHard/}

\label{sec:on-conj}

\makecommand{\staticRepeatedSort}{\arbpile^m \backslash \ldots \backslash \arbpile^m}

Finally
%in the prequel
we have placed variable-round Dealer's choice sort feasibility (Prob.~\ref{prob:variable-round})
in the \NPHard/ complexity class.
While interesting on its own,
this result falls short of fully addressing our motivating problem
of sorting a deck of cards on a physical table.
In such scenarios the problem variables include the deck of cards and the surface area of the table,
where
the latter does not typically change between rounds of shuffle.
Therefore,
we are most interested in the complexity of repeated-round sort feasibility (Prob.~\ref{prob:repeated-round}), \ie,
sortation by $\staticRepeatedSort$
with the number of rounds and the number of piles $m$ as problem parameters.

Unfortunately we leave it as an open problem, with conjecture only,
that repeated-round sort feasibility is \NPHard/.
However, in this section we consider a partial proof strategy
which adapts the machinery developed toward variable-round sort feasibility, so that
each round of shuffle has the same number of piles.
Sadly the strategy is incomplete, but
it is enough to demonstrate that
sort on $\alignedStaticRepeatedSort$
($\probAlias{5}$ of Table~\ref{table:all-problems})
is \NPHard/.
$\probAlias{5}$ is \emph{nearly} repeated---%
each round has $6k$ piles---%
but has the forced types $\designAlign^k$ in the first round.

To begin let us recall that 
we may reduce any instance of $\SAT$---%
a CNF formula 
of $m$ clauses on $n$ variables---%
to sort feasibility on $\designAlign \backslash \arbpile^{n+2} \backslash \arbpile^m$
in polynomial time
(Corollary~\ref{cor:reduction-2}).
It is fairly straightforward then to augment the second and third rounds of sort:
We may easily increase the number of piles in the second round
by introducing arbitrarily many unused dummy variables.
Similarly, we may increase the number of piles in the third round
by introducing ineffectual clauses, \eg,
adding additional copies of any clauses already in the formula.
All we need then is a strategy to augment the first round.

Suppose we widen the first round by repetition, \ie, using $\designAlign^k$ for some $k \geq 1$.
As we do this,
a key insight is that
$$\designAlign^k \backslash x_0 \ldots x_{n-1}
=
\underbrace{
\left( \designAlign \backslash x_0 \right) \ldots \left( \designAlign \backslash x_0 \right)
}_{
\text{$k$ times}
}
\quad \ldots \quad
\underbrace{
\left( \designAlign \backslash x_{n-1} \right) \ldots \left( \designAlign \backslash x_{n-1} \right)
}_{
\text{$k$ times}
}
.
$$
The transitions 
between the indicated $k$-measure blocks present the same sequence of test sites we have used previously, however,
we must now augment our change profile to consume the regular repetitions of each measure between those transitions.

We require one more low-level gadget for the task:
Let
$\gtPass \dfneq daaaadadadadddda$.
The purpose of this gadget is to consume exactly one measure from the $\bStart$ or $\bEnd$ beats
on any $\designAlign$-aligned chain,
and \emph{at least} one measure from $\bChainDisq$:
(Note $\gtDk$ suits the same purpose for $\bActd$ or $\bNActd$, and $\bClauseDisq$, respectively.)
\begin{lemma}
The following statements hold on all
$\pi 
%\in \designAlign \backslash x_0 \ldots x_{n} 
\in \designAlign \backslash \arbpile^{n+1}
$, for each measure $k \in [n]-1$:
\begin{itemize}

\item
$\gtPass \in \bStart \backslash k \to \bStart \backslash k+1$;

\item
$\gtPass \in \bChainDisq \backslash k \to \, \geq \bChainDisq \backslash k+1$; and 

\item
$\gtPass \in \bEnd \backslash k \to \bEnd \backslash k+1$.

\end{itemize}
\end{lemma}
\begin{proof}
The proof is exhaustive and can be verified by the diagrams in Fig.~\ref{fig:traj-pass}.
\end{proof}

\begin{figure}[h!]
	\centering
	\begin{subfigure}[t]{.3\linewidth}
		\includegraphics[width=\linewidth]{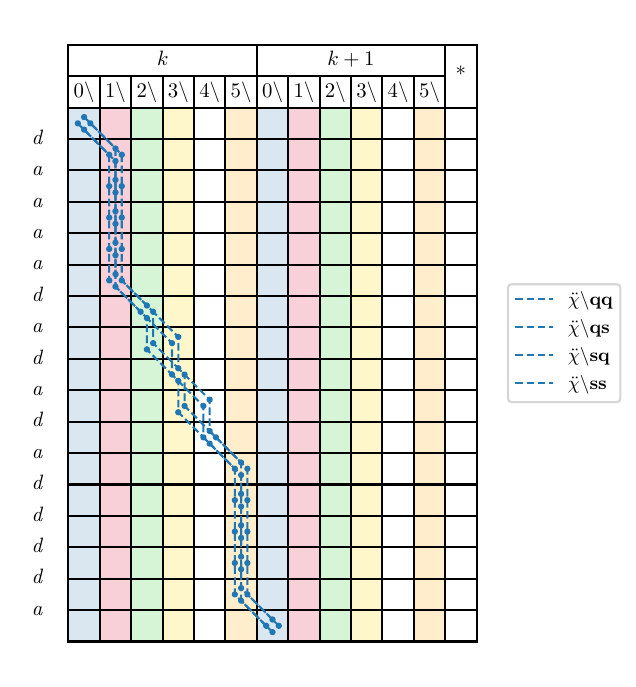}
		\caption{$\gtPass$ from $\bStart \backslash k$}
	\end{subfigure}
	\begin{subfigure}[t]{.3\linewidth}
		\includegraphics[width=\linewidth]{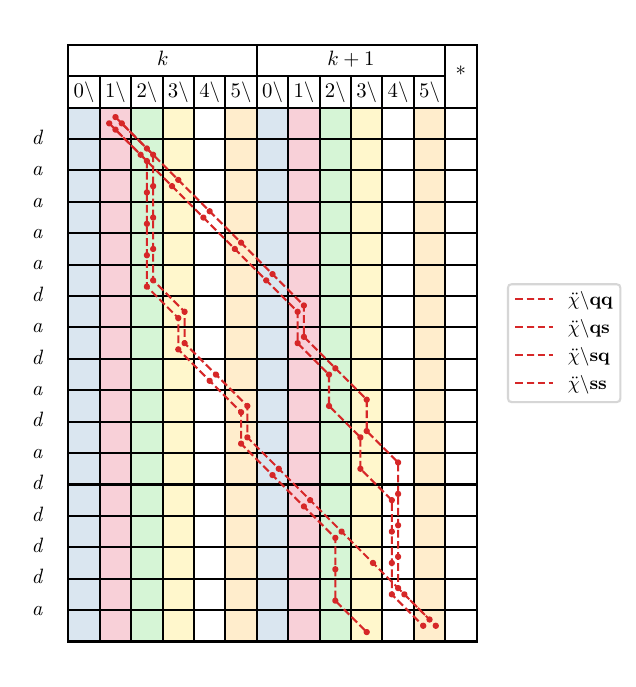}
		\caption{$\gtPass$ from $\bChainDisq \backslash k$}
	\end{subfigure}
	\begin{subfigure}[t]{.3\linewidth}
		\includegraphics[width=\linewidth]{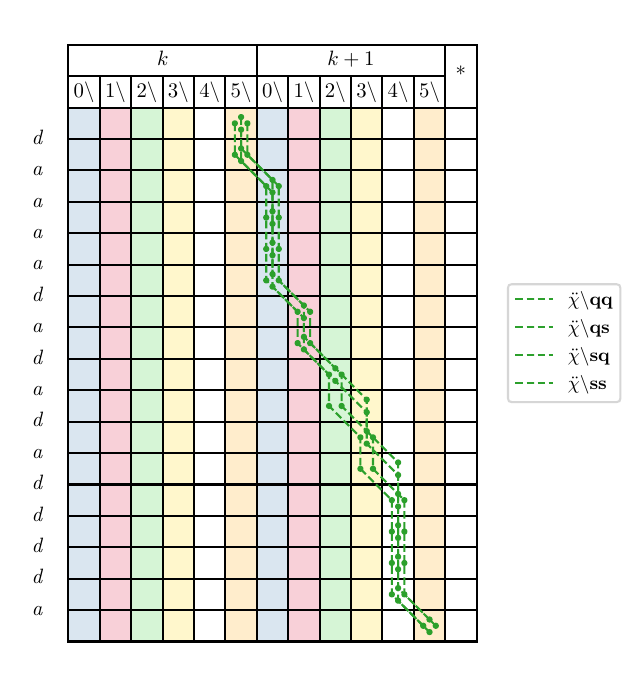}
		\caption{$\gtPass$ from $\bEnd \backslash k$}
	\end{subfigure}
\caption{
Trajectories of $\gtPass$
from three starting positions
on measure $k$
of the chain automata in $\designAlign \backslash 
%\arbpile^{n+1}$.
\arbpile^*$.
}
\label{fig:traj-pass}
\end{figure}

\begin{lemma}
\label{lemma:red4-language}

\makecommand{\profile}{\delta}
Given a CNF formula $\clause$,
let the change profile $\profile$
be formed from $\formSeg_{\Rom{2}}(\clause)$
by the following substitutions:
\begin{itemize}

\item We replace each $\litSeg_i(\clause_j)$ with ${\gtDk}^{k-1} \, \litSeg_i(\phi_j)$;

\item
and each $\lastLitSeg_i(\clause_j)$ with ${\gtDk}^{k-1} \, \lastLitSeg_i(\clause_j) \, {\gtPass}^{k-1}$.

\item We replace each instance of $\gtQ$ with $\gtQ \, {\gtPass}^{k-1}$.

\item (We leave instances of $\gtStart$ and $\gtNext$ alone.)

\end{itemize}
A chain $\pi \in \designAlign^k \backslash \arbpile^{n+2} \backslash \arbpile^m$ accepts $\profile$
if and only if
$\pi = \designAlign^k \backslash \qsym \assignEmbed(\assignvec) \qsym \backslash \qsym^m$ for some $\assignvec \in \SAT(\clause)$.
\end{lemma}
The proof is by a straightforward but rather tedious alteration of our previous series of developments
toward Lemma~\ref{lemma:red2-language}, and we leave it as an (involved) exercise for the reader.

Once a sufficiently large number $k$ of repetitions of $\designAlign$ is chosen
in the first
%\emph{alignment}
round, then
we can augment the second
%\emph{assignment} 
and third 
%\emph{collection} 
rounds as previously described
to obtain three rounds of the same number of piles.
Therefore, 
Lemma~\ref{lemma:red4-language} leads in the usual way to a proof of the following result.
\begin{lemma}
$\probAlias{5}$ is \NPHard/.
\end{lemma}

Conj.~\ref{conj:repeated-hard}
could perhaps be proved from this point
by an approach
similar to that of Section~\ref{sec:variable-round-feasibility-complexity}
to remove the alignment constraint.
However,
the author has not yet discovered any viable combination of these two approaches.
The author suspects that a successful blending of the techniques might require
each type of gadget to incorporate its \emph{own} penalization of unaligned chains,
rather than relying on a penalizing prefix prepended to each clause.

Another approach the author suspects holds promise would be to comprehend a pattern
for the key type sequence $\designAlign$---\eg, some partition into $\designAlign_1 \ldots \designAlign_n$---%
and likewise for each word gadget
of our $\questionSet_\Rom{3}$ reduction,
so that
it may be transformed
into a reduction
with sufficiently large first round capacity
by repeating certain elements of the pattern.

\section{Conclusions}
\label{sec:conclusion}

We have brought some amount of closure
to a study of the pile shuffle and its capabilities as a sorting device---begun in~\cite{treleaven2025sortingpermutationspileshuffle}---%
in the setting where a dealer may use a combination of face-up and face-down piles in each round of shuffle.
%with humble beginnings in the mere act of shuffling of deck of cards on a table---%
%had snowballed into a curiously deep set of computer science questions.
%
Building on top of a mathematical framework developed in~\cite{treleaven2025sortingpermutationspileshuffle},
for the analysis of shuffle in multiple sequential rounds,
we have proved the existence of \NPHard/ sort feasibility problems,
including (notably) variants like Prob.~\ref{prob:variable-round}
which allow arbitrary assignment of pile facings in all rounds of shuffle.

Our results are by a sequence of reductions
from the famously \NPHard/ Boolean satisfiability problem (SAT)
to a series of pile shuffle sort feasibility problems.
They leverage
a novel framework
which equates the certificates of sort feasibility
to members of a particular class of chain-like deterministic finite automata.
The correspondence allows us to reframe questions of sort feasibility 
in terms of the existence of solution-encoding DFAs that accept the encoding of a CNF formula as input.
Our proof strategy 
%of reducing $\SAT$ to such accepting chain problems
is vaguely reminiscent of genetic coding, where
the change profile of a deck permutation (a genetic code)
admits only specific chains (synthesized proteins).

We leave as an open question whether Prob.~\ref{prob:repeated-round}---%
sort feasibility with a uniform pile capacity bound across all rounds of dealer choice shuffle---%
is \NPHard/.
The author's conjecture is that it \emph{is} hard, and we have discussed potential directions of search for proof.

At the end of the day,
pile shuffle sort feasibility remains tractable in the most common homogeneous cases,
and moreover,
the sorting capacity of a given pile capacity bound scales exponentially with the number of rounds of shuffle.

\emph{Acknowledgements}:
The author would like to acknowledge Dr. Joshua Bialkowski for posing a surprisingly challenging question, a long time ago.

\bibliographystyle{unsrt}
\bibliography{main}

\appendix

\clearpage
\section{Background: Decision Problems and \NP/}
\label{sec:complexity}

In computational complexity theory, a \emph{decision problem}
is a set of questions with yes or no answers;
for our purposes, 
a set $\questionSet$ of possible questions, and
an implicit subset $\questionSet^+ \subseteq \questionSet$
of those which are in the affirmative.
An algorithm \emph{decides} problem $\probVar$ if
for every input $\questionVar\in\questionSet$,
it correctly returns as output whether $\questionVar\in\questionSet^+$.

If two problems,
$\probVar_1$ and $\probVar_2$,
are related by a function $f: \questionSet_1 \to \questionSet_2$,
such that
$\questionVar \in \questionSet_1^+ \iff f(\questionVar) \in \questionSet_2^+$,
then $f$ \emph{reduces} $\probVar_1$ to $\probVar_2$.
That is,
we can obtain the answer to any $\questionVar \in \questionSet_1$
by computing and deciding $f(\questionVar)$ instead.
If $f$ is computable in polynomial time---%
\ie, time that is bounded by a polynomial function of the transcription length of the question---%
then
it is a polynomial-time reduction of $\probVar_1$ to $\probVar_2$.

\NPComp/ is an important class of decision problems.
It is one of the ``easiest'' classes of problems that is considered ``hard''.
It is the set all problems that are in both \NP/ (upper bound) and \NPHard/ (lower bound).
\NP/ is the set of all problems admitting proofs, or \emph{certificates}, which can be verified in polynomial time.
\NPHard/ is the set of problems that all problems in \NP/ can be reduced to in polynomial time.
It is well-known therefore that:
\begin{lemma}
If problem $\probVar_1$ is \NPHard/
and
there is a polynomial-time reduction of $\probVar_1$ to $\probVar_2$,
then $\probVar_2$ is \NPHard/ also.
\end{lemma}
%\begin{proof}
%Since $\probVar_1 \in$~\NPHard/,
%by definition,
%every problem $\probVar' \in$~\NP/ has a polynomial-time reduction $f'$ to $\probVar_1$.
%Then $f(f')$ is a polynomial-time reduction of $\probVar'$ to $\probVar_2$.
%\end{proof}
%
\NPHard/ is considered hard in the sense that
there is currently no known algorithm for any \NPHard/ problem that produces its answer in polynomial time.

\section{Derivation of the ``virtual shuffle'' (Lemma~\ref{prop:hetero-multi-sort})}
\label{sec:multi-round-derivation}

\makecommand{\index}{\phase}
%\makecommand{\itOne}{\designpile_\index}
%\makecommand{\itTwo}{\vrpile_{\index+1}}
\makecommand{\itOne}{r}
\makecommand{\itTwo}{q}

We summarize the derivation of equations for Lemma~\ref{prop:hetero-multi-sort},
originally presented in~\cite{treleaven2025sortingpermutationspileshuffle}.

Suppose a multi-round pile shuffle $(\typeSchedules, \schedules)$ is in $\numphase \geq 0$ rounds,
with
type schedule $\typeSchedules = \left( \pileType_1, \ldots, \pileType_\numphase \right)$ and
pile assignments
$\schedules = \left( \pile_1, \ldots, \pile_\numphase \right)$.
Let $\numpile_\phase$ denote,
for each $\phase \in [\numphase]$,
the number of piles available  in round $\phase$.
We do not assume the number of piles is the same in each round, however,
we do assume w.l.g. that
each $\pile_\phase$ takes value on the range $[\numpile_\phase]-1$.

Let
$\reverseFn_m(y) \dfneq -y + m - 1$ denote
the interval-reversing function, mapping the interval $[\numpile]-1$ to itself in reverse order, and
let $\reverseFn_m^n$ denote the composition of $\reverseFn_m$ $n$ times.
(Although it has a binary orbit.)
%
%it is easy to check $\reverseFn_m^{2k}$ is identity for all $k \geq 0$, and
%$\reverseFn_m^{2k+1} = \reverseFn_m$.

%
For each $\phase \in [\numphase]$,
let
\begin{equation}
\label{eq:pile-type-indicator}
\pileTypeInd_\phase(\pile)
\dfneq \left[ \pileType_\phase(\pile) = \ssym \right]
= \begin{cases}
0	& \pileType_\phase(\pile) = \qsym
\\
1	& \pileType_\phase(\pile) = \ssym
,
\end{cases}
\end{equation}
and
let $\{ \vrPileTypeFn_\phase \}_{\phase=1}^{\numphase + 1}$
be defined
by the backward recurrence starting from
$\vrPileTypeFn_{\numphase + 1} = 0$
with
\begin{varscope}
\makecommand{\realParity}{ \vrPileTypeFn_{\index+1}(\itTwo) }
\begin{align}
\label{eq:type-indicator-recurrence}
\vrPileTypeFn_{\index}(
	\itOne
	+
	\numpile_\index \itTwo
)
&=
\pileTypeFn_\index\left(
	\reverseFn_{\numpile_\index}^{\realParity}
	(\itOne)
\right)
%, \qquad 1 \leq \phase < \numphase - 1
\xor
\realParity
,
\qquad
\itOne \in [\numpile_\index] - 1
, \,
\itTwo \in [\vrNumPile_{\phase+1}] - 1
,
\end{align}
\end{varscope}%
where
$$ %\[
\vrNumPile_\phase \dfneq \prod_{\phase \leq \phase' \leq \numphase} \numpile_{\phase'}
,
\qquad 1 \leq \phase \leq \numphase + 1
;
$$
the operator $\xor$ denotes modulo-$2$ addition.
Note that
$\vrNumPile_{\numphase + 1} = 1$,
and
\eqref{eq:type-indicator-recurrence} defines
$\vrPileTypeFn_\phase$
over the co-domain $[\vrNumPile_\phase] - 1$
for each $\phase \in [\numphase]$.

The types assignment $\vrPileType_1$
of our virtual shuffle
is the assignment of types to $\vrNumPile_1$ virtual piles which satisfies
$$
\vrPileTypeFn_1(\pile)
= \left[ \vrPileType_1(\pile) = \ssym \right]
,
\qquad
\pile \in [\vrNumPile_1] - 1
.
$$
$\vrpile_1$ in turn falls out of the following reversible backward recurrence:
\begin{align}
\begin{cases}
\vrpile_\numphase = \pile_\numphase,
&
\\
\designpile_\phase(s)
=
\reverseFn_{\numpile_\phase}^{
	\vrPileTypeFn_{\phase+1}(\vrpile_{\phase+1}(s))
}(\pile_\phase(s))
,
&
1 \leq \phase \leq \numphase - 1, \, s \in [n],
\\
\vrpile_\phase
=
	\designpile_\phase
	+
	\numpile_\phase \vrpile_{\phase+1}
	,
&
1 \leq \phase \leq \numphase-1
.
\end{cases}
\label{eq:hetero-multi-embedding}
\end{align}
Each $\vrpile_\phase$ takes value on the range $[\vrNumPile_\phase]-1$.

\section{Trajectory plots from Section~\ref{sec:gadgets}}
\label{sec:more-traj}

\begin{figure}[h!]
	\makecommand{\theFrac}{.48}
	\centering
	
	\begin{subfigure}[t]{\theFrac\linewidth}
		\includegraphics[width=\linewidth]{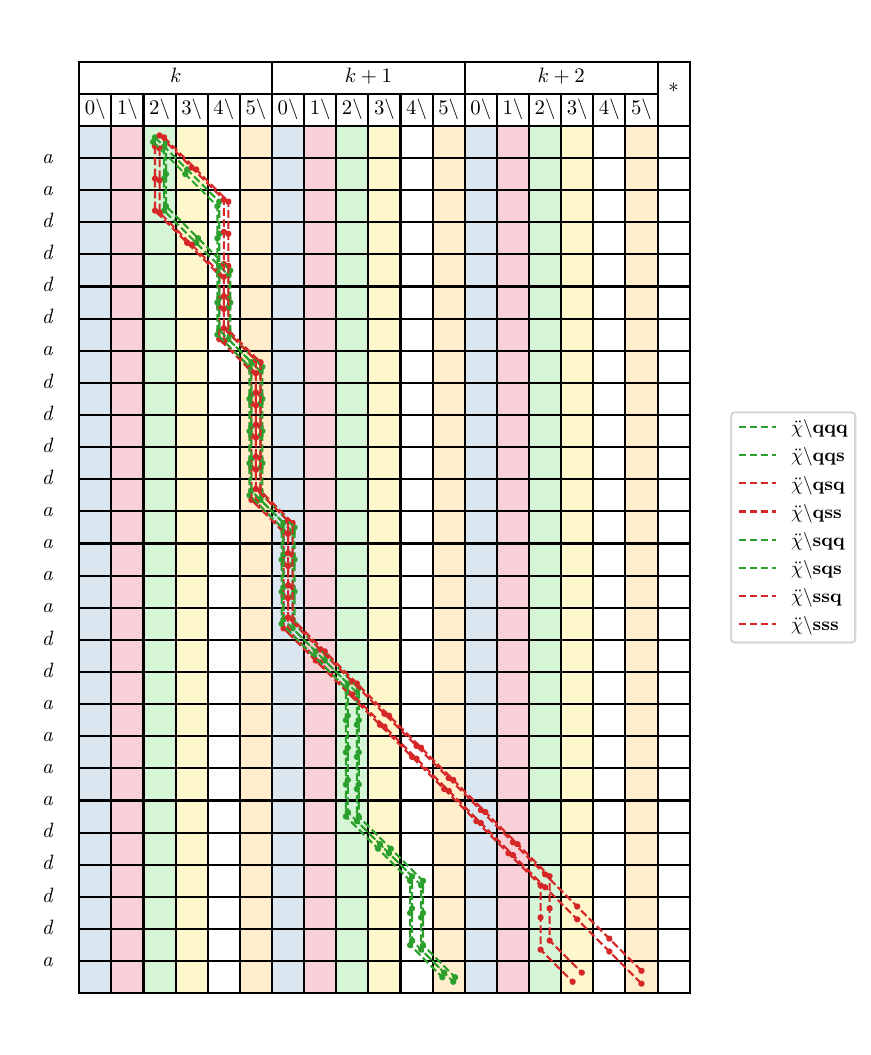}
		\caption{$\gtLastPos$ starting activated}
	\end{subfigure}
	\begin{subfigure}[t]{\theFrac\linewidth}
		\includegraphics[width=\linewidth]{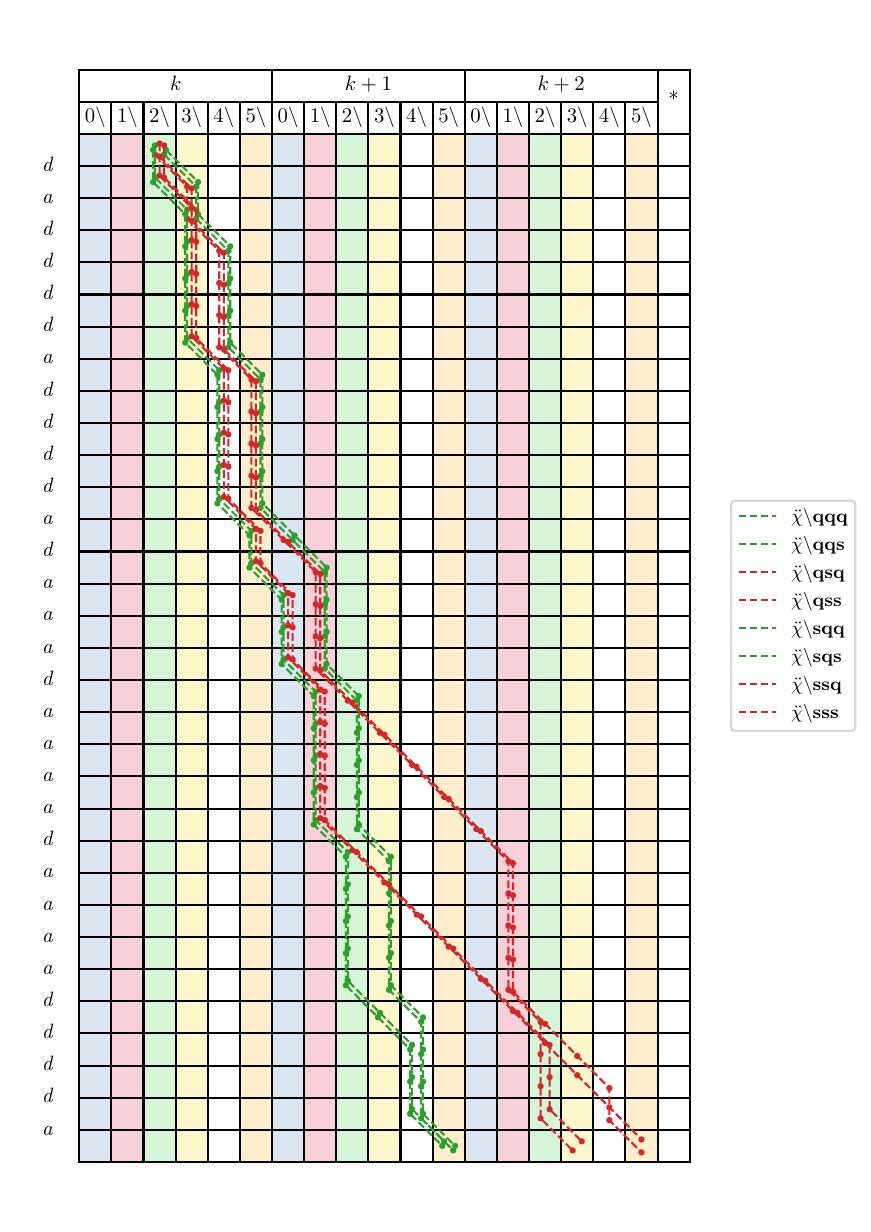}
		\caption{$\gtLastNeg$ starting activated}
	\end{subfigure}
	\begin{subfigure}[t]{\theFrac\linewidth}
		\includegraphics[width=\linewidth]{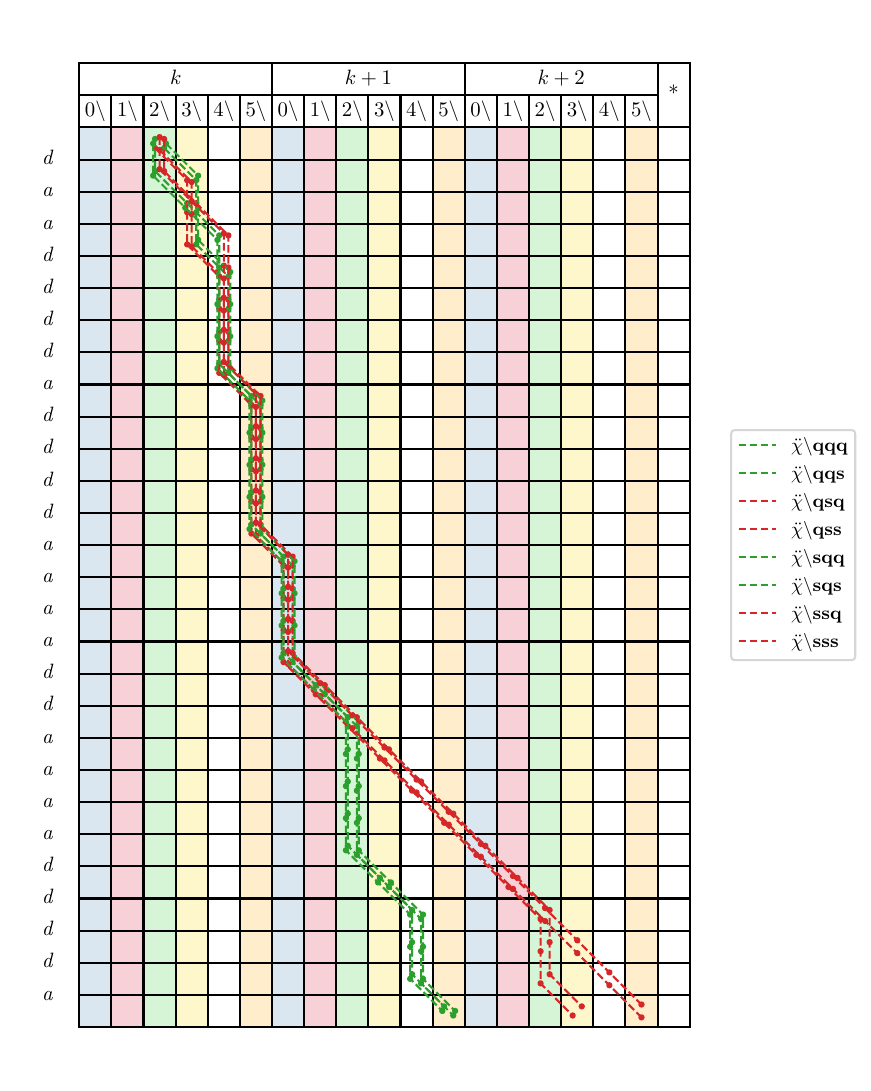}
		\caption{$\gtLastDk$ starting activated}
	\end{subfigure}

\caption{Trajectories of $\gtLastPos$, $\gtLastNeg$, and $\gtLastDk$
starting from an activated state
on the $k$-th measure
of the chain automata in $\designAlign \backslash \arbpile^*$.}
\label{fig:endtest-actd}
\end{figure}

\begin{figure}[h!]
	\makecommand{\theFrac}{.48}
	\centering

	\begin{subfigure}[t]{\theFrac\linewidth}
		\includegraphics[width=\linewidth]{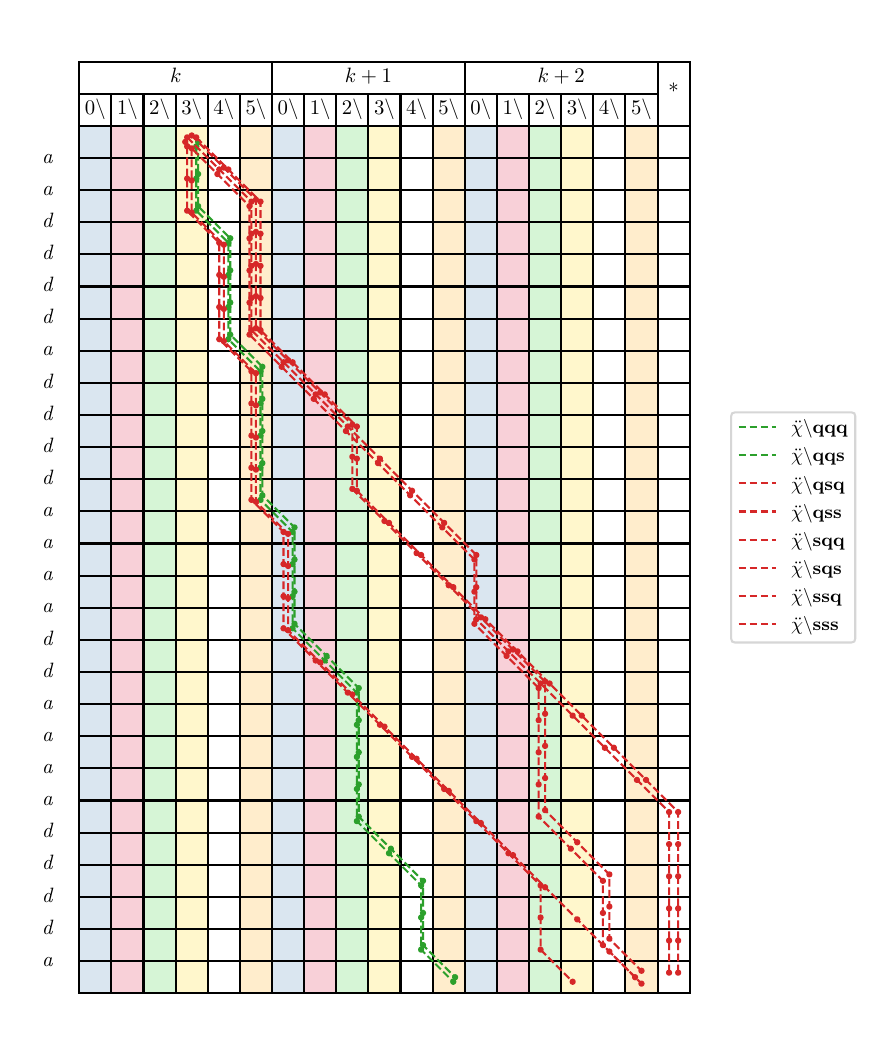}
		\caption{$\gtLastPos$ starting not activated}
	\end{subfigure}
	\begin{subfigure}[t]{\theFrac\linewidth}
		\includegraphics[width=\linewidth]{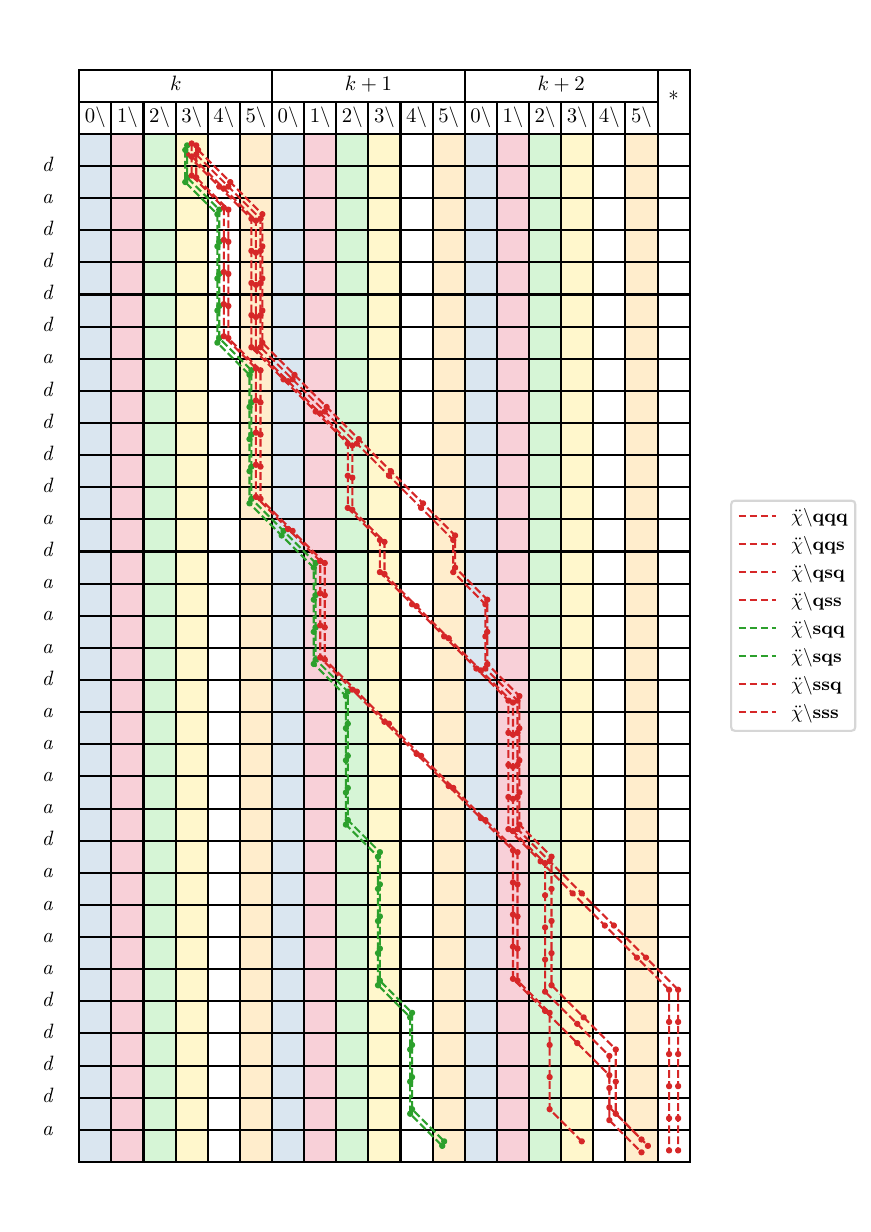}
		\caption{$\gtLastNeg$ starting not activated}
	\end{subfigure}
	\begin{subfigure}[t]{\theFrac\linewidth}
		\includegraphics[width=\linewidth]{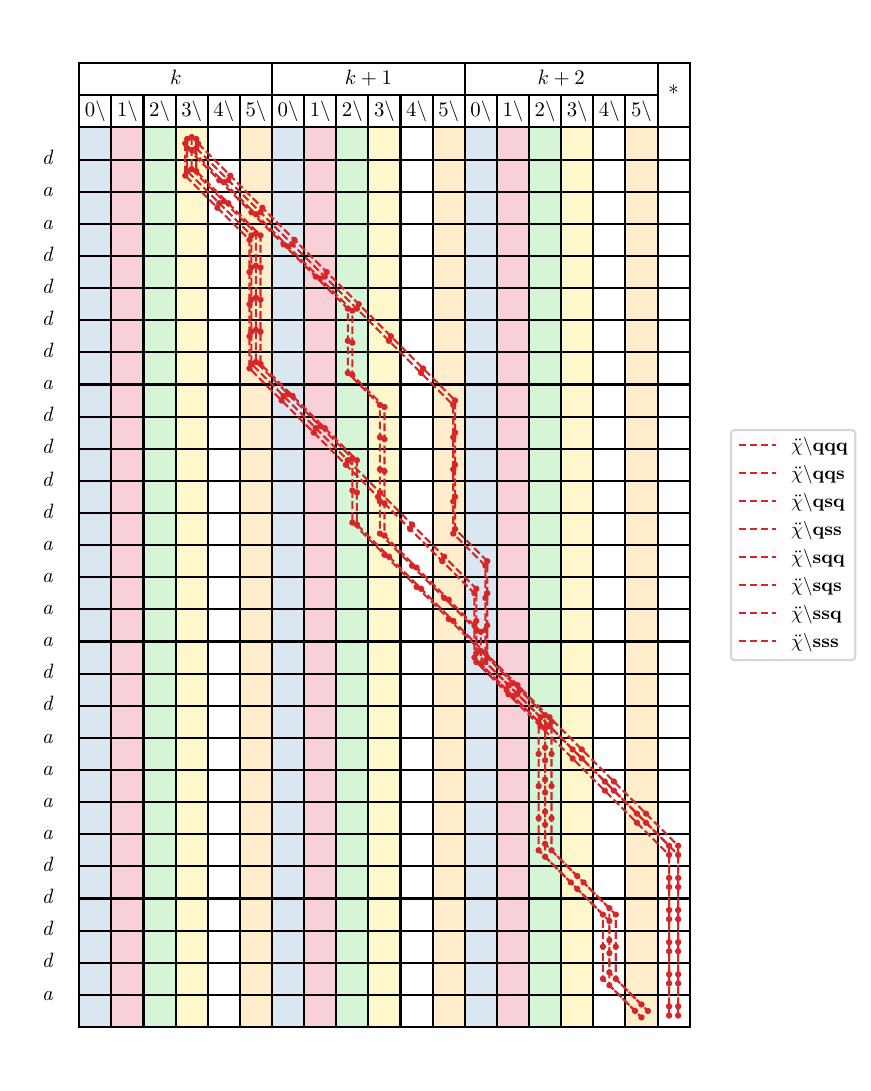}
		\caption{$\gtLastDk$ starting not activated}
	\end{subfigure}

\caption{Trajectories of $\gtLastPos$, $\gtLastNeg$, and $\gtLastDk$
starting from a not-activated state
on the $k$-th measure
of the chain automata in $\designAlign \backslash \arbpile^*$.}
\label{fig:endtest-nactd}
\end{figure}

\begin{figure}[h!]
	\makecommand{\theFrac}{.48}
	\centering

	\begin{subfigure}[t]{\theFrac\linewidth}
		\includegraphics[width=\linewidth]{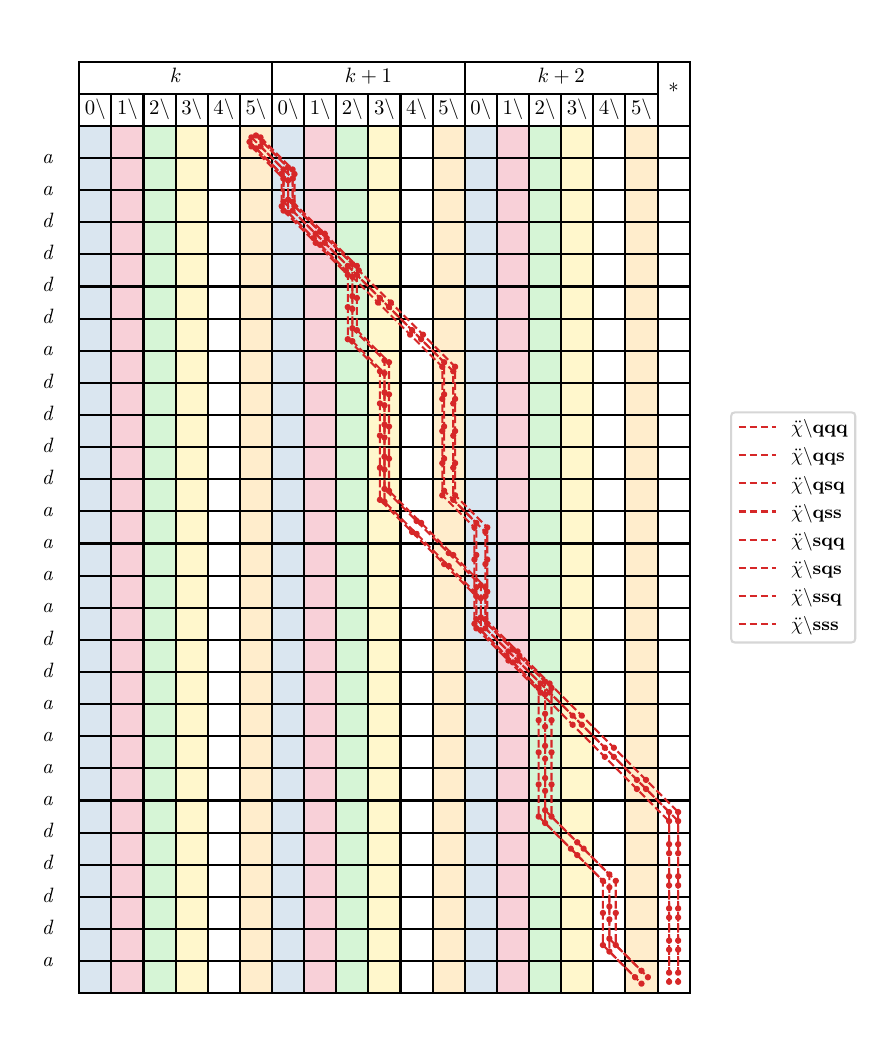}
		\caption{$\gtLastPos$ starting disqualified}
	\end{subfigure}
	\begin{subfigure}[t]{\theFrac\linewidth}
		\includegraphics[width=\linewidth]{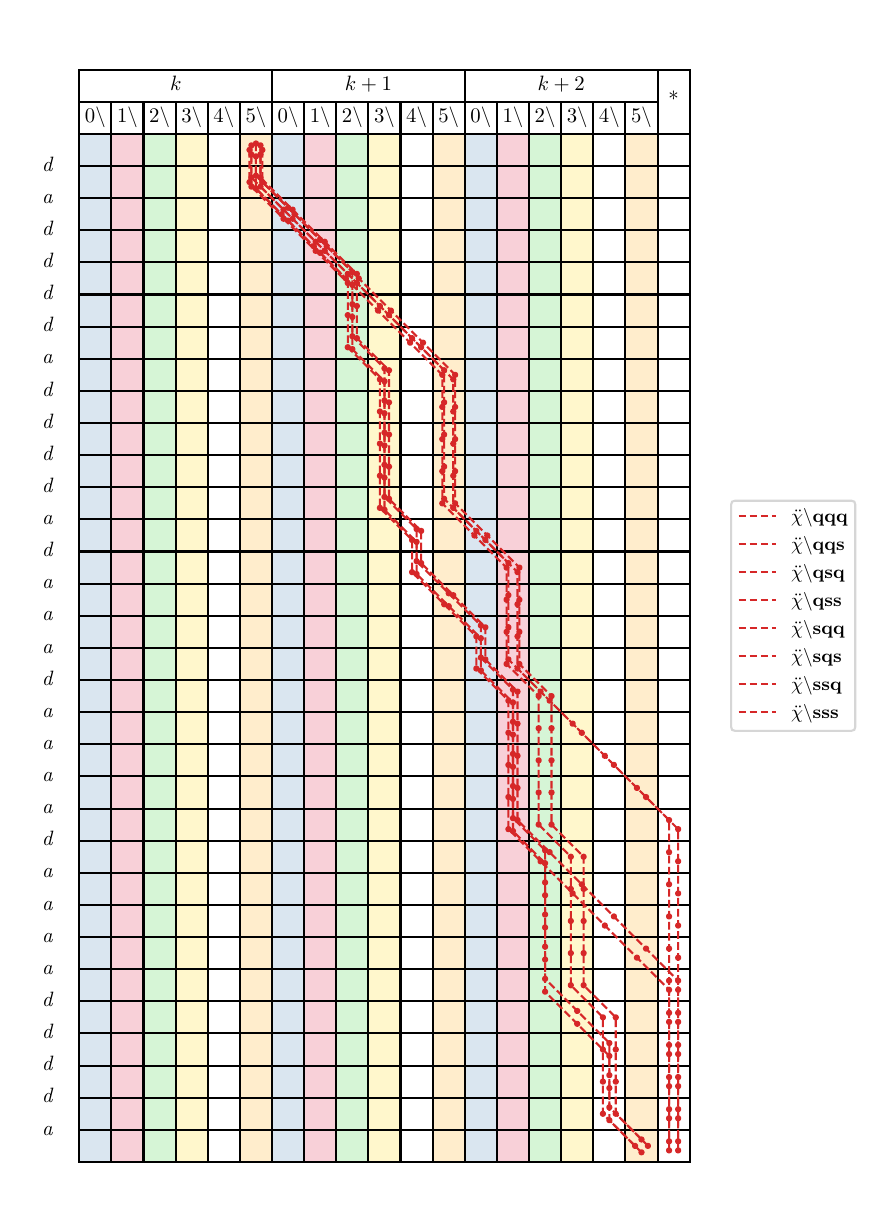}
		\caption{$\gtLastNeg$ starting disqualified}
	\end{subfigure}
	\begin{subfigure}[t]{\theFrac\linewidth}
		\includegraphics[width=\linewidth]{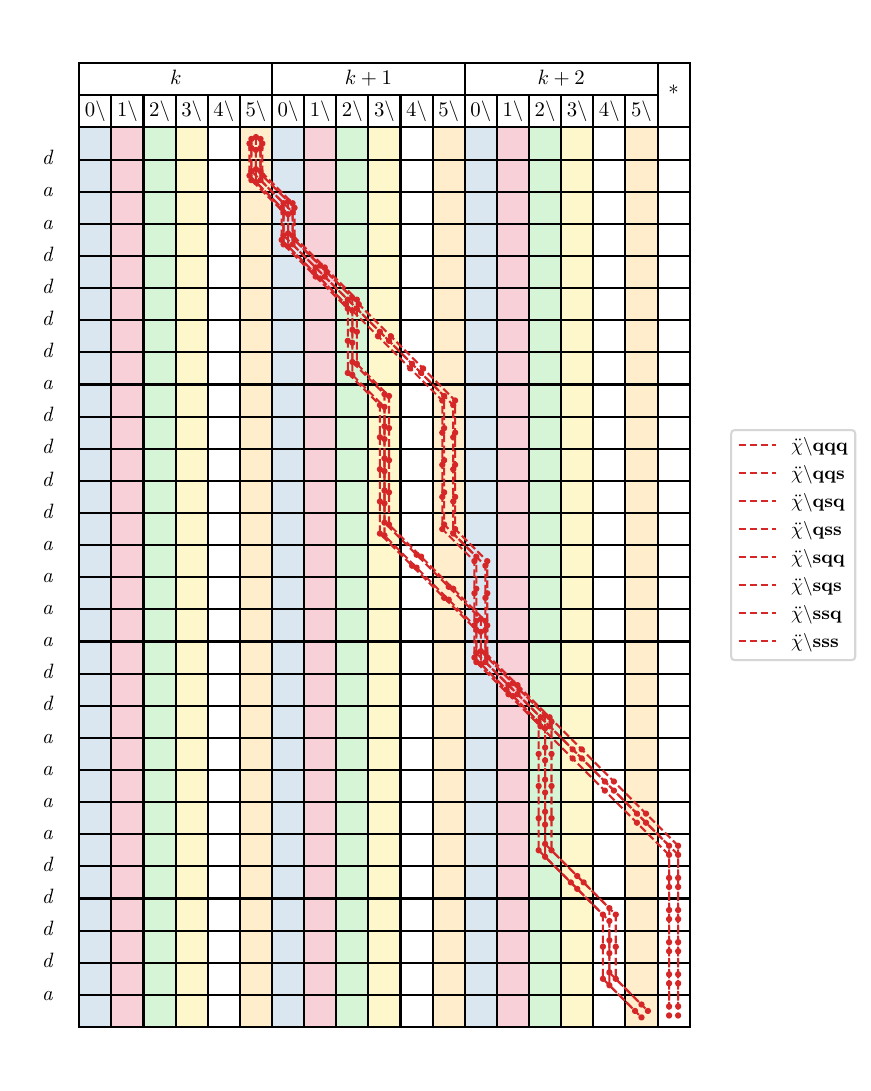}
		\caption{$\gtLastDk$ starting disqualified}
	\end{subfigure}

\caption{Trajectories of $\gtLastPos$, $\gtLastNeg$, and $\gtLastDk$
starting from a disqualified state
on the $k$-th measure
of the chain automata in $\designAlign \backslash \arbpile^*$.}
\label{fig:endtest-disq}
\end{figure}

%\clearpage
%\listoftodos
%\todototoc

\end{document}